%% file: GATE.tex
\documentclass[11pt,a4]{article}

\usepackage{amsmath,amsfonts,amssymb}
\usepackage{graphicx,psfrag,epsf,pslatex}
\usepackage{enumerate,tikz}
\usepackage{natbib}
\usepackage{url} % not crucial - just used below for the URL
\usetikzlibrary{plotmarks}

\usepackage{caption}
\usepackage{subcaption}
\usepackage{thmtools, thm-restate}
\usepackage{subfloat}
\usepackage{float}
\usepackage{lipsum}

\restylefloat{table}

\usepackage{xcolor}
\usepackage[draft]{todonotes}
\usepackage{lineno,hyperref}
\usepackage{ulem}
\input{header.tex}

\input{mathsymbol.tex}

\usepackage{authblk}

\usepackage{algpseudocode}
\usepackage{algorithm}% http://ctan.org/pkg/algorithms

\numberwithin{equation}{section}
\theoremstyle{plain}

\numberwithin{equation}{section}

\makeatletter
\renewcommand\paragraph{\@startsection{paragraph}{4}{\z@}%
            {-2.5ex\@plus -1ex \@minus -.25ex}%
            {1.25ex \@plus .25ex}%
            {\normalfont\normalsize\bfseries}}
\makeatother

\title{Local False Discovery Rate Based Methods for Multiple Testing of One-Way Classified Hypotheses}

\author[1]{Sanat K. Sarkar\thanks{sanat@temple.edu}}
\author[2]{Zhigen Zhao\thanks{zhaozhg@temple.edu}}
\affil[1,2]{Department of Statistics, Operations, and Data Science, Temple University, Philadelphia, PA, 19122, USA}
\begin{document}
\maketitle

\begin{abstract}
  This paper continues the line of research initiated in \cite{Liu:Sarkar:Zhao:2016} on developing a novel framework for multiple testing of hypotheses grouped in a one-way classified form using hypothesis-specific local false discovery rates (Lfdr's). It is built on an extension of the standard two-class mixture model from single to multiple groups, defining hypothesis-specific Lfdr as a function of the conditional Lfdr for the hypothesis given that it is within an important group and the Lfdr for the group itself and involving a new parameter that measures grouping effect. This definition captures the underlying group structure for the hypotheses belonging to a group more effectively than the standard two-class mixture model. Two new Lfdr based methods, possessing meaningful optimalities, are produced in their oracle forms. One,  designed to control false discoveries across the entire collection of hypotheses, is proposed as a powerful alternative to simply pooling all the hypotheses into a single group and using commonly used Lfdr based method under the standard single-group two-class mixture model. The other is proposed as an Lfdr analog of the method of \cite{Benjamini:Bogomolov:2014} for selective inference. It controls Lfdr based measure of false discoveries associated with selecting groups concurrently with controlling the average of within-group false discovery proportions across the selected groups. Simulation studies and real-data application show that our proposed methods are often more powerful than their relevant competitors.
\end{abstract}

{\bf Keywords: } False Discovery Rate, Grouped Hypotheses, Large-Scale Multiple Testing

\section{Introduction}\label{sec:introduction}

Modern scientific studies aided by high-throughput technologies, such as those related to brain imaging, microarray analysis, astronomy, atmospheric science, drug discovery, and many others, are increasingly relying on large-scale multiple testing as an integral part of statistical investigations focused on high-dimensional inference. With many of these investigations, notably in genome-wide association and neuroimaging studies, giving rise to testing of hypotheses that appear in groups, the multiple testing paradigm seems to be shifting from testing single to multiple groups of hypotheses. These groups, forming at single or multiple levels creating respectively one- or multiway classified hypotheses, can occur naturally due to the underlying biological or experimental
process or be created using internal or external information capturing certain specific features
of the data. Several newer questions arise with this paradigm shift. We will focus in this paper on the following two questions related to one-way classified hypotheses that seem relevant in light of what is available in the literature:
\begin{enumerate}
\item[Q1.] For multiple testing of hypotheses grouped into a one-way classified form, how to effectively capture the underlying group/classification structure, instead of simply pooling all the hypotheses into a single group, while controlling overall false discoveries across all individual hypotheses?
\item[Q2.] For hypotheses grouped into a one-way classified form in the context of post-selective inference where groups are selected before testing the hypotheses in the selected groups, how to effectively capture the underlying group/classification structure to control the expected average of false discovery proportions across the selected groups?
\end{enumerate}

Progress has been made toward answering Q1 (\cite{Hu:Zhao:Zhou:2010, Nandi:Sarkar:2018}) and Q2 (\cite{Benjamini:Bogomolov:2014}) for one-way classified hypotheses in the framework of Benjamini-Hochberg (BH, \cite{Benjamini:Hochberg:1995}) type false discovery rate (FDR) control. However, research addressing these questions based on local false discovery rate (Lfdr) (\cite{Efron:Tibshirani:Storey:Tusher:2001}) based methodologies are largely absent, excepting the work of \cite{Cai:Sun:2009} and a recent work of \cite{Liu:Sarkar:Zhao:2016} where a method has been proposed in its oracle form to answer the following question slightly different from Q1: When making important discoveries within each group is as important as making those discoveries across all hypotheses, how to maintain a control over falsely discovered hypotheses within each group while controlling it across all hypotheses?

The fact that an Lfdr based approach with its Bayesian/empirical Bayesian and decision theoretic foundation can yield powerful multiple testing method controlling false discoveries effectively capturing dependence as well as other structures of the data in single- and multiple-group settings has been demonstrated before  (\cite{Sun:etal:2006, Sun:Cai:2007, Efron:2008, Ferkingstad:2008, Sarkar:Zhou:Ghosh:2008, Sun:Cai:2009, Cai:Sun:2009, Zhao:2010, Hu:Zhao:Zhou:2010,  Zhao:Hwang:2012, Zablocki:etal:2014, Zhao:Sarkar:2015, Ignatiadis:etal:2016, kwonf}). However, the work of \cite{Liu:Sarkar:Zhao:2016} is fundamentally different from these works in that it takes into account the sparsity of signals both across groups and within each active group. Consequently, the effect of a group's importance in terms of its Lfdr can be explicitly factored into an Lfdr based importance measure of each hypothesis within that group.

In this article, we continue the line of research initiated in \cite{Liu:Sarkar:Zhao:2016} to answer Q1 in an Lfdr framework. More specifically, we borrow ideas from \cite{Liu:Sarkar:Zhao:2016} in developing methodological steps to present a unified group-adjusted multiple testing framework for one-way classified hypotheses that introduces a grouping effect into overall false discoveries across all individual hypotheses. This new insight clearly demonstrates how the group effect influences the existing testing method. We also provide an answer to Q2 from the Bayes/empirical Bayes perspective. To the best of our knowledge, this is the first result to study the FDR control for the selection from the Bayesian/empirical Bayesian perspective.

The paper is organized as follows. In Section \ref{sec:literature}, we present the current state of knowledge closely related to the present work, before presenting our Lfdr based methods answering Q1 and Q2 in Section \ref{sec:method}. In Section \ref{sec:numerical}, we introduce a Gibbs sampler using a hierarchical Bayes approach to estimate the model parameters and demonstrate the performances of the proposed methods based on these estimates and compare it to many existing methods. We further applied the proposed method answering Q1 to the Adequate Year Progress (AYP) data set [\cite{Liu:Sarkar:Zhao:2016}]. We conclude the paper with a few remarks in Section \ref{sec:concluding}.

\section{Literature Review}\label{sec:literature}
Suppose there are $N$ hypotheses that appear in $m$ non-overlapping families/groups, with $H_{ij}$ being the $j$th hypothesis in the $i$th group ($i=1, \ldots, m; j=1, \ldots, n_i$). We refer to such a layout of hypotheses as one-way classified hypotheses.

Let $X_{ij}$ be the test statistic/$p$-value associated with $H_{ij}$, and $\theta_{ij}$ be the binary parameter indicating the truth ($\theta_{ij} = 0$) or falsity ($\theta_{ij} = 1$) of $H_{ij}$. The Lfdr corresponding to $H_{ij}$, defined by the posterior probability $P(\theta_{ij}=0|\vX)$, where $\vX = \{X_{ij}, i=1, \ldots, m; j=1, \ldots, n_i\}$, is the basic ingredient for constructing Lfdr based approaches controlling false discoveries. The single-group case (or the case ignoring the group structure) has been considered extensively in the literature, notably by \cite{Sun:Cai:2007, Cai:Sun:2009} and \cite{He:Sarkar:Zhao:2015} who focused on constructing methods that are optimal, at least in their oracle forms. These oracle methods correspond to Bayes multiple decision rules under a single-group two-class mixture model (\cite{Efron:Tibshirani:Storey:Tusher:2001, Newton:etal:2004, Storey:2002, zhao:2022}) that minimize marginal false non-discovery rate (mFNR), a measure of false non-discoveries closely related to the notion of false non-discoveries (FNR) introduced in \cite{Genovese:Wasserman:2002} and \cite{Sarkar:2004}, subject to controlling marginal false discovery rate (mFDR), a measure of false discoveries closely related to the BH FDR and the positive FDR (pFDR) of \cite{Storey:2002}. Multiple-group versions of single-group Lfdr based approaches to multiple testing have started getting attention recently, among them the following seem more relevant to our work.

\cite{Cai:Sun:2009} extended their work from single to multiple groups (one-way classified hypotheses) under the following model: with $i$ taking the value $k$ with some prior probability $\pi_k$, $(X_{ij},\theta_{ij})$, $j=1, \ldots, n_i$, given $i=k$, are assumed to be iid random pairs with
\[
X_{kj}|\theta_{kj} \sim (1-\theta_{kj})f_{k0} + \theta_{kj} f_{k1},
\]
for some given densities $f_{k0}$ and $f_{k1}$, and $\theta_{kj} \sim Bernoulli(p_k)$.
They developed a method, which in its oracle form minimizes mFNR subject to controlling mFDR and is defined in terms of thresholding the conditional Lfdr's: CLfdr$^i(X_{ij}) = (1-p_{i})f_{i0}(X_{ij})/f_i(X_{ij})$, where $ f_i(X_{ij}) = (1- p_i)f_{i0}(X_{ij}) + p_i f_{i1}(X_{ij})$, for $j=1, \ldots, n_i$, $i=1, \ldots, m$, before proposing a data-driven version of the oracle method that asymptoticaly maintains the original oracle properties. It should be noted that the probability $p_k$ relates to the size of group $k$ and provides little information about the importance of the group itself. \cite{Ferkingstad:2008} brought the grouped hypotheses setting into testing a single family of hypotheses in an attempt to empower typical Lfdr based thresholding approach by leveraging an external covariate. They partitioned the $p$-values into a number of small bins (groups) according to ordered values of the covariate. With the underlying two-class mixture model defined separately for each bin depending on the corresponding value of the covariate, they defined the so called covariate-modulated Lfdr as the posterior probability of a null hypothesis given the value of the covariate for the corresponding bin. They estimated the covariate-modulated Lfdr in each bin using a Bayesian approach before proposing their thresholding method, not necessarily controlling an overall measure of false discoveries such as the mFDR or the posterior FDR. An extension of this work from single to multiple covariates can be seen in \cite{Zablocki:etal:2014, Scott:etal:2015}. Recently, \cite{Cai:Sun:Wang:2019} developed a novel grouped hypotheses testing framework for two-sample multiple testing of the differences between two highly sparsed mean vectors, having constructed the groups to extract sparsity information in the data by using a carefully constructed auxiliary covariate. They proposed an Lfdr based optimal multiple testing procedure controlling FDR as a powerful alternative to standard  procedures based on the sample mean differences.

A sudden upsurge of research has taken place recently in selective/post-selection inference due to its importance in light of the realization by the scientific community that the lack of reproducibility of a scientist's work is often caused by his/her failure to account for selection bias. When multiple hypotheses are simultaneously tested in a selective inference setting, it gives rise to a grouped hypotheses testing framework with the tested groups being selected from a given set of groups of hypotheses. \cite {Benjamini:Bogomolov:2014} introduced the notion of the expected average of false discovery proportion across the selected groups as an  appropriate error rate to control in this setting and proposed a method that controls it. Since then, a few papers have been written in this area (\cite{Peterson:etal:2016} and \cite{Heller:Chatterjee:Krieger:Shi:2017}); however, no research has been produced yet in Lfdr framework.

\begin{remark}\label{remark:1}
\rm When grouping of hypotheses occurs, an assumption can be made that the importance of a hypothesis is influenced by that of the group it belongs to. The Lfdr under the standard two-class mixture model, however, does not help in assessing a group's influence on {\it true} importance of its hypotheses. This has been the main motivation behind the work of \cite{Liu:Sarkar:Zhao:2016}, who considered a group-adjusted two-class mixture model that yields an explicit representation of each hypothesis-specific Lfdr as a composition of its group-adjusted form and the Lfdr for the group it is associated with. It allowed them to produce a method that provides a separate control over within-group false discoveries for truly important groups in addition to having a control of false discoveries across all individual hypotheses. This paper motivates us to proceed further with the development of newer Lfdr based multiple testing methods for one-way classified hypotheses as described in the following section.
\end{remark}

\section{Proposed Methods}\label{sec:method}

This section presents the two Lfdr based methods we propose in this article to answer Q1 and Q2. The development of the methods takes place under the model introduced in \cite{Liu:Sarkar:Zhao:2016}, which extends the standard two-class mixture model (Efron et al., 2001) from single to multiple groups. For completeness, we will recall this model here, with a different name, along with the formulas for different types of Lfdr associated with it before developing the methods.

\subsection{\it Model and Lfdr Formulas}
Following are the two basic ingredients in building the aforementioned model in \cite{Liu:Sarkar:Zhao:2016}: (i) expressing each $\theta_{ij}$ as $\theta_{ij} = \theta_{i\centerdot}\cdot \theta_{j|i}$, with $\theta_{i \centerdot} = I(\sum_{j=1}^{n_i} \theta_{ij} > 0)$ indicating the truth $(\theta_{i\centerdot}=0)$ or falsity $(\theta_{i\centerdot}=1)$ of $H_{i\centerdot} = \cap_{j=1}^{n_i} H_{ij}$, and $\theta_{j|i} =\theta_{ij}|\theta_{i\cdot}=1
$, which reflects the underlying group structure of the hypotheses; and (ii) use of the following distribution for the $\theta_{j|i}$'s given $\theta_{i \centerdot} = 1$, as an adjustment of the product Bernoulli distribution for the binary states of a set of hypotheses given that the group they belong to is important:
\begin{definition}\label{def:trunc:bernoulli}
\noindent {\rm [Truncated Product Bernoulli (TPBern ($\pi, n$)).]} {\it A set of $n$ binary variables $Z_1, \ldots, Z_n$ with the following joint probability distribution is said to have a TPBern ($\pi, n$) distribution:
\begin {eqnarray}
  P(Z_1=z_1, \ldots, Z_n=z_n) & =& \frac {1}{1-(1-\pi)^{n}} \prod_{i=1}^n \left \{ \pi^{z_i}(1-\pi)^{1-z_i} \right \} I\left (\sum_{i=1}^n z_i >0 \right ) \nonumber \\ & = & \frac {(1-\pi)^n}{1-(1-\pi)^{n}} \left ( \frac{\pi}{1-\pi}\right )^{\sum_{i=1}^n z_i} I\left (\sum_{i=1}^n z_i >0 \right ).  \nonumber
                                                                                                                                                                              \end {eqnarray}
                                                                                                                                                                            }
\end{definition}

The model is stated in the following:

\begin{definition}\label{def:model:1}
  \noindent {\rm [Group-Adjusted Two-Class Mixture Model for One-Way Classified Hypotheses (One-Way GAMM)]}. {\it %Let $(X_{ij}, \theta_{i\centerdot}, \theta_{j|i},j=1, \ldots, n_i, \theta_{i\centerdot}, \theta_{j|i}, j=1, \ldots, n_i)$ be the set of random variables associated with group $i$, for $i=1, \ldots, m$.
For each $i = 1, \ldots, m$,
\label{def:model1}
%&\theta_{gj}=\theta_{i}*\theta_{j|i},\notag\\
\begin{eqnarray*}%%\label{def:model1}
  \left\{ \begin{array}{l}  X_{ij}~|~\theta_{i\centerdot},\theta_{j|i}  \stackrel {ind} \sim (1-\theta_{i \centerdot} \cdot \theta_{j|i})f_{0}(x_{ij})+\theta_{i \centerdot} \cdot \theta_{j|i}f_{1}(x_{ij}), \textrm{ for some densities $f_{0}$ and $f_{1}$,}\\
    P(\theta_{j|i} =0|~\theta_{i\centerdot} = 0) = 1, \textrm{ for $j=1, \ldots, n_i$};\\
    	(\theta_{1|i}, \ldots, \theta_{n_i|i})~|~\theta_{i\centerdot} =1  \sim TPBern(\pi_{2i}, n_i),\\
    \theta_{i\centerdot} \sim Bern(\pi_{1}).
  \end{array}\right.
\end{eqnarray*}}
The sets $(X_{ij}, \theta_{i\centerdot}, \theta_{j|i}, j=1, \ldots, n_i)$, for $i = 1, \ldots, m$, are mutually independent.
\end{definition}

Let
\[
  {\rm Lfdr}_{ij}(\pi_{1}, \pi_{2i}) \equiv {\rm Lfdr}_{ij}(\vx;\pi_{1}, \pi_{2i}) = Pr (\theta_{ij}=0~|~\vX=\vx ),
  \]
  \[
    {\rm Lfdr}_{i\centerdot} (\pi_{1}, \pi_{2i}) \equiv {\rm Lfdr}_{i\centerdot} (\vx; \pi_{1}, \pi_{2i}) = Pr (\theta_{i \centerdot}=0~|~\vX=\vx ),
    \]
    and
    \[
      {\rm Lfdr}_{j|i}(\pi_{1},\pi_{2i}) \equiv {\rm Lfdr}_{j|i}(\vx;\pi_{1},\pi_{2i})= Pr (\theta_{j|i}=0~|~\theta_{i \centerdot} = 1,\vX=\vx),
    \]
respectively, be the local FDRs under One-Way GAMM corresponding to $H_{ij}$, $H_{i \centerdot}$, and $H_{ij}$ conditional on $H_{i\centerdot}$ being  false. It is easy to see that
\begin{eqnarray}\label{eqn:decomposition}
  {\rm Lfdr}_{ij}(\pi_{1}, \pi_{2i})  & = & 1 - [1 - {\rm Lfdr}_{i\centerdot} (\pi_{1},\pi_{2i})] [1- {\rm Lfdr}_{j|i}(\pi_{1},\pi_{2i})],
\end{eqnarray}
showing how a hypothesis specific local FDR factors into the local FDR for the group and that for the hypothesis conditional on the group being important.

Let ${\rm Lfdr}^*_{ij}(\pi_{2i}) = [(1-\pi_{2i})f_0(x_{ij})]/m_i(x_{ij})$, with $m_i(x) = (1-\pi_{2i})f_0(x) + \pi_{2i} f_1(x)$, be the local FDR corresponding to $H_{ij}$ under the standard two-class mixture model. Then, as noted in \cite{Liu:Sarkar:Zhao:2016}, and also shown in Appendix of this article using alternative and simpler arguments, ${\rm Lfdr}_{j|i}(\pi_{1},\pi_{2i})$ and  ${\rm Lfdr}_{i\centerdot}(\pi_{1}; \pi_{2i})$ can be written explicitly in terms of the ${\rm Lfdr}^*_{ij}(\pi_{2i})$'s as follows:
\begin{eqnarray}\label{eqn:Lfdr:ji}
  {\rm Lfdr}_{j|i}(\pi_{1},\pi_{2i}) \equiv  {\rm Lfdr}_{j|i}(\pi_{2i}) = \frac {{\rm Lfdr}^*_{ij} (\pi_{2i}) - {\rm Lfdr}^*_{i\centerdot} (\pi_{2i})}{1 - {\rm Lfdr}^*_{i\centerdot} (\pi_{2i})},
\end {eqnarray}
and
\begin{eqnarray}\label{eqn:Lfdr:idot}
  {\rm Lfdr}_{i\centerdot}(\pi_{1}; \pi_{2i}) \equiv {\rm Lfdr}_{i\centerdot}(\lambda_{i}; \pi_{2i}) = \frac{{\rm Lfdr}^*_{i\centerdot}(\pi_{2i})}{{\rm Lfdr}^*_{i\centerdot}(\pi_{2i}) + \lambda_i [1- {\rm Lfdr}^*_{i \centerdot}(\pi_{2i})]},
\end{eqnarray} where ${\rm Lfdr}^*_{i\centerdot} (\pi_{2i}) = \prod_{j=1}^{n_i}{\rm Lfdr}^*_{ij} (\pi_{2i})$, and
\begin{equation}\label{eqn:lambda}
  \lambda_i = \frac{\pi_{1}}{1-\pi_{1}} \div \frac{1-(1-\pi_{2i})^{n_i}}{(1-\pi_{2i})^{n_i}}. \end{equation}

\begin{remark} \rm The parameter $\lambda_i$ plays a key role in One-Way GAMM.
As noted from (3.3), $$\frac{Pr (\theta_{i \centerdot} = 1 | \vX)}{Pr (\theta_{i \centerdot} = 0 | \vX)} = \lambda_i \;
\times \;  \frac{Pr^*(\theta_{i \centerdot} = 1 | \vX)}{Pr^* (\theta_{i \centerdot} = 0 | \vX)},
$$ with $Pr^*$ denoting the probability under the standard two-class mixture model. That is,
\begin{eqnarray} \left.
  \begin{array}{r}
   \mbox{Posterior odds of significance } \\
   \mbox{of group $i$ under one-way GAMM} \end{array}
\right. & = & \lambda_i \times \left.
  \begin{array}{l}
   \mbox{Posterior odds of significance} \\
   \mbox{of group $i$ under no group structure} \end{array}\right.
\nonumber
\end{eqnarray}
In other words, $\lambda_i$ can be seen to act as a `group effect' in One-Way GAMM.
When $\lambda_i=1$, ${\rm Lfdr}_{ij}(\pi_{1}, \pi_{2i})$ reduces to ${\rm Lfdr}^*_{ij}(\pi_{2i})$, and so One-Way GAMM with $\lambda_i=1$ for all $i$ represents the case of `no group effect'.
As $\lambda_i$ increases from $1$, the posterior odds of the $i$th group being important increases under one-way grouping, which is likely to make our proposed procedures developed under One-Way GAMM more powerful in the sense of making more discoveries than those developed under the standard two-class mixture model. \end{remark}

We are now ready to develop our methods under One-Way GAMM.

\subsection{\it {Methods Answering Q1 and Q2}}

Let $\delta_{ij}(\vX) \in \{0,1\}$ be the decision rule associated with $\theta_{ij}$. Similar to $\theta_{ij}$, we express $\delta_{ij}(\vX)$ as follows: $\delta_{ij} (\vX) =\delta_{i\centerdot}(\vX)\cdot \delta_{j|i}(\vX)$, with $\delta_{i \centerdot}(\vX) = I(\sum_{j=1}^{n_i}\delta_{ij}(\vX) > 0) \in \{0,1\}$ and $\delta_{j|i}(\vX) = \delta_{ij}(\vX)/\max(\delta_{i\centerdot} (\vX), 1)$.

This article focuses on developing $\delta_{ij}(\vX)$ for $i=1, \ldots, m$, $j=1, \ldots, n_i$, controlling the following error rates at a given level $\alpha$ under One-Way GAMM: (i) The posterior expected proportion of false discoveries across all hypotheses, referred to as the total posterior FDR (PFDR$_T$), defined below \begin{eqnarray}
{\rm PFDR}_T = E \left [  \frac {\sum_{i=1}^m\sum_{j=1}^{n_i} (1- \theta_{ij})\delta_{ij}(\vX)}{\max \left \{\sum_{i=1}^m\sum_{j=1}^{n_i} \delta_{ij}(\vX), 1 \right \}} \; \bigg | \vX \right ], \end{eqnarray} to answer Q1; and (ii) the posterior expected average false discovery proportion across selected groups, referred to as
the selective posterior FDR (PFDR$_{\mathcal{S}}$), defined below \begin{eqnarray}   {\rm PFDR}_{S}  =   E \left [ \frac {1}{|\mathcal{S}|} \sum_{i \in \mathcal{S}} \frac{ \sum_{j=1}^{n_i}(1-\theta_{ij})\delta_{ij}(\vX)}{\max \left \{\sum_{j=1}^{n_i} \delta_{ij}(\vX), 1 \right \}} \bigg | \vX \right ], \end {eqnarray} with $\mathcal{S}$ being the set of indices for the selected groups, to answer Q2. The expectations in (3.5) and (3.6) are taken with repect to the $\theta_{ij}$'s given $\vX$.

\begin {remark} \rm Some remarks regarding the methods to be developed in the next subsection are worth making at this point. Hiding the symbol $\vX$ in the $\delta$'s for notational convenience, we first note that $E \left [  \sum_{j=1}^{n_i} (1- \theta_{ij})\delta_{ij}| \vX \right ] = \sum_{j=1}^{n_i} {\rm Lfdr}_{ij}(\lambda_i, \pi_{2i})\delta_{ij}$ can be expressed as either
\begin {eqnarray} \sum_{j=1}^{n_i} \left [ 1- \frac{\lambda_i [1- {\rm Lfdr}^*_{ij}(\pi_{2i})]}{\lambda_i+ (1-\lambda_i){\rm Lfdr}^*_{i\centerdot}(\pi_{2i})} \right ] \delta_{ij}, \end {eqnarray} as noted by using (3.2) and (3.3) in (3.1), or as
\begin {eqnarray}
  R_i\delta_{i\centerdot}  \left \{ 1- \left [1 - {\rm Lfdr}_{i\centerdot}(\lambda_i, \pi_{2i}) \right ] \left [1- {\rm PFDR}_{W_i} \right ] \right \},
\end{eqnarray}
where ${\rm PFDR}_{W_i} =  \sum_{j=1}^{n_i}\delta_{j|i}{\rm Lfdr}_{j|i}(\pi_{2i})/\max \left (R_i, 1 \right )$, with $R_i = \sum_{j=1}^{n_i} \delta_{j|i}$, is the within-group posterior FDR for group $i$, as noted from (3.1). However, we'll be using (3.7) in the expression for PFDR$_T$ and determine $\delta_{ij}$'s that will provide a single-stage approach to controlling this error rate, as opposed to \cite{Liu:Sarkar:Zhao:2016} where they use (3.8) to develop a two-stage Lfdr based approach to controlling not only ${\rm PFDR}_{T}$ but also PFDR$_{W_i}$, for each $i$. While determining $\delta_{ij}$'s controlling PFDR$_{\mathcal{S}}$, which, as said in Introduction, will produce for the first time an Lfdr analog of \cite {Benjamini:Bogomolov:2014}, we will consider controlling it along with controlling a measure of false selection of groups. In other words, our approach to controlling PFDR$_{\mathcal{S}}$ will be a two-stage one relying on its expression in terms of (3.8). \end{remark}

The above discussions provide a {\bf G}roup {\bf A}djusted {\bf TE}esting (GATE) framework for one-way classified hypotheses allowing us to produce Lfdr based algorithms (in their oracle forms) answering Q1 and Q2. We commonly refer to these algorithms as One-Way GATE algorithms.

\subsubsection{Answering Q1}

Before we present an algorithm in its oracle form answering Q1, it is important to note the following theorem that drives the development of it with some optimality property.

\begin{theorem}\label{thm:1:optimal} Let
\begin{equation}\label{eqn:pfnr}
  PFNR_{T}(\delta) =E\left[ \frac{\sum_{i=1}^m\sum_{j=1}^{n_i}\theta_{ij}(1-\delta_{ij}(\vX))}{ \max\{ \sum_{i=1}^m\sum_{j=1}^{n_i}(1-\delta_{ij}(\vX)),1\} }\bigg |\vX\right]
\end{equation} 
denote the total posterior FNR (PFNR$_{T}$) of a decision rule 
\[
\delta(\vX) = \{\delta_{ij}(\vX), i=1, \ldots m, j=1, \ldots, n_i\}.
\]
The PFNR$_{T}(\delta)$ of the decision rule $\delta(\vX)$ with $\delta_{ij}(\vX)= I(Lfdr_{ij}(\lambda_i, \pi_{2i})\le c)$, for $c \in (0,1)$ satisfying $PFDR_{T}(\delta)=\alpha$, is always less than or equal to that of any other $\delta_{ij}'(\vX)$ with $PFDR_{T}(\delta') \le  \alpha$.
\end{theorem}

A proof of this theorem can be seen in Appendix.

\begin{algorithm}[H]
  \caption{One-Way GATE 1 (Oracle). \label{alg:gate:1}}
  \begin{algorithmic}[1]
    \vspace*{1mm}
    \item Calculate $$\rm Lfdr_{ij}(\lambda_{i}, \pi_{2i}) = 1- \frac{\lambda_i [1- {\rm Lfdr}^*_{ij}(\pi_{2i})]}{\lambda_i+ (1-\lambda_i){\rm Lfdr}^*_{i\centerdot}(\pi_{2i})}, $$ the hypothesis specific local FDR under One-Way GAMM, for each $i=1, \ldots, m, j= 1, \ldots, n_i$;
    \item Pool all these ${\rm Lfdr}_{ij}(\lambda_{i}, \pi_{2i})$'s together and sort them as ${\rm Lfdr}_{(1)} \le \cdots \le {\rm Lfdr}_{(N)}$;
    \item Reject the hypotheses associated with Lfdr$_{(k)}$, $k=1, \ldots, R$, where $R = \max \left \{l: \sum_{k=1}^{l} {\rm Lfdr}_{(k)} \le l \alpha \right \}$.
  \end{algorithmic}
\end{algorithm}

\begin{theorem}\label{thm:1}
  The oracle One-Way GATE 1 controls PFDR$_T$ at $\alpha$.
\end{theorem}

This theorem can be proved using standard arguments used for Lfdr based approaches to testing single group of hypotheses (see, e.g., \cite{Sun:Cai:2007, Sarkar:Zhou:2008}). It is important to note that ${\rm PFDR}_{T}$ may not equal a pre-specified value of $\alpha$, and so Algorithm 1 is generally sub-optimal in the sense that it is the closest to one that is optimal as stated in Theorem 1.

\begin{remark}\label{remark:3}
  \rm When $\lambda_i=1$ for all $i$, i.e., when the underlying grouping of hypotheses has no effect in the sense that a group's own chance of being important is no different from when it is formed by combining a set of independent hypotheses, One-Way GATE 1 reduces to the standard Lfdr based approach (like that in \cite{Sun:Cai:2007, He:Sarkar:Zhao:2015}; and in many others). As we will see from simulation studies in Section 4, with $\lambda_i$ increasing (or decreasing) from $1$, i.e., when a group's chance of being important gets larger (or smaller) than what it is if the group consists of independent hypotheses, the standard Lfdr based approach becomes less powerful (or fails to control the error rate).

\end{remark}

\subsubsection {Answering Q2}

\iffalse 
There are applications in the context of selective inference of multiple groups/familes of hypotheses where discovering significant groups, and hence a control over a measure of their false discoveries, is scientifically no less meaningful than making such discoveries for individual hypotheses subject to a control over a similar measure of false discoveries across all of them. For instance, as \cite{Peterson:etal:2016} noted, in a multiphenotype genome-wide association study, which is often focused on groups/families of all phenotype specific hypotheses related to different genetic variants, rejecting  $H_{i\centerdot}$ corresponding to variant $i$ is considered an important discovery in the process of identifying phenotypes that are significantly associated with that variant. They borrowed ideas from \cite{Benjamini:Bogomolov:2014} and considered a hierarchical testing method that allows control of this so-called between-group FDR in the process of controlling the expected average of false discovery proportions across significant groups (due to \cite{Benjamini:Bogomolov:2014}).
\fi

To introduce the algorithm, consider the following notations: for $0< \alpha' < 1$, $R_{i}(\alpha') = \max \{1 \le k \le n_i: \sum_{j=1}^k{\rm Lfdr}_{(j)|i}(\pi_{2i}) \le k\alpha' \}$, with ${\rm Lfdr}_{(j)|i}(\pi_{2i})$, $j=1, \ldots, n_i$, being the sorted values of the Lfdr$_{j|i}(\pi_{2i})$'s in group $i$. 
When assuming all the Lfdr scores are available, we derive an algorithm answering Q2 alternative to the hierarchical testing method of \cite{Peterson:etal:2016}. It allows a control over \[ {\rm PFDR}_{B} = \frac{\sum_{i=1}^{m}\delta_{i \centerdot}{\rm Lfdr}_{i \centerdot}(\lambda_i, \pi_{2i})}{\max \left (\sum_{i=1}^{m} \delta_{i \centerdot}, 1 \right )}, \] an Lfdr analog of the aforementioned between-group FDR for the selected groups, while controlling PFDR$_{S}$.

\begin{algorithm}[H]
  \caption{One-Way GATE 2 (Oracle). \label{alg:gate:4:alt}}
  \begin{algorithmic}[1]
    \vspace*{1mm}
  \item Given an $\eta \in (0, \alpha)$, select the largest subset of group indices $\mathcal{S}$ such that $\frac{1}{|\mathcal{S}|} \sum_{i \in \mathcal{S}}{\rm Lfdr}_{i \centerdot}(\lambda_i, \pi_{2i}) \le \eta$;
  \item For each $i \in \mathcal{S}$, and any given $\alpha' \le \alpha$, find $R_i(\alpha')$ to calculate
  \begin {eqnarray}\label{alg:2} {\rm PFDR}_{S}(\alpha') = 1 - \frac{1}{|\mathcal{S}|} \sum_{i \in \mathcal{S}}\left(1 - {\rm Lfdr}_{i\centerdot}(\lambda_i, \pi_{2i}) \right) \left (1- \frac{1}{R_i(\alpha')} \sum_{j=1}^{R_i(\alpha')}{\rm Lfdr}_{(j)|i}(\pi_{2i}) \right ); \end {eqnarray}
  \item Find $\alpha^{*}(\mathcal{S}) = sup \{\alpha': {\rm PFDR}_S (\alpha') \le \alpha\}$;
  \item Reject the hypotheses associated with ${\rm PFDR}_S(\alpha^{*}(\mathcal{S}))$.
  \end{algorithmic}
\end{algorithm}

\begin{theorem}\label{thm:3}
{\it The oracle One-Way GATE 2 controls PFDR$_S$ at $\alpha$ subject to a control over PFDR$_B$ at $\eta < \alpha$.}
\end{theorem}

This theorem can be proved by noting that the left-hand side of (\ref{alg:2}) is the PFDR$_{S}$ of the procedure produced by Algorithm 2.

Let \[ {\rm PFNR}_{B} = E \left[ \frac{\sum_{i=1}^{m}\theta_{i \centerdot}(1-\delta_{i \centerdot}(\vX))}{ \max\{\sum_{i=1}^{m}(1-\delta_{i \centerdot}(\vX)),1\} }\bigg |\vX\right], \] and \[ {\rm PFNR}_{W_i}=E\left[ \frac{\sum_{j=1}^{n_i}\theta_{j|i}(1-\delta_{j|i}(\vX))}{ \max\{\sum_{j=1}^{n_i}(1-\delta_{j|i}(\vX)),1\} }\bigg |\vX\right] \]
denote between-group posterior FNR and within-group posterior FNR for group $i$, respectively, for a decision rule of the form $\delta_{ij}(\vX)= \delta_{i \centerdot}(\vX) \delta_{j|i}(\vX)$, with $\delta_{i \centerdot}(\vX) = I(Lfdr_{i \centerdot} (\lambda_i, \pi_{2i})\le c)$ and $\delta_{j|i}(\vX) = I(Lfdr_{j|i}(\pi_{2i}) \le c')$, for some $0<c,c' <1$, $i=1, \ldots, m$.

\begin {remark} \rm From Theorem \ref{thm:1:optimal}, we have the following optimality result regarding One-Way GATE 2: Given any $0<\eta<\alpha <1$,

\noindent (i) the PFNR$_{B}$ of the decision rule of the form $\delta_{i \centerdot}(\vX)= I(Lfdr_{i \centerdot }(\lambda_i, \pi_{2i})\le c)$ with $0<c<1$ satisfying $PFDR_{B}=\eta$ is less than or equal to that of any other $\delta_{i \centerdot}'(\vX)$ with $PFDR_{B} \le  \eta$.

\noindent (ii) Given $\delta_{i \centerdot}(\vX)$, $i=1, \ldots, m$, with $PFDR_{B} \le  \eta$, there exists an $\alpha'(\eta) \le \alpha$, subject to  ${\rm PFDR}_{S} = \alpha$, such that, for each $i$, ${\rm PFNR}_{W_i}$ of the decision rule of the form $\delta_{j|i}(\vX)= I(Lfdr_{j|i}(\pi_{2i})\le c')$ with $0<c'<1$ satisfying $PFDR_{W_i}= \alpha(\eta)$ is less than or equal to that of any other decision rule in that group for which  $PFDR_{W_i} \le \alpha'(\eta)$. \end{remark}

\begin{remark}\rm
  It is important to note that One-Way GATE  2 without Step 1 can be used in situations where the focus is on controlling PFDR$_S$ given a selection rule (or $\mathcal{S}$).
\end{remark}

\section{Numerical Studies}\label{sec:numerical} In this section we propose data-driven versions of One-Way GATE 1 and One-Way GATE 2, assuming that $\pi_{2i} = \pi_2$ for all $i$, and present results of numerical studies we conducted to examine their performances against other relevant methods.

\subsection{Data-Driven Methods}\label{sec:esti}
The development of our methods are based on One-Way GAMM model where $\pi_1, \pi_{2}$ and $f_1(x)$ are assumed to be known. In practice when these parameters are unknown, one can develop data-driven versions of these methods using various estimation techniques, such as the EM algorithm (\cite{Liu:Sarkar:Zhao:2016}) and others (\cite{Cai:Sun:Wang:2019}). In this paper, however, we take a Bayesian approach by considering a hierarchical structure for modeling these unknown parameters and deriving a Gibbs sampler to estimate the parameters using the corresponding posterior distributions. To make it more clear, we restate the One-Way GAMM by adding to it this hierarchical structure for the unknown parameters in the following:
\begin{eqnarray*}%%\label{def:model1}
  \left\{ \begin{array}{l}  X_{ij}~|~\theta_{i\centerdot},\theta_{j|i}  \stackrel {ind} \sim (1-\theta_{i \centerdot} \cdot \theta_{j|i})f_{0}(x_{ij})+\theta_{i \centerdot} \cdot \theta_{j|i}f_{1}(x_{ij}), \textrm{ for some densities $f_{0}$ and $f_{1}$,}\\
    P(\theta_{j|i} =0|~\theta_{i\centerdot} = 0) = 1, \textrm{for each $j=1, \ldots, n_i$};\\
     (\theta_{1|i}, \ldots, \theta_{n_i|i})~|~\theta_{i\centerdot} =1  \sim TPBern(\pi_{2}, n_i),\\
            \theta_{i\centerdot} \sim Bern(\pi_{1});
            \\
            f_1(x_{ij}) \sim \sum_{k=1}^KZ_{ij}^k N(\mu_k, \sigma_k^2), where \; \vZ_{ij}\sim Multinomial(1, \boldsymbol{\eta}_K),\\ \boldsymbol{\eta}_K\sim Dirichlet(d_1,d_2,\cdots,d_K),\\
            \mu_k\sim N(0, \sigma_{\mu}^2), \sigma_k^2 \sim InGamma(r,\nu), \pi_1\sim Beta(\alpha_1,\beta_1), \pi_2\sim Beta(\alpha_2,\beta_2).
  \end{array}\right.
\end{eqnarray*}

The following is the Gibbs sampler given by the above hierarchical structure of the model parameters:

\begin{algorithm}[H]
  \caption{Gibbs Sampler. \label{alg:gibbs}}
  \begin{algorithmic}[1]
    \vspace*{1mm}
  \item $P(\theta_{i\cdot}=0|\vX, \textrm{the rest})= Lfdr_{i\cdot}(\pi_1,\pi_2)$, provided in (\ref{eqn:Lfdr:idot});\
  \item Given $\theta_{i\cdot}=1$, generate $\theta_{j|i}\ind Bernoulli(1-Lfdr_{j|i}(\pi_1,\pi_2))$. Keep $\theta_{j|i}$ if $\sum_{j}\theta_{j|i}>0$; otherwise set $\theta_{k|i}=1$ where $Lfdr_{k|i}(\pi_1,\pi_2) = min_j Lfdr_{j|i}(\pi_1,\pi_2)$;
  \item \[
      \pi_1|\vX, \textrm{the rest} \sim Beta(\alpha_1+\sum_i\theta_{i\cdot}, \beta_1+G-\sum_i\theta_{i\cdot}),\]
    and
    \[
      \pi_2|\vX,\textrm{the rest}\sim Beta( \alpha_2+\sum_i\theta_{i\cdot}\sum_j\theta_{j|i}, \beta_2+\sum_i\theta_{i\cdot}\sum_j(1-\theta_{j|i}));
    \]
  \item $\vZ_{ij}|\vX, \textrm{the rest}\sim Multinomial(1, \mathbf{p})$ where
    \[
      \mathbf{p}=\frac{1}{\sum_k \eta_k\frac{1}{\sigma}\phi(\frac{x_{ij}-\mu_k}{\sigma})}\left(\eta_1\frac{1}{\sigma}\phi(\frac{x_{ij}-\mu_1}{\sigma}), \eta_2\frac{1}{\sigma}\phi(\frac{x_{ij}-\mu_2}{\sigma}), \cdots, \eta_K\frac{1}{\sigma}\phi(\frac{x_{ij}-\mu_K}{\sigma})\right);
    \]
    \item $$\mu_k|\vX, \textrm{the rest} \sim N\left( \frac{\sum_{ij}\theta_{i\cdot}\theta_{j|i}z_{ij}^kx_{ij}}{ \sum_{ij}\theta_{i\cdot}\theta_{j|i}z_{ij}^k}, \frac{\sigma_k^2}{\sum_{ij}\theta_{i\cdot}\theta_{j|i}z_{ij}^k}\right),$$
    and
    \[
    \sigma_k^2\sim InGamma\left( r+\frac{\sum_i\theta_i\sum_j\theta_{j|i}(x_{ij}-\mu_k)^2z_{ij}^k}{2}, \left( \frac{1}{\nu}+\sum_i\theta_i\sum_j\theta_{j|i}z_{ij}^k\right)^{-1}\right).
    \]
    \item $(\eta_1,\cdots,\eta_k)|\vX, \textrm{the rest} \sim Dirichlet\left( d_1+\sum_{ij}\theta_{i\cdot}\theta_{j|i}z_{ij}^1, d_2+\sum_{ij}\theta_{i\cdot}\theta_{j|i}z_{ij}^2,\cdots, d_K+\sum_{ij}\theta_{i\cdot}\theta_{j|i}z_{ij}^K\right)$.
  \end{algorithmic}
\end{algorithm}

In Step 2 we have used an approximation, which is necessitated by the computational difficulty, especially for large $n_i$, in finding the exact posterior distribution of the $\theta_{j|i}$'s as it requires enumeration of all the elements in $\{0,1\}^{n_i}$. 

The parameters $\pi_1, \pi_{2}, \eta_k$'s, $\mu_k$'s and $\sigma_k^2$'s are estimated by the medians based on samples drawn from the corresponding posterior distributions. These estimates are then used in place of the corresponding parameters appearing in the formulas for $Lfdr_{ij}(\pi_1,\pi_{2i})$ and $Lfdr_{j|i}(\pi_1,\pi_{2i})$ (with $\pi_{2i}=\pi_2$) in Algorithms \ref{alg:gate:1} and \ref{alg:gate:4:alt}, yielding our proposed data-driven GATE methods. These data-driven methods will be referred to as simply One-Way GATE 1 or One-Way GATE 2 in what follows.

\subsection{One-Way GATE 1}

We considered various simulation settings involving 2,000/5,000 hypotheses grouped into equal-sized groups to investigate how One-Way GATE 1 performs against its relevant competitors. These competitors are the Naive Method, SC Method (\cite{Sun:Cai:2009}) and GBH (\cite{Hu:Zhao:Zhou:2010}) Method, operating as follows under our model:

\noindent {\it Naive Method}: The single-group Lfdr based method of \cite{Sun:Cai:2007} is applied to the $mn$ hypotheses pooled together into a single group under a two-class mixture model $X_{ij}\sim (1-p)f_0(x_{ij}) + p f_1(x_{ij})$, with $p=p_{ij}= \pi_1 \pi_{2}/[1-(1-\pi_{2})^{n}]$, where the parameter $\pi_2$ and the density function $f_1(x)$ are estimated using Algorithm \ref{alg:gibbs}.

\noindent {\it SC Method}: The single-group Lfdr based method of \cite{Sun:Cai:2007} is applied to the $mn$ hypotheses pooled together into a single group assuming a two-class mixture model $X_{ij}\sim (1-\pi_{2})f_0(x_{ij}) + \pi_{2} f_1(x_{ij})$, where the parameters  $\pi_2$ and the density function $f_1(x)$ are estimated using Algorithm \ref{alg:gibbs}. %% for the $n$ hypotheses in group $i$, for each $i=1, \ldots, m$.
%% The single-group Lfdr based method of \cite{Sun:Cai:2007} is applied to the $mn$ hypotheses pooled together into a single group under a two-class mixture model $X_{ij}\sim (1-p)f_0(x_{ij}) + p f_1(x_{ij})$, with $p= \pi_1 \pi_{2}/[1-(1-\pi_{2})^{n}]$.

\noindent {\it GBH Method}: $X_{ij}$ is converted to its $p$-value $P_{ij}$. We considered both the LSL-GBH (Least-Slope Group BH) method and the TST-GBH (Two-Stage Group BH) method proposed in \cite{Hu:Zhao:Zhou:2010}.

The simulations involved generated triplets of observations $(X_{ij}, \theta_{i \centerdot}, \theta_{j|i})$ independently, $i=1, \ldots, m (= 100$ or $1,000$); $j=1, \ldots, n (=50$ or $5$), with (i) $\theta_{j|i}$'s jointly following TPBern($\pi_{2}, n)$; (ii) $\theta_{i\cdot}\sim Bern(\pi_{1})$ %where $\pi_{1}$ is determined from (\ref{eqn:lambda}) using $\lambda = k^2/49$ %%, the common value of the $\lambda_i$'s, for $k=1$, $2$, $\ldots$, $19$ or $20$ when $n_i=10$ and
; and (iii) $X_{ij}|\theta_{ij} \sim N(0,1)$ if $\theta_{ij}=0$, and $ \sim \sum_{k=1}^K c_k N(\mu_k, \sigma_k^2)$ if $\theta_{ij}=1$.  Different cases for choosing $(\pi_1,\pi_2)$ were considered. When $m=1,000, n=5$, we fixed $\pi_2$ and computed $\pi_1$ according to Equation (\ref{eqn:lambda}), with $\lambda$ being allowed to vary between 0.02 and 1.65. We wanted to observe the pattern of the aforementioned methods when $\lambda$ increases. When $m=100, n=50$, we fixed $\pi_2$ and let $\pi_1$ vary between 0.1 and 0.9 with increment of 0.1. The value of $\pi_2$ was chosen among 0.1, 0.3, 0.5, and 0.7. We considered $K=1$ and $2$ and 
\begin{eqnarray*}
  (\mathbf{c}, \mathbf{\mu},\mathbf{\sigma}^2) =\left\{\begin{array}{cc} (1, 2, 1) & \mbox{when} \quad K=1, \\ ((0.4,0.6), (-2, 2.5),(1.2, 0.8)) & \mbox{when} \quad K=2.
  \end{array}
  \right.
  \end{eqnarray*}

The One-Way GATE 1 in its oracle form, SC, TST-GBH, LSL-GBH, the One-Way GATE 1 using EM algorithm and the One-Way GATE 1 using Gibbs sampler were applied to the data for testing $\theta_{ij}=0$ against $\theta_{ij}=1$ simultaneously for all $(i,j)$ at $\alpha=0.05$. When running the Gibbs sampler, the hyper-parameters are chosen as: $\alpha_1=\alpha_2=\beta_1=\beta_2=1$, $d_1=\cdots=d_K=1$, $\sigma_\mu^2=1000$, $r=0.0001$ and $\nu=1000$. The simulated values of Bayes FDR (defined as the expectation of $\rm{PFDR}_{T}$ over $\vX$) and expected number of rejections were obtained for each of them based on 1000 replications.

\begin{figure}[H]
  \centering
  \includegraphics[height=40mm,width=40mm]{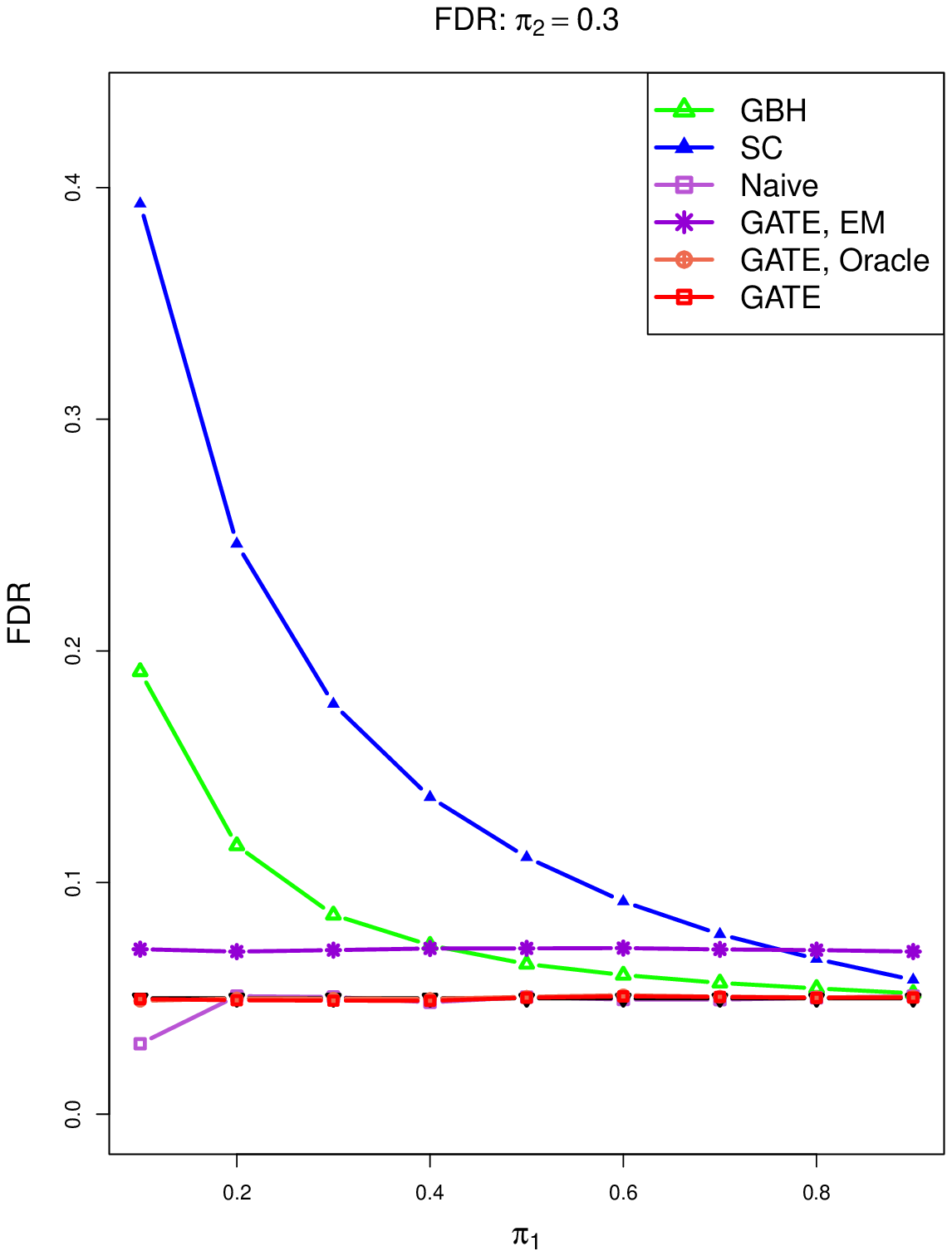}
  \includegraphics[height=40mm,width=40mm]{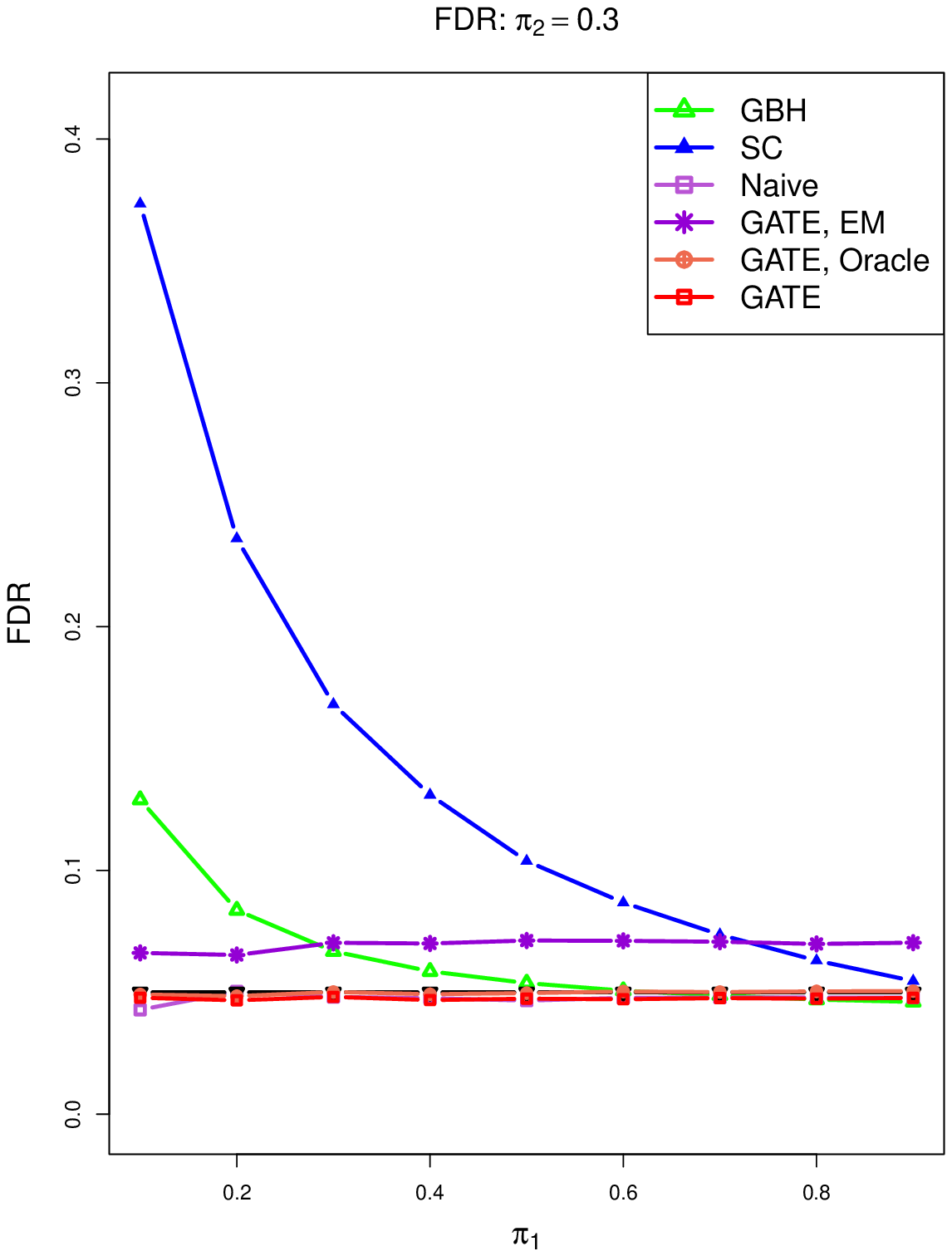}
  \caption{Performance in terms of FDR of One-Way GATE 1 across different values of $\pi_1$ when $m=100, n=50$. The  left and right panels correspond to $K=1$ and $2$, respectively.
  }\label{fig:gate1:fdr:1}
\end{figure}

\begin{figure}[H]
	\centering
	\includegraphics[height=40mm,width=40mm]{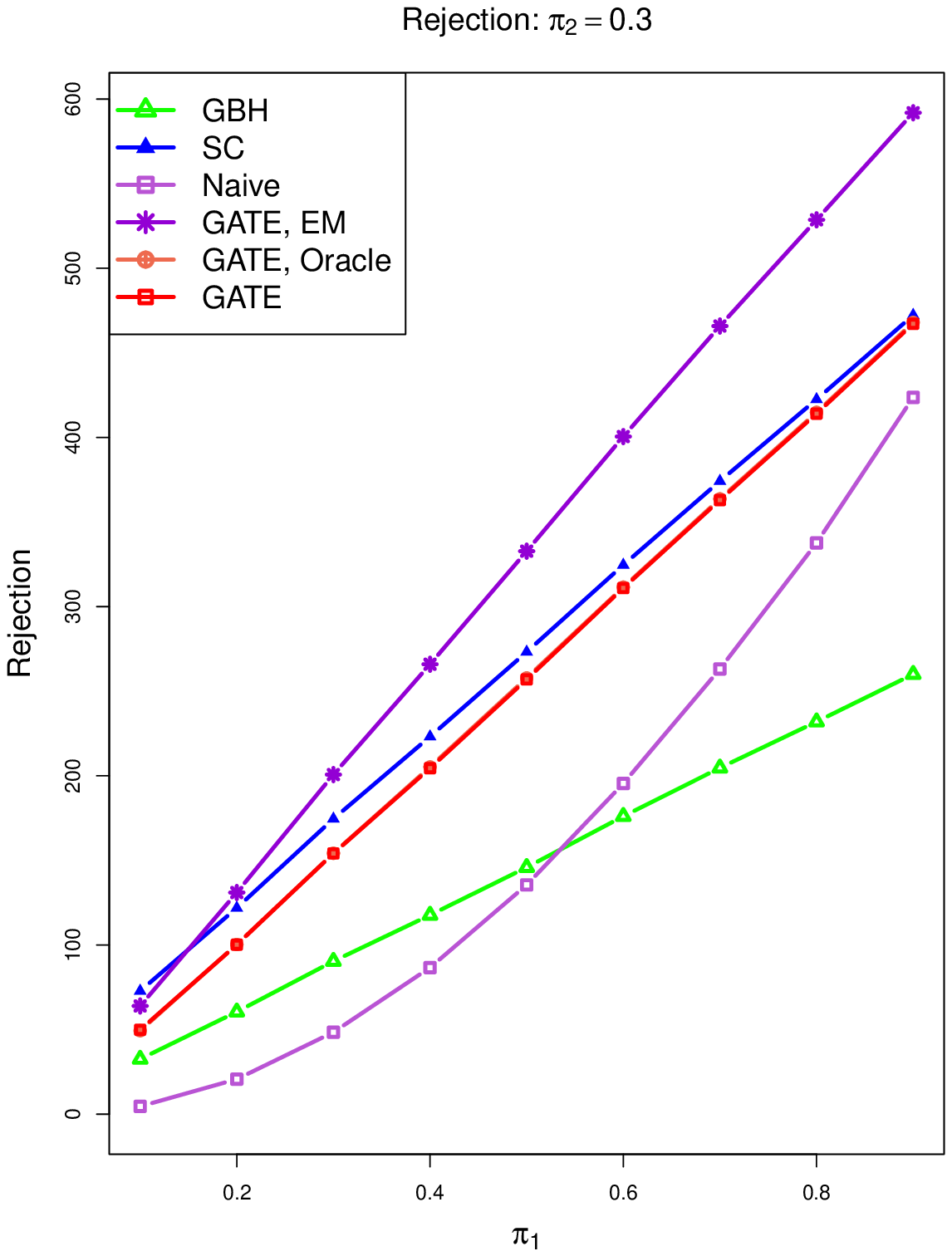}
	\includegraphics[height=40mm,width=40mm]{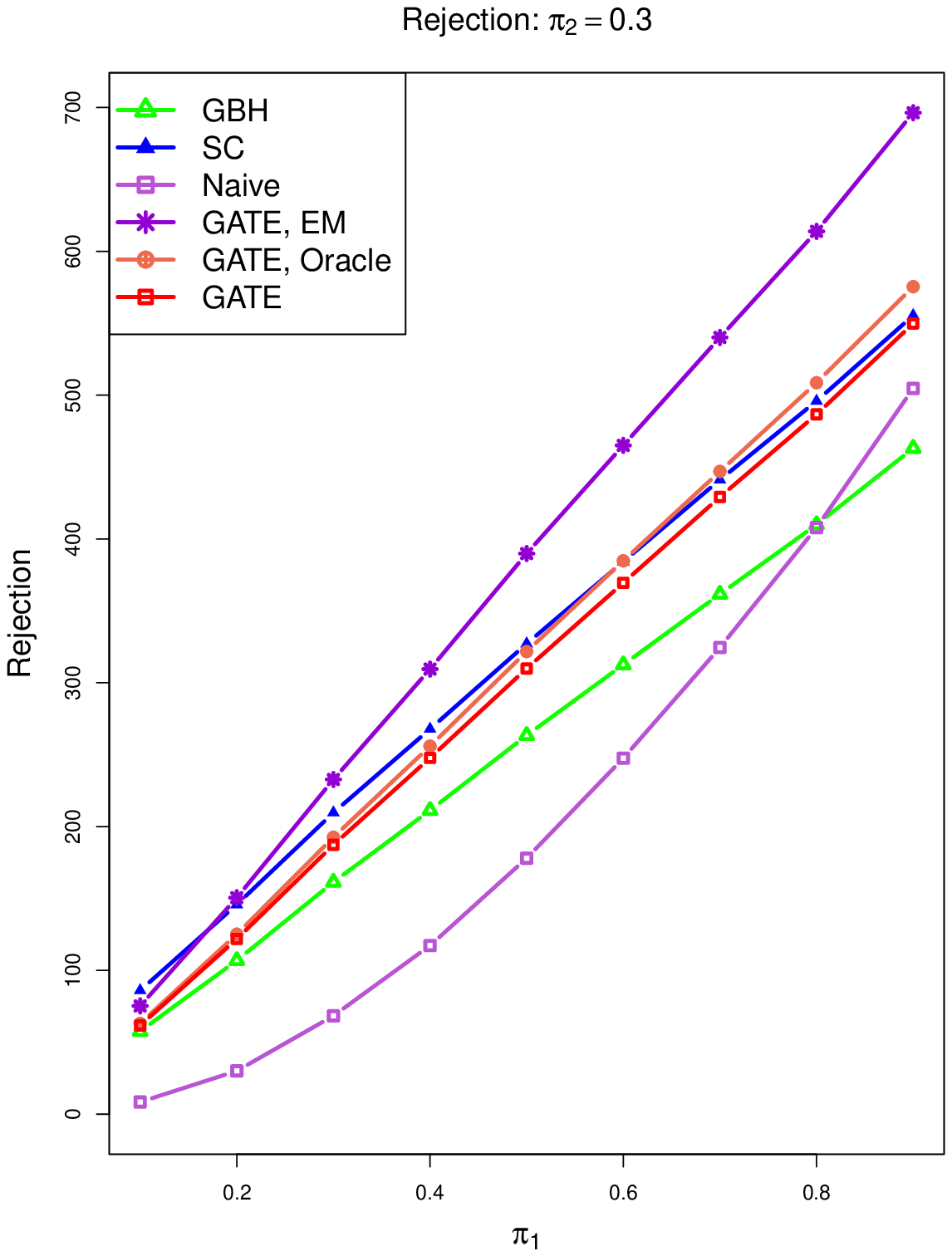}
	\caption{Performance in terms of expected number of rejections of One-Way GATE 1 across different values of $\pi_1$ when $m=100, n=50$. The left and right panels correspond to $K=1$ and $2$, respectively.
	}\label{fig:gate1:rej:1}
\end{figure}

Figures \ref{fig:gate1:fdr:1} and Figure \ref{fig:gate1:rej:1} display how these methods compare across different values of $\pi_1$. Our proposed One-Way GATE 1 (labeled GATE) is clearly seen to control the false discovery rate at the desired level 0.05. The SC, TST-GBH and One-Way GATE 1 using EM algorithm (labeled GATE-EM) fail to control the FDR at the desired level. The Naive Method and GBH-LSL control the FDR at the desired level; however, they are much less powerful than the One-Way GATE 1 in terms of the total number of rejections. The One-Way GATE 1 performs similar to its oracle version.

The results for the case of $m=1,000$ and $n=5$ are displayed in Figures \ref{fig:gate1:fdr:2} and \ref{fig:gate1:rej:2}. We plot the simulated values of FDR and expected number of rejections for all the methods   against $\lambda = \frac{\pi_1}{1-\pi_1} \frac{(1-\pi_2)^n}{1- (1- \pi_2)^n}$. As noted from these figures, the SC method fails to control the FDR when $\pi_1$ is small, or namely $\lambda$ is small. This happens because it uses a larger value of $\pi_{1}$ when $\lambda$ is small, inflating the FDR by an amount related to the value of $\lambda$. When $\lambda$ is larger, it uses a smaller value of $\pi_{1}$, resulting in a method which is conservative. The FDR level of the GBH-TST is inflated and GBH-LSL is too conservative. Interestingly, the oracle and two data-adaptive versions of One-Way GATE 1 work similarly under this setting.

\begin{figure}[H]
  \centering
  \includegraphics[height=40mm,width=40mm]{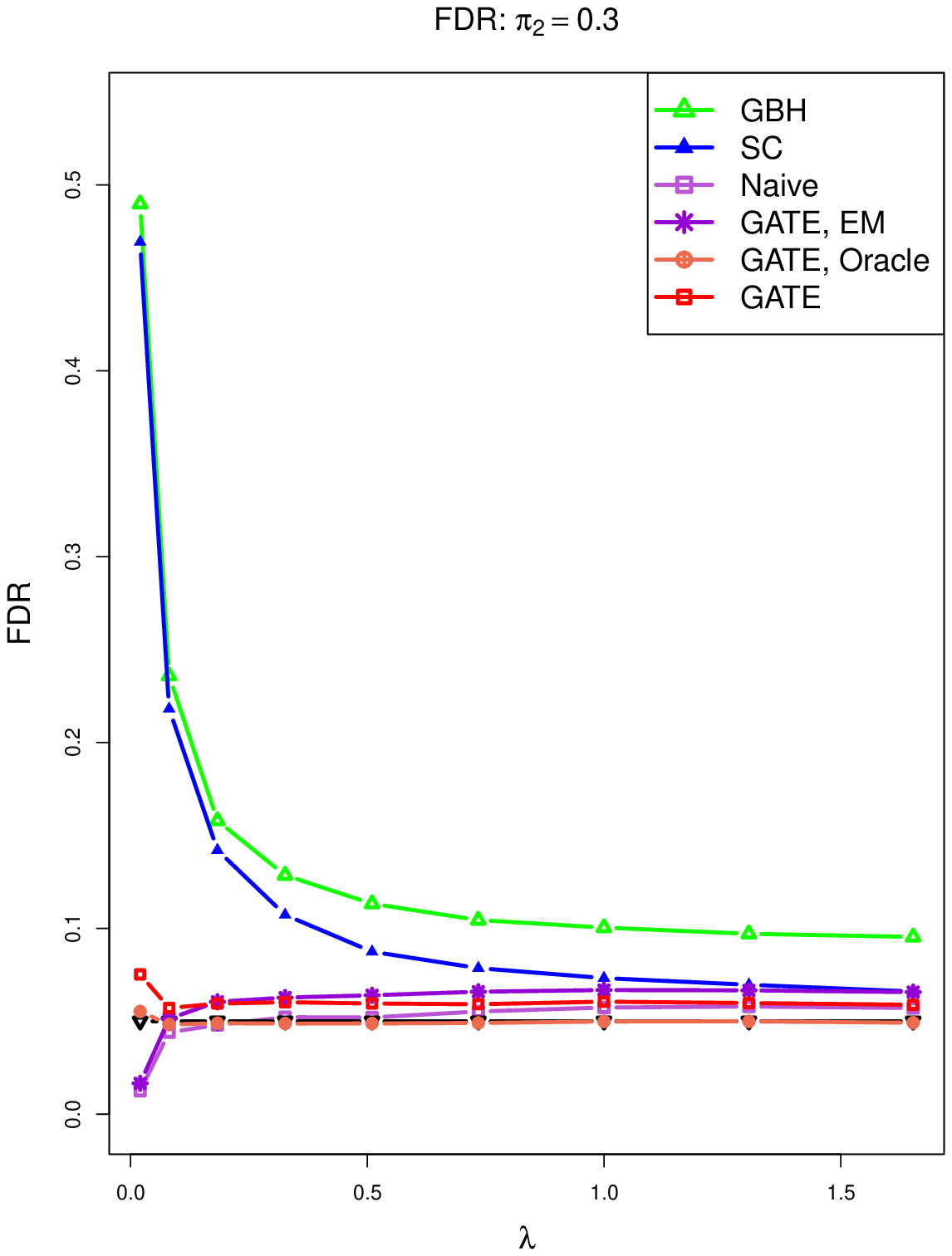}
  \includegraphics[height=40mm,width=40mm]{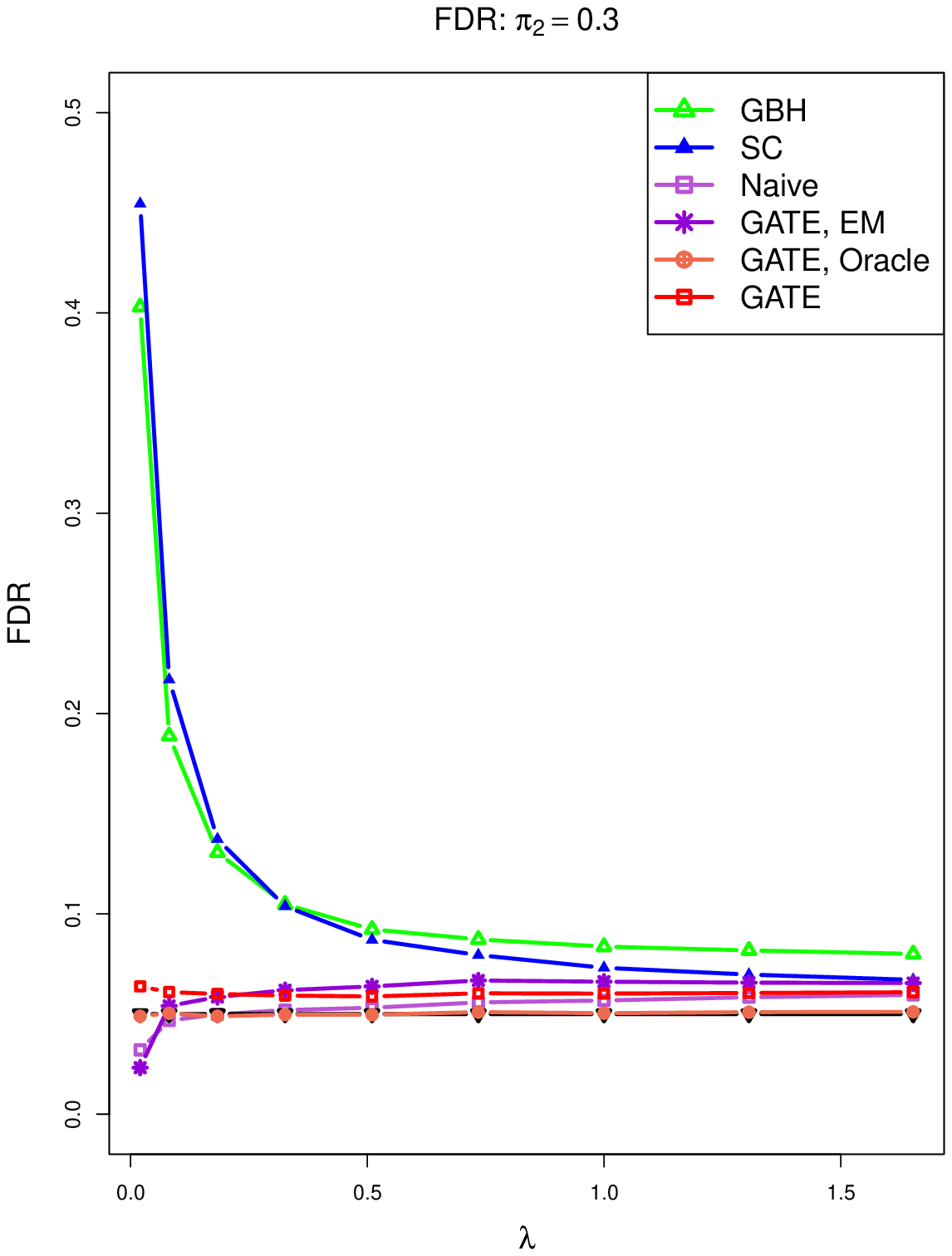}
  \caption{Performance in terms of FDR of One-Way GATE 1 across different values of $\lambda$ when $m=1000, n=5$.
  }\label{fig:gate1:fdr:2}
\end{figure}

\begin{figure}[H]
	\centering
	\includegraphics[height=40mm,width=40mm]{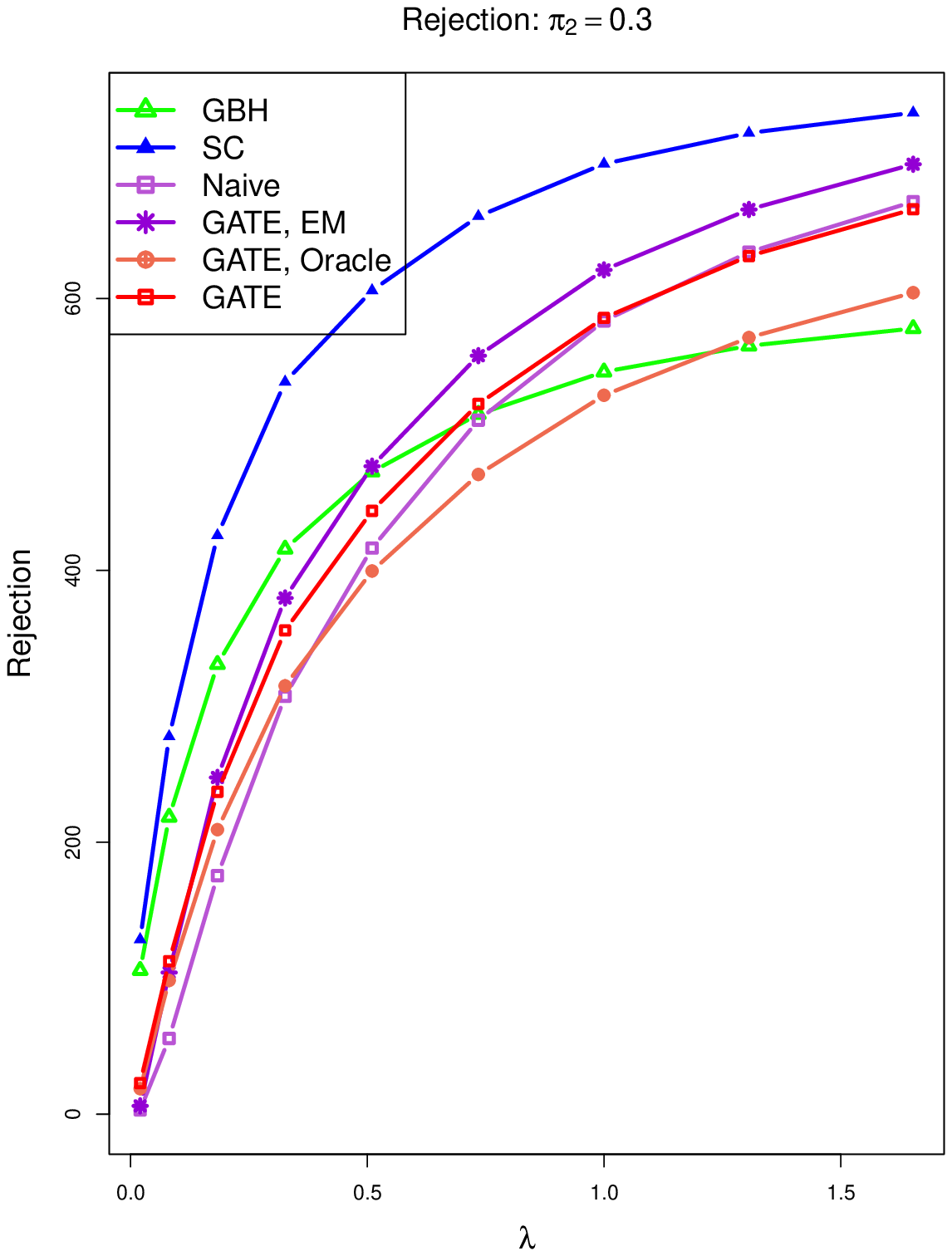}
	\includegraphics[height=40mm,width=40mm]{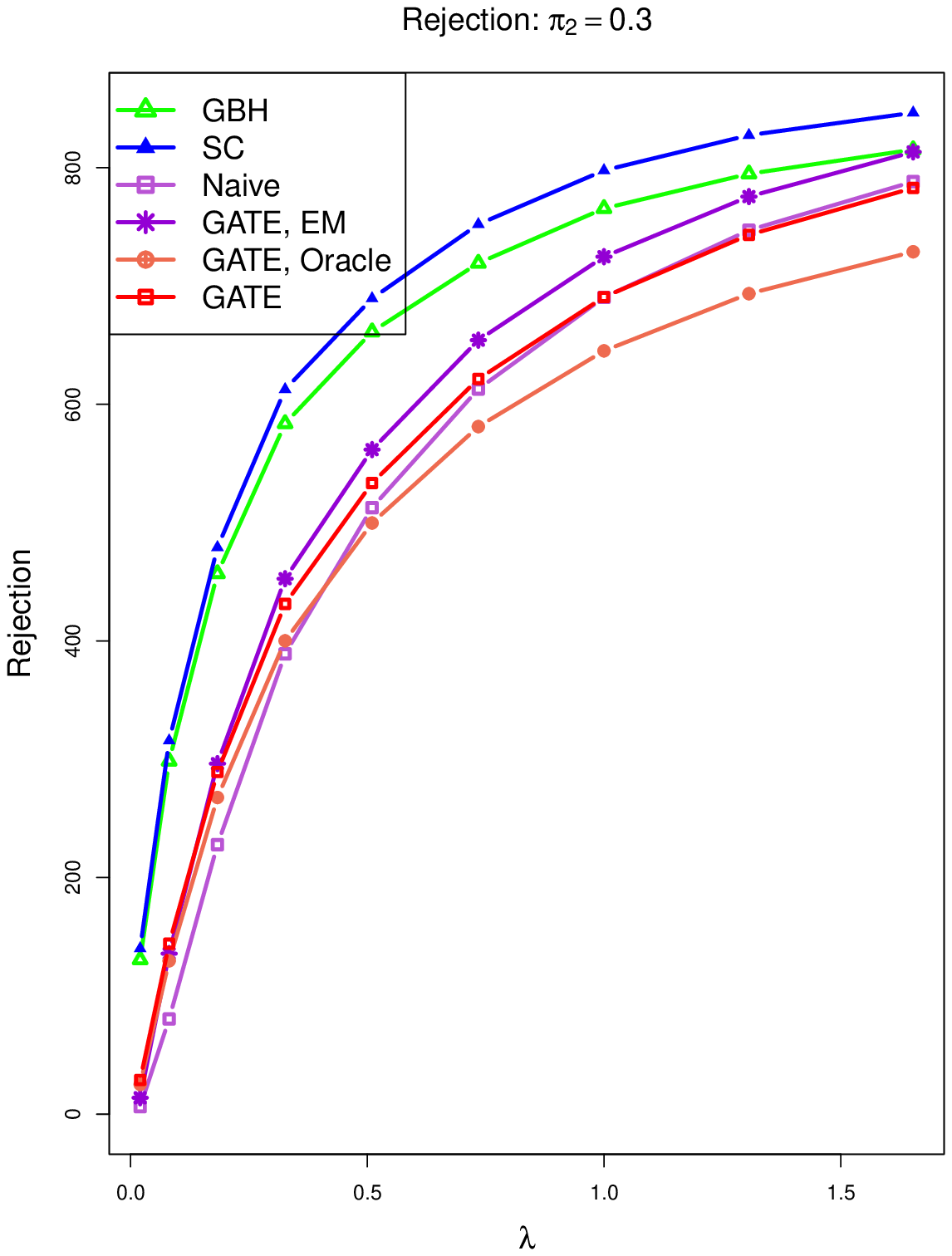}
	\caption{Performance in terms of expected number of rejections of One-Way GATE 1 across different values of $\lambda$ when $m=1000, n=5$.
	}\label{fig:gate1:rej:2}
\end{figure}

As suggested by a referee, we investigated how well the estimated Lfdr score using our proposed One-Way GATE 1, denoted by $\widehat{\rm{Lfdr}}$, approximates the true score $\rm{Lfdr}$. For that, we considered the following quantity: 
\[
{E L}(\widehat{\rm{Lfdr}}, \rm{Lfdr}) = \frac{ \sum_{ij} \mathbf{1}(\rm{Lfdr}_{ij}<0.1) \left(\widehat{\rm{Lfdr}}_{ij}-\rm{Lfdr}_{ij}\right)^2}{\sum_{ij} \mathbf{1}(\rm{Lfdr}_{ij}<0.1)},
\]
the mean squared error of $\widehat{\rm {Lfdr}}$ given that the true scores are low. It is an efficiency measure for our estimate when there are few strong signals. We simulated this measure under different parameter settings based on 1000  replications and reported them in Figure \ref{fig:gate1:lfdr:s1}. As seen from this figure, the estimated local fdr using our method gets more efficient as $\pi_2$ and $\lambda$ increase. 

\begin{figure}[H]
  \centering
  \includegraphics[height=50mm,width=50mm]{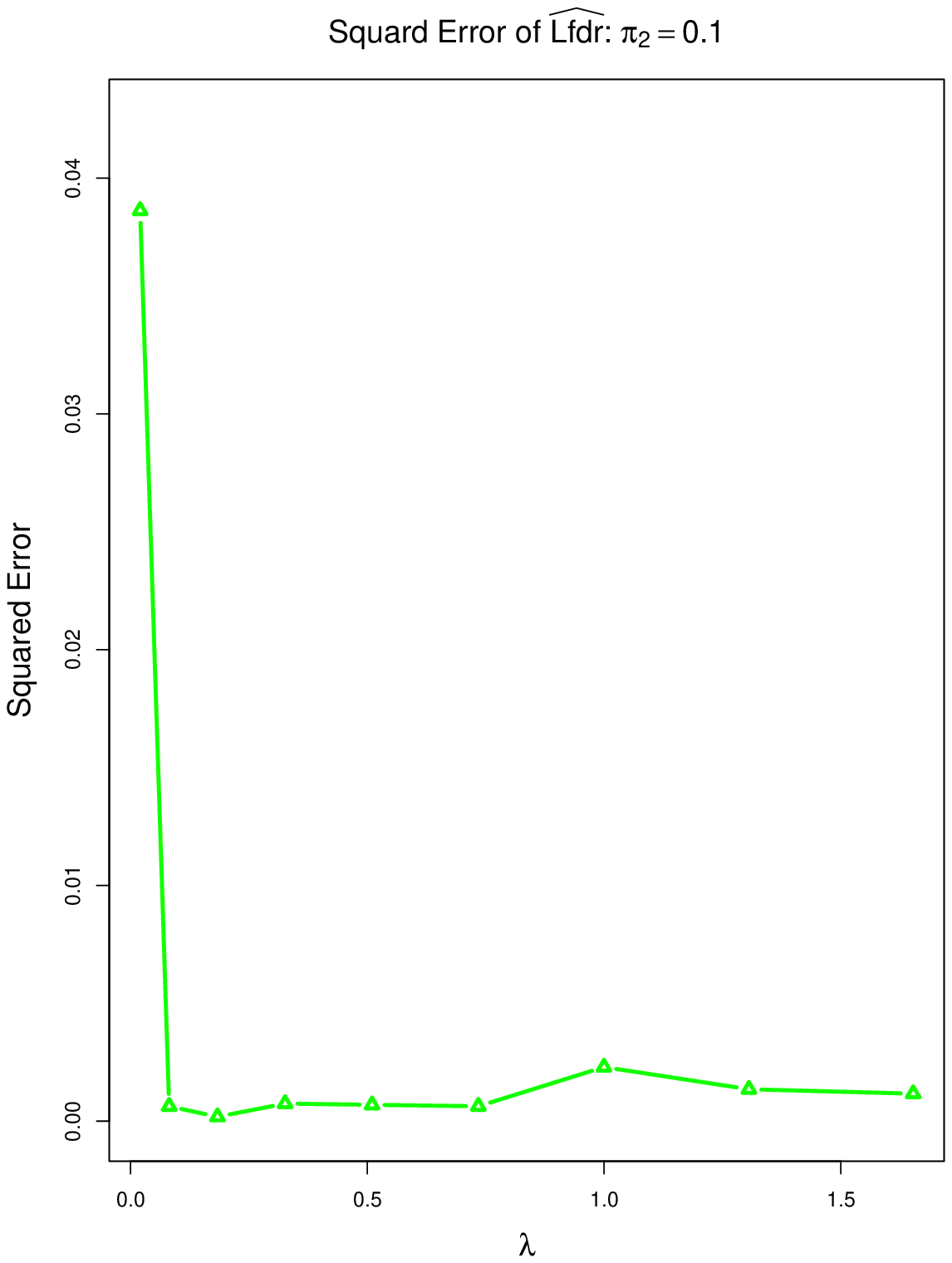}
   \includegraphics[height=50mm,width=50mm]{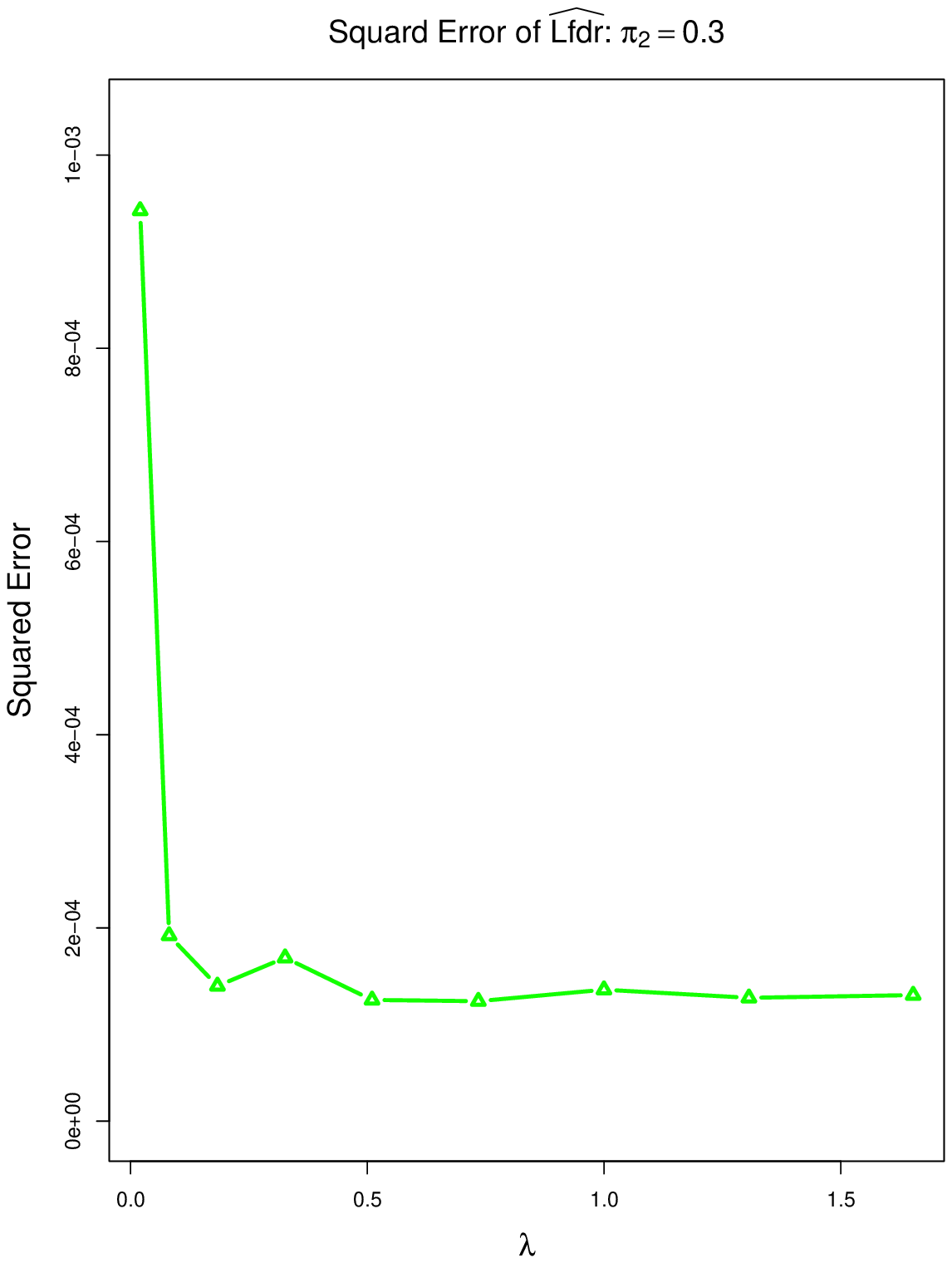}\\
    \includegraphics[height=50mm,width=50mm]{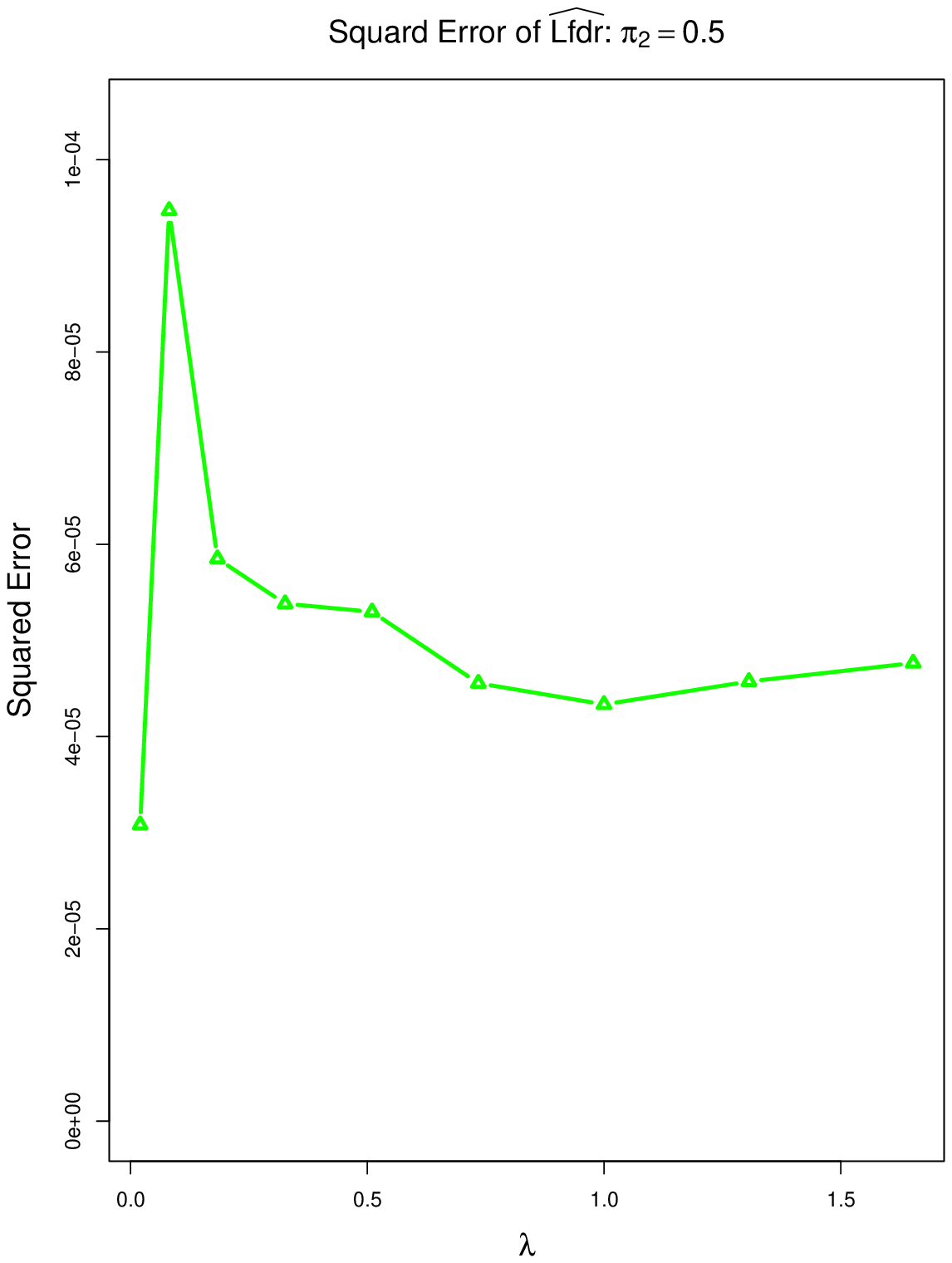}
    \includegraphics[height=50mm,width=50mm]{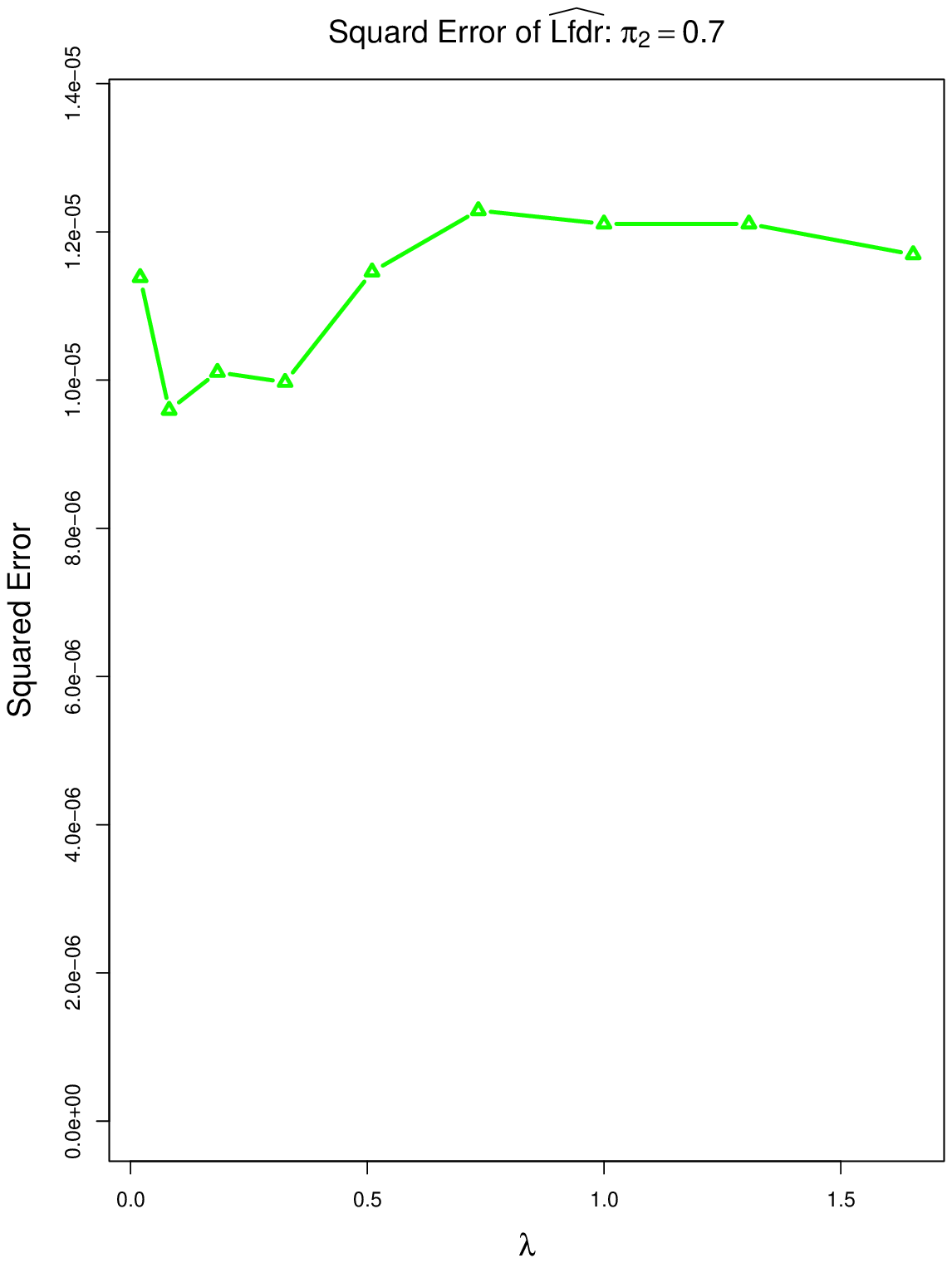}
  
  \caption{ The mean squared error for the estimation of the local fdr scores using the proposed method: $m=1000, n=5, L=1, \pi_2=0.1,0.3,0.5$ and $0.7$ respectively.
  }\label{fig:gate1:lfdr:s1}
\end{figure}

Additional numerical results obtained through extensive numerical calculations are put in Appendix. It is clearly seen that One-Way GATE 1 performs similarly to its oracle form, and, more importantly, it beats the performances of others which ignore the underlying group structure.

\subsection{One-Way GATE 2}
Simulation studies were conducted to compare One-Way GATE 2 to its competitors: the BB method (\cite{Benjamini:Bogomolov:2014}),  One-Way GATE 2 where the parameters are estimated according to the EM algorithm.

\iffalse
\noindent {\it Oracle BB method using Simes' combination}:
$X_{ij}$ is converted to its $p$-value $P_{ij}$. With $P_{i(1)} \le \cdots \le P_{i(n_i)}$ denoting the sorted $p$-values in group $i$, let   $P_{i\centerdot} = \min_{1\le j \le n_i} \{n_i(1-\pi_{2i})P_{i(j)}/j \}$ denote Simes' combination of the $p$-values in group $i$ in its oracle form, for $i=1, \ldots, m$. Let $\mathcal{G}$ be the set of indices of the group specific hypotheses $H_{i \centerdot}$ rejected using the oracle level $\alpha$ BH method based on  $(1-\pi_1)P_{i\centerdot}$, $i=1, \ldots, m$. Reject the hypotheses corresponding to $P_{i(j)}$ for all $i \in \mathcal{G}$ and $j \le R_i = \max \{ j: (1-\pi_1\pi_{2|i}) P_{i(j)} \le j|\mathcal{G}|\alpha/mn_i \}$.
\fi

{\it BB method}: Convert $X_{ij}$ to its $p$-value $P_{ij}$. With $P_{i(1)} \le \cdots \le P_{i(n_i)}$ denoting the sorted $p$-values in group $i$. Let $\mathcal{G}$ be the set of indices of the group that are rejected according to Step 1 in Algorithm \ref{alg:gate:4:alt}. Reject the hypotheses corresponding to $P_{i(j)}$ for all $i \in \mathcal{G}$ and $j \le R_i$ where $R_i= \max \{ j: P_{i(j)} \le \frac{j}{n_i}\frac{\#\mathcal{G}}{G} \}$.

{\it One-Way GATE 2 (EM)}: the hyper-parameters are estimated according to the EM algorithm before applying it in One-Way GATE 2.

The comparison was made in terms of Bayes selective FDR (defined as the expectation of ${\rm PFDR}_{\mathcal{S}}$ over $\vX$), expected  number of total rejections, and average number of true rejections calculated under same simulation setting as in One-Way GATE 1. %We have done extensive simulation studies but report the results of a few representative ones.
Figures \ref{fig:gate4:a} and \ref{fig:gate4:d} present the comparison for the setting where $m=500$, $n=20$, $\pi_{1}=0.14$, $0.95$, and $\pi_{2}=0.30$. Additional numerical results are put in the appendix.

One-Way GATE 2 is seen to control the selective FDR, like the BB Method, quite well. However, the One-Way GATE 2 is seen to be more powerful in terms of yielding a large number of true rejections. The EM version of One-Way GATE 2, however, fails to control the overall FDR and selective FDR.

\begin{figure}[H]
   \centering
   \includegraphics[height=40mm,width=40mm]{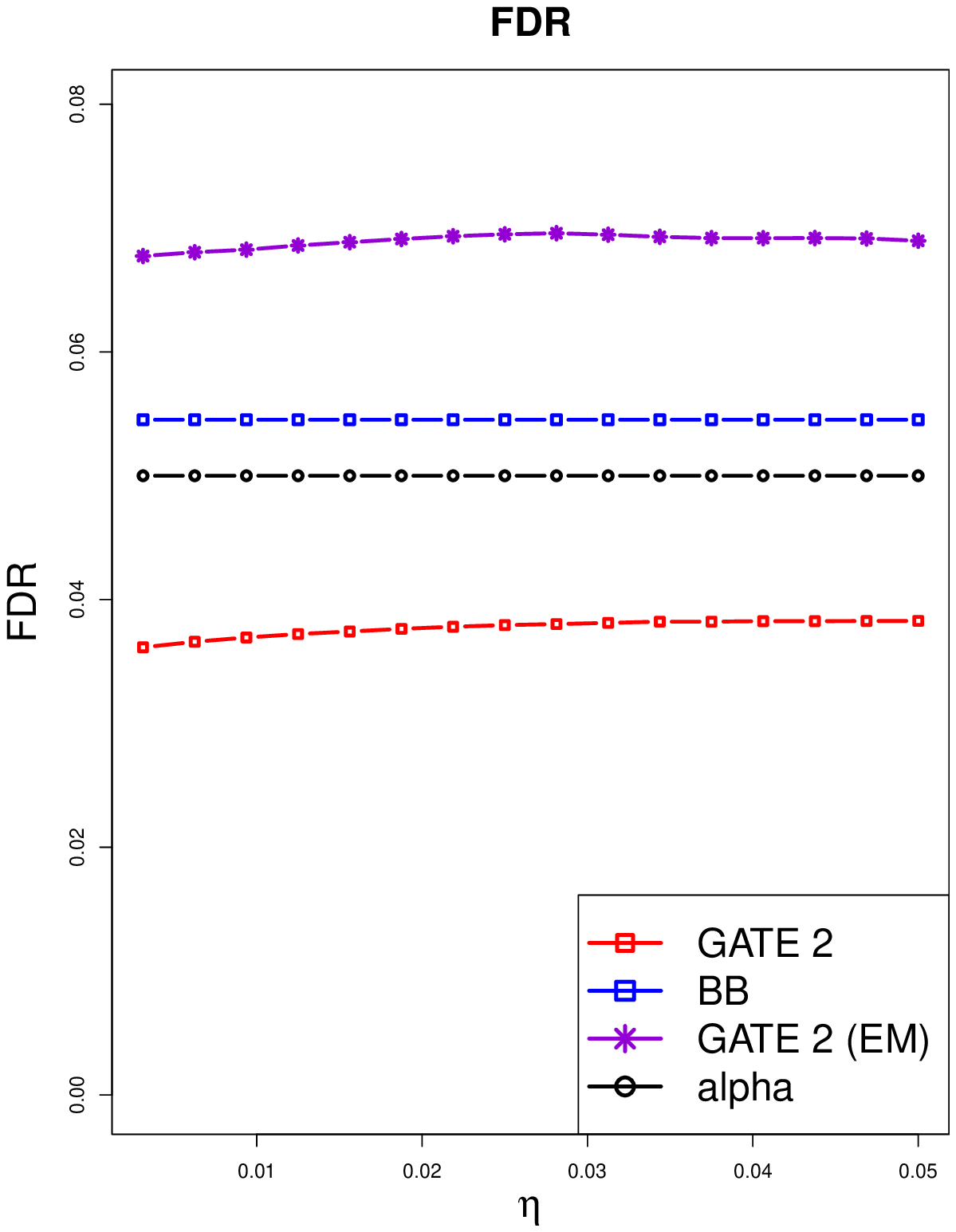}
   \includegraphics[height=40mm,width=40mm]{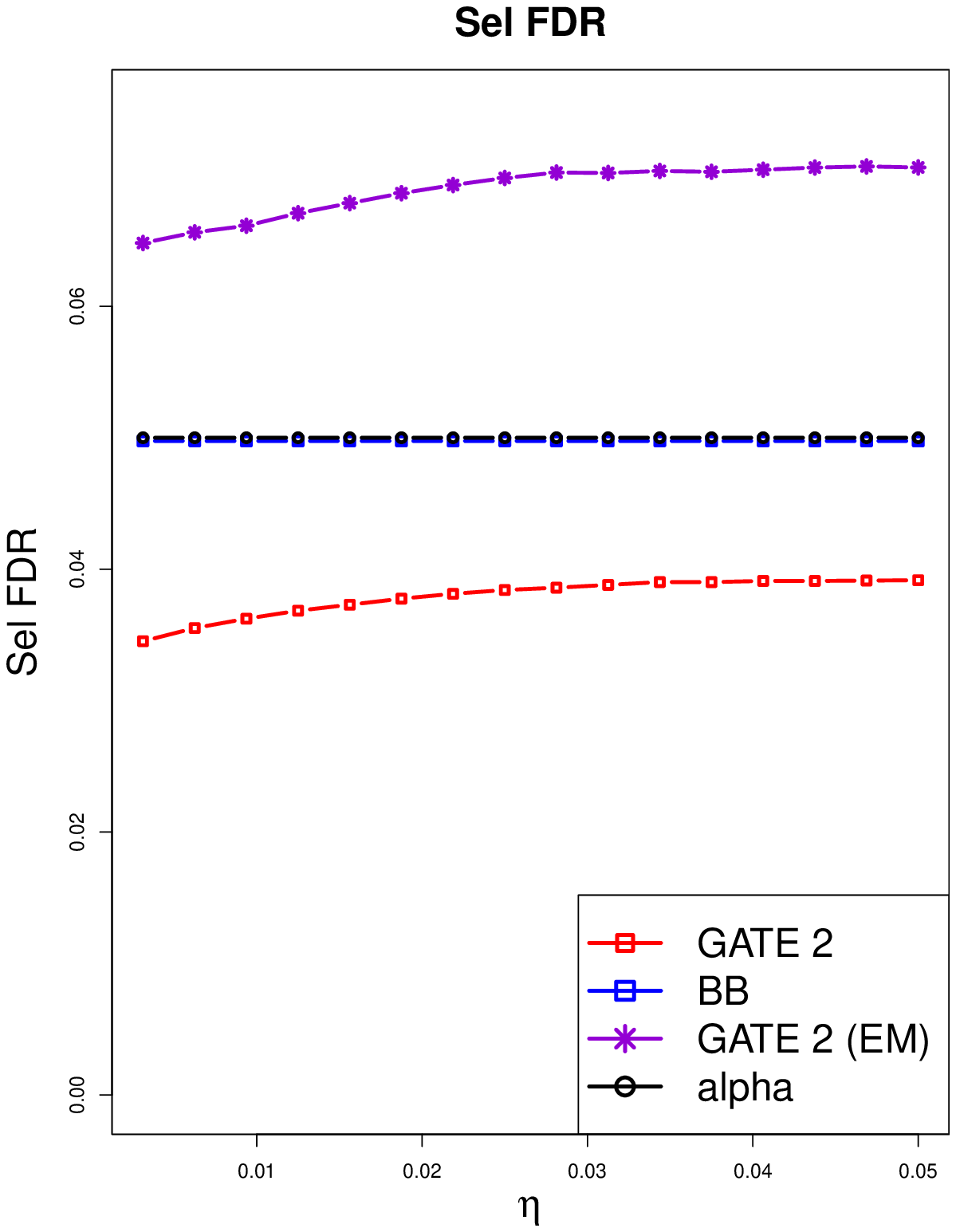}
   \includegraphics[height=40mm,width=40mm]{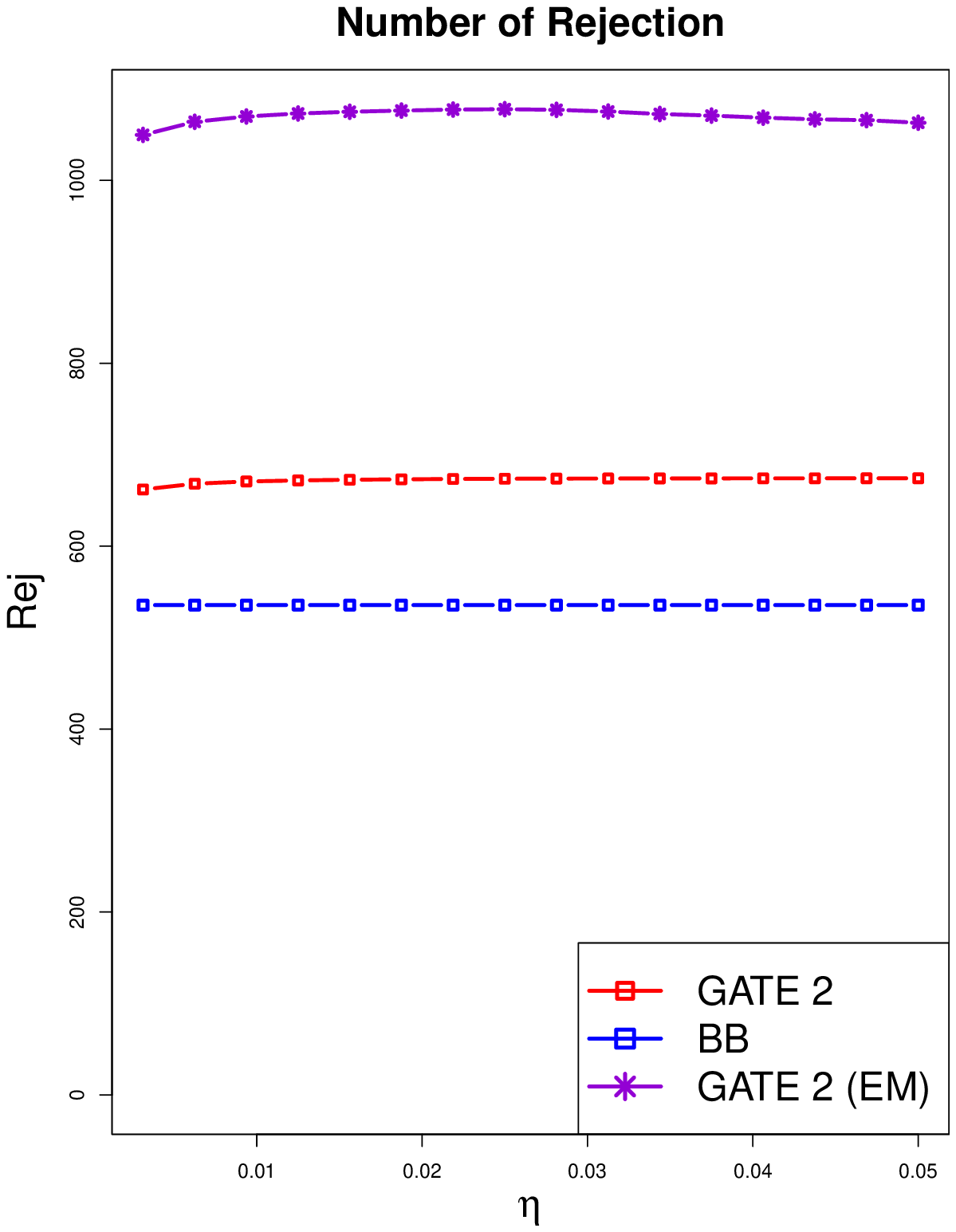}\\
   \includegraphics[height=40mm,width=40mm]{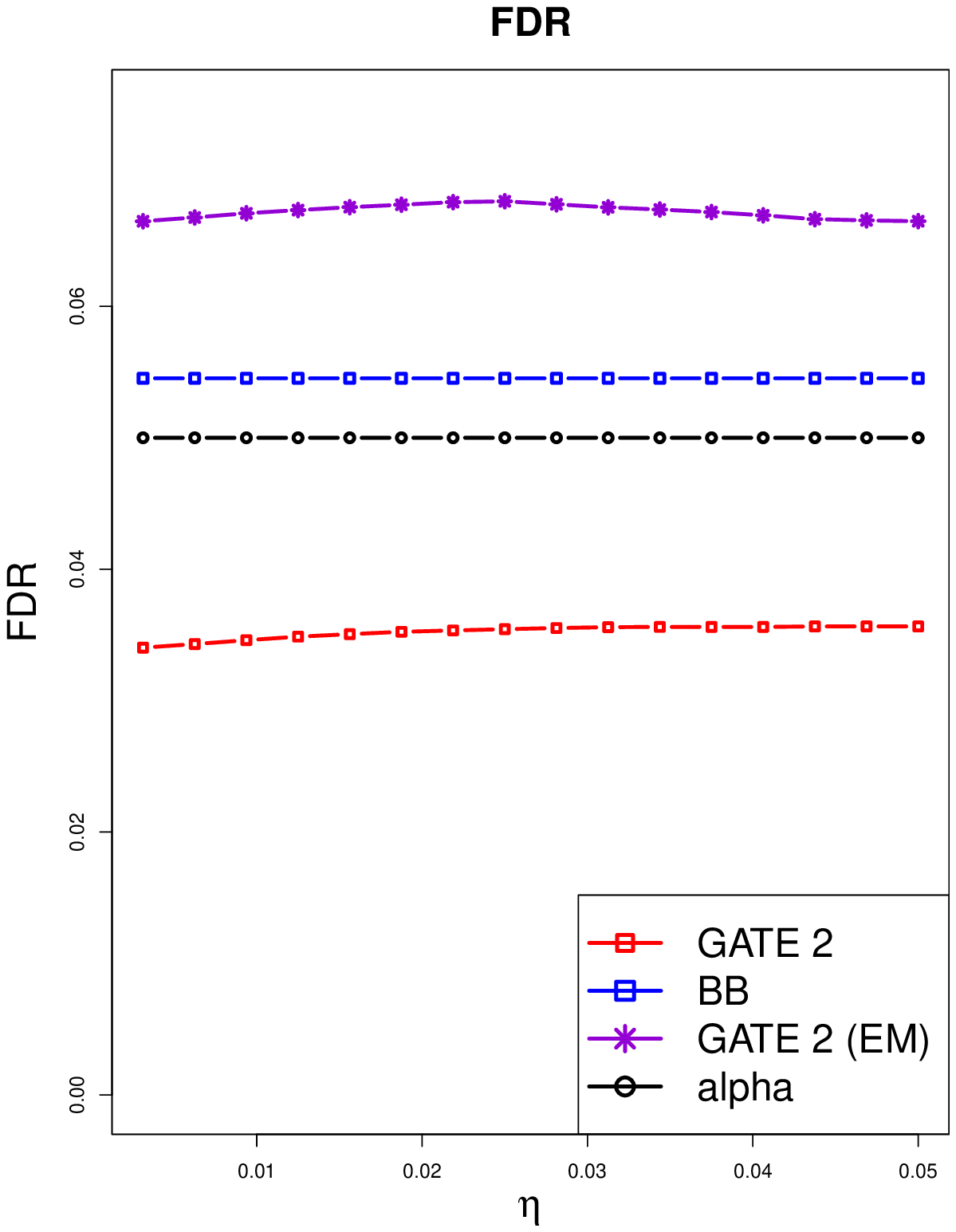}
   \includegraphics[height=40mm,width=40mm]{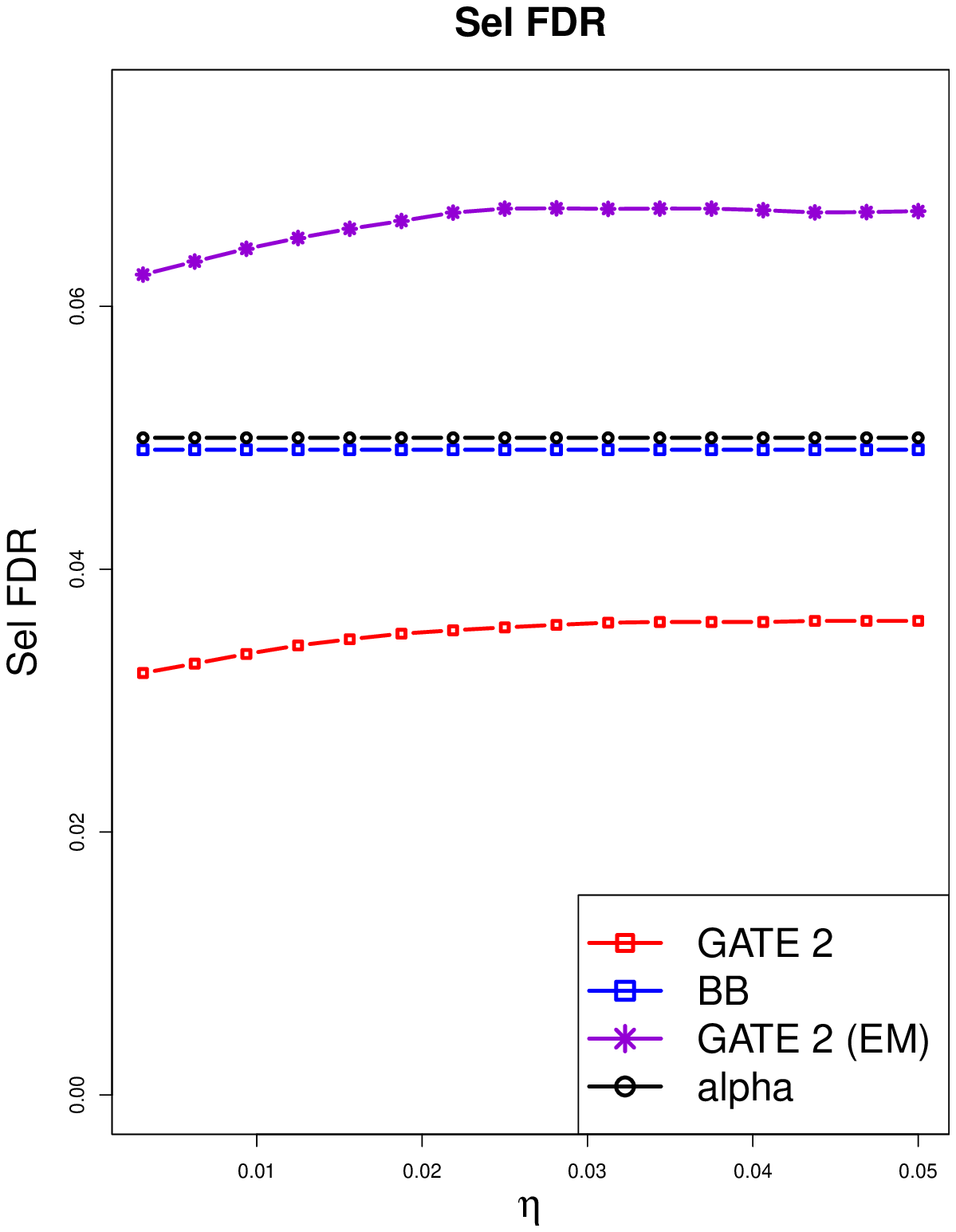}
   \includegraphics[height=40mm,width=40mm]{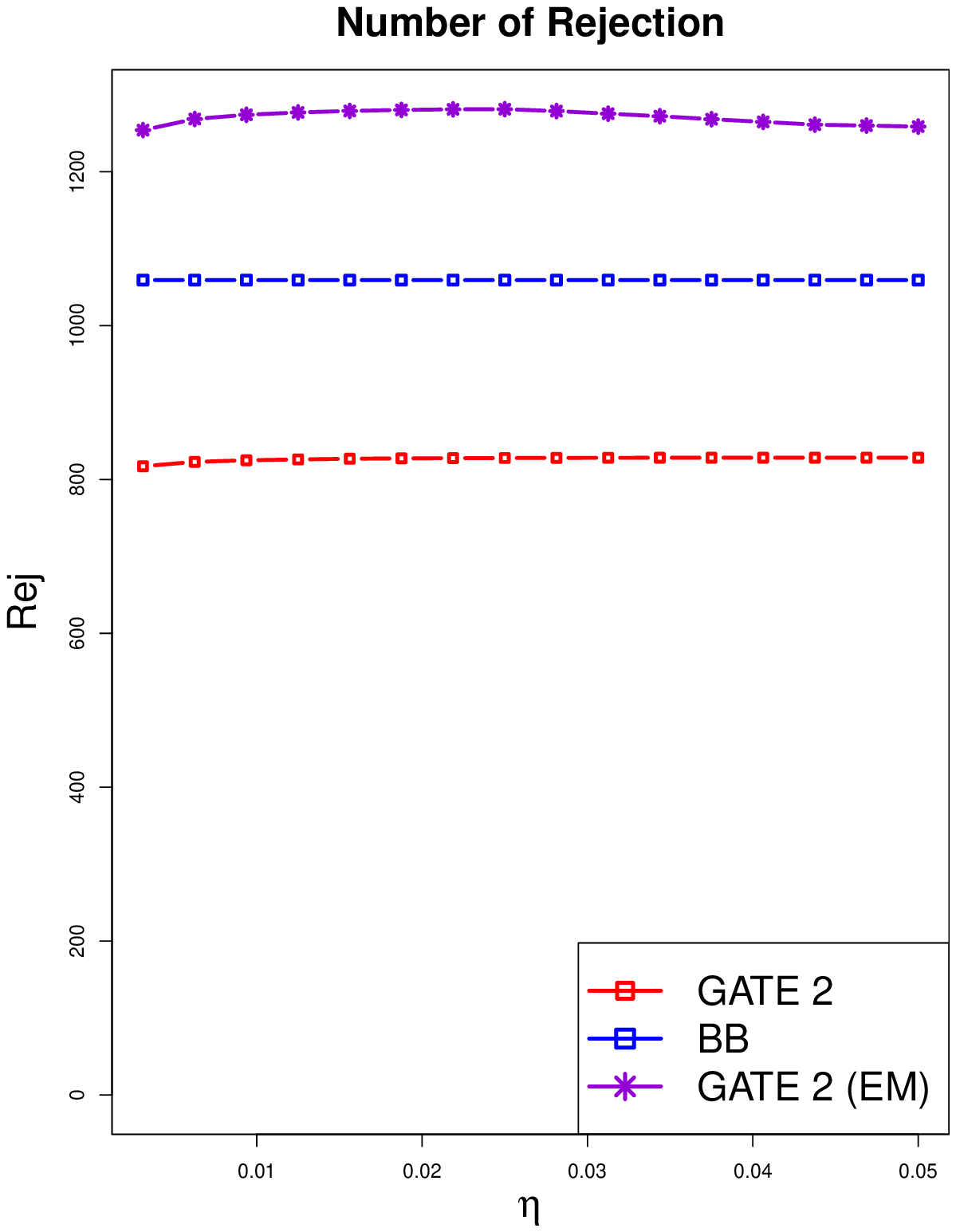}
   \caption{Performance of One-Way GATE 2 when $m=500, n=20, \pi_{1}=0.143$. The top panels correspond to cases when $K=1$ and the bottoms ones correspond to $K=2$. %The four panels correspond to four choices of $\pi_1$, which are 0.05, 0.3, 0.6, and 0.95 respectively.
   }\label{fig:gate4:a} 
\end{figure}

\begin{figure}[H]
   \centering
   \includegraphics[height=40mm,width=40mm]{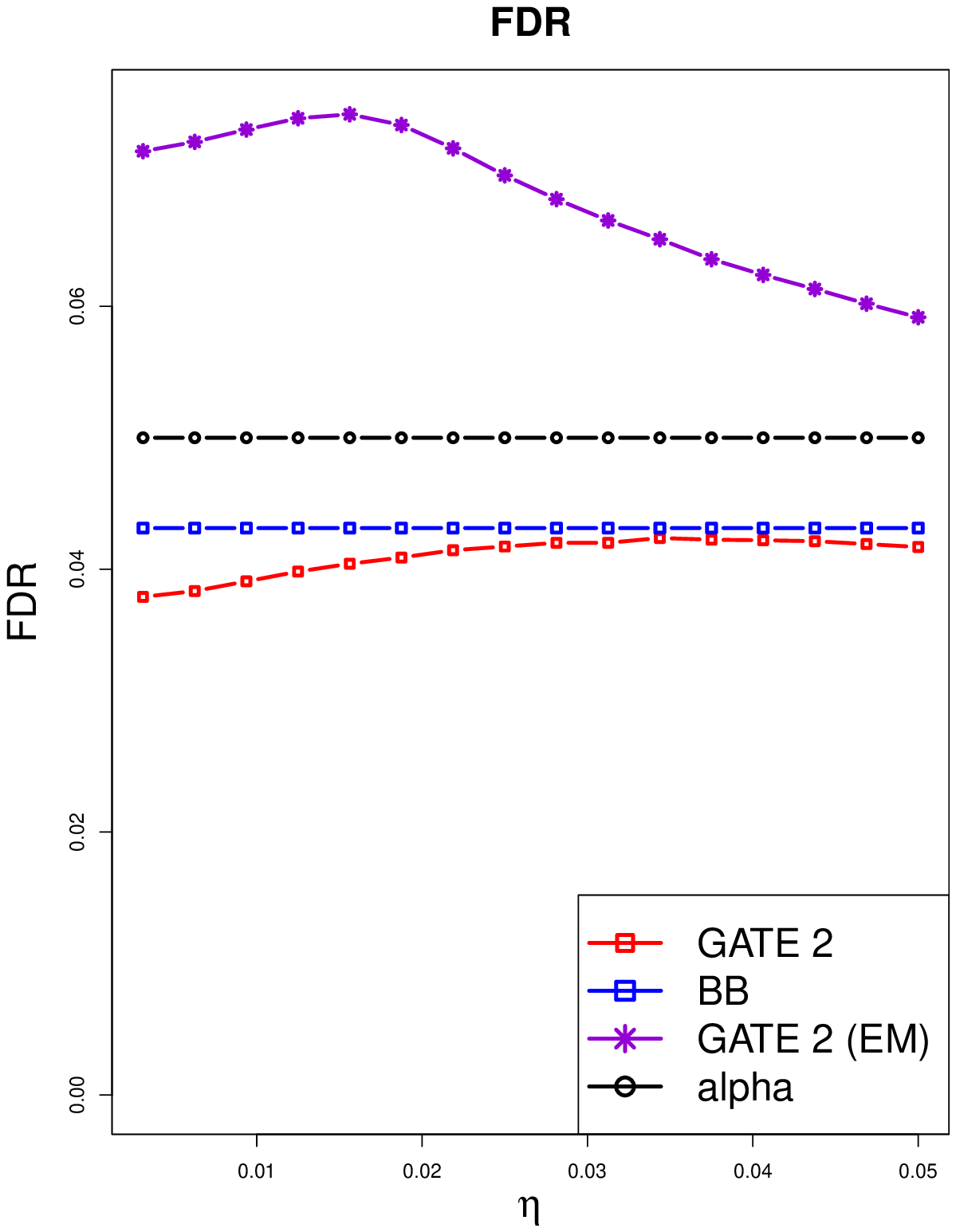}
   \includegraphics[height=40mm,width=40mm]{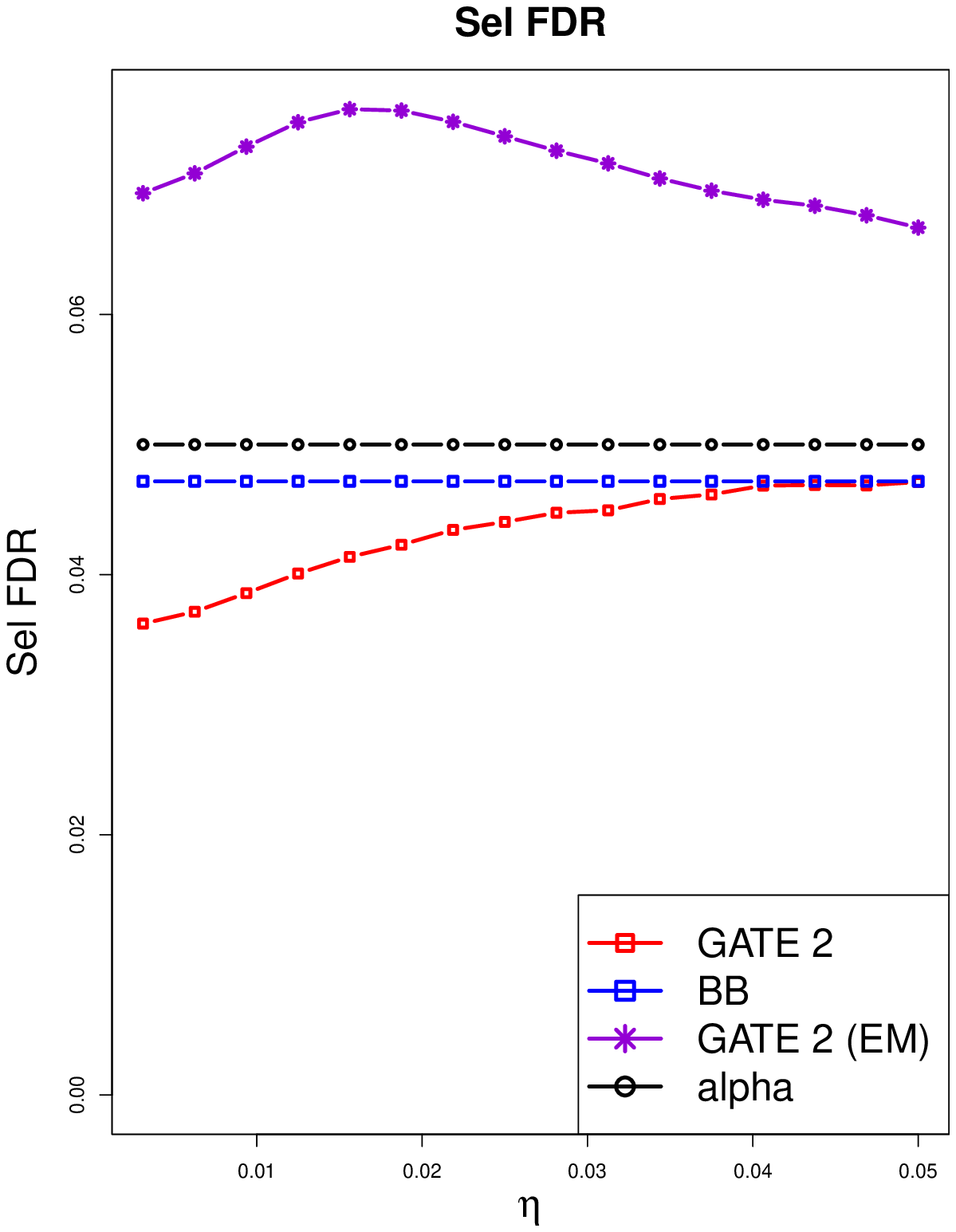}
   \includegraphics[height=40mm,width=40mm]{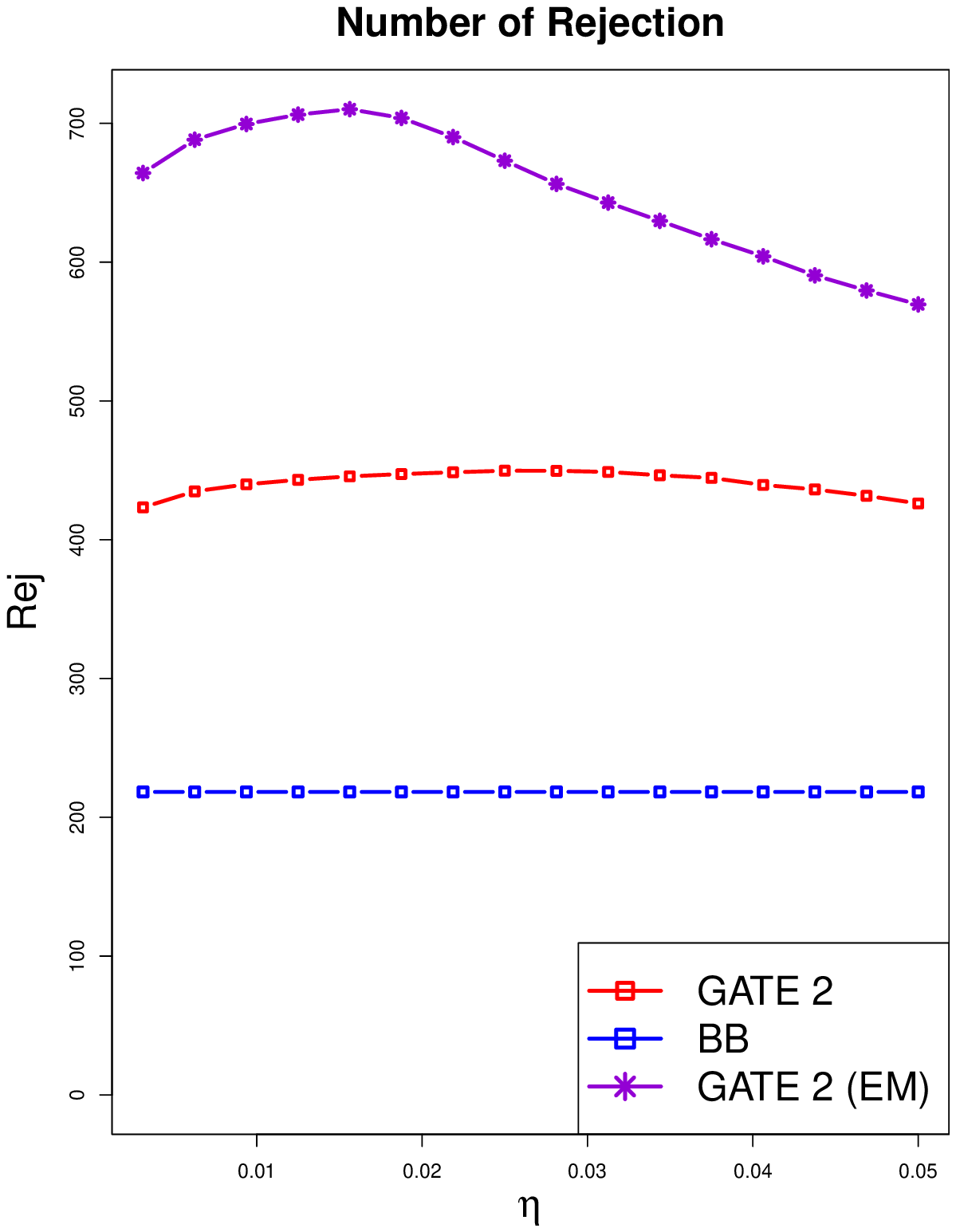}\\
   \includegraphics[height=40mm,width=40mm]{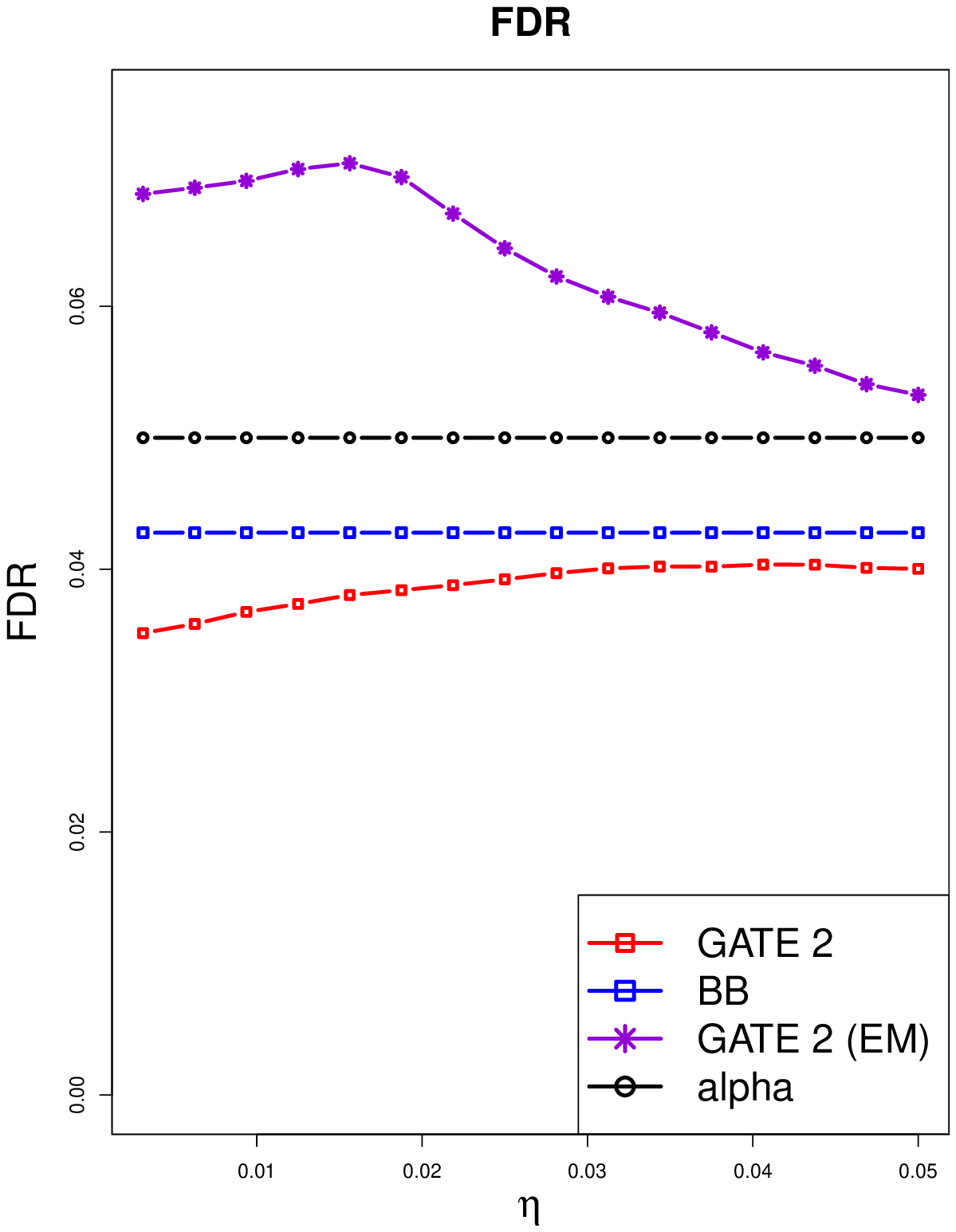}
   \includegraphics[height=40mm,width=40mm]{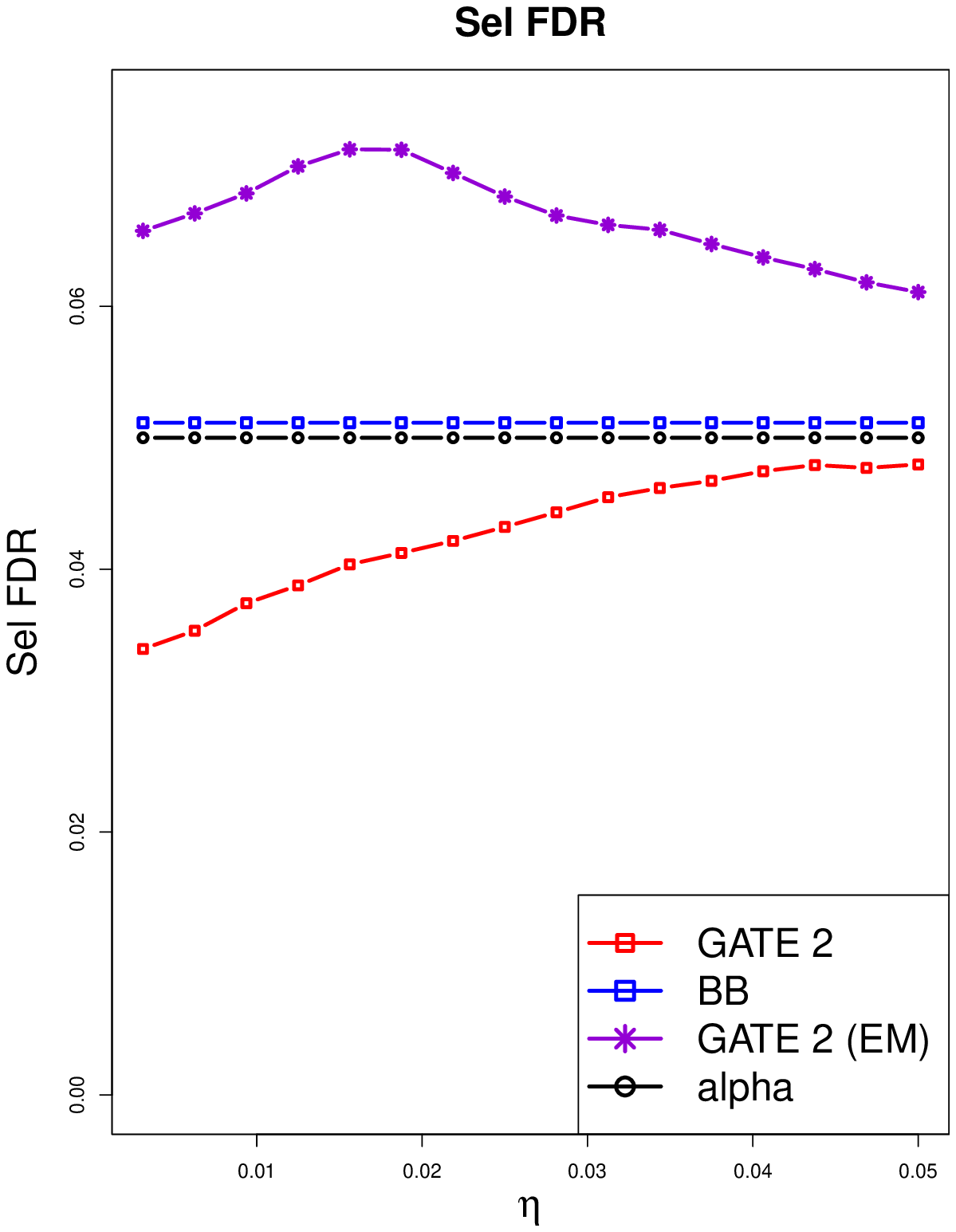}
   \includegraphics[height=40mm,width=40mm]{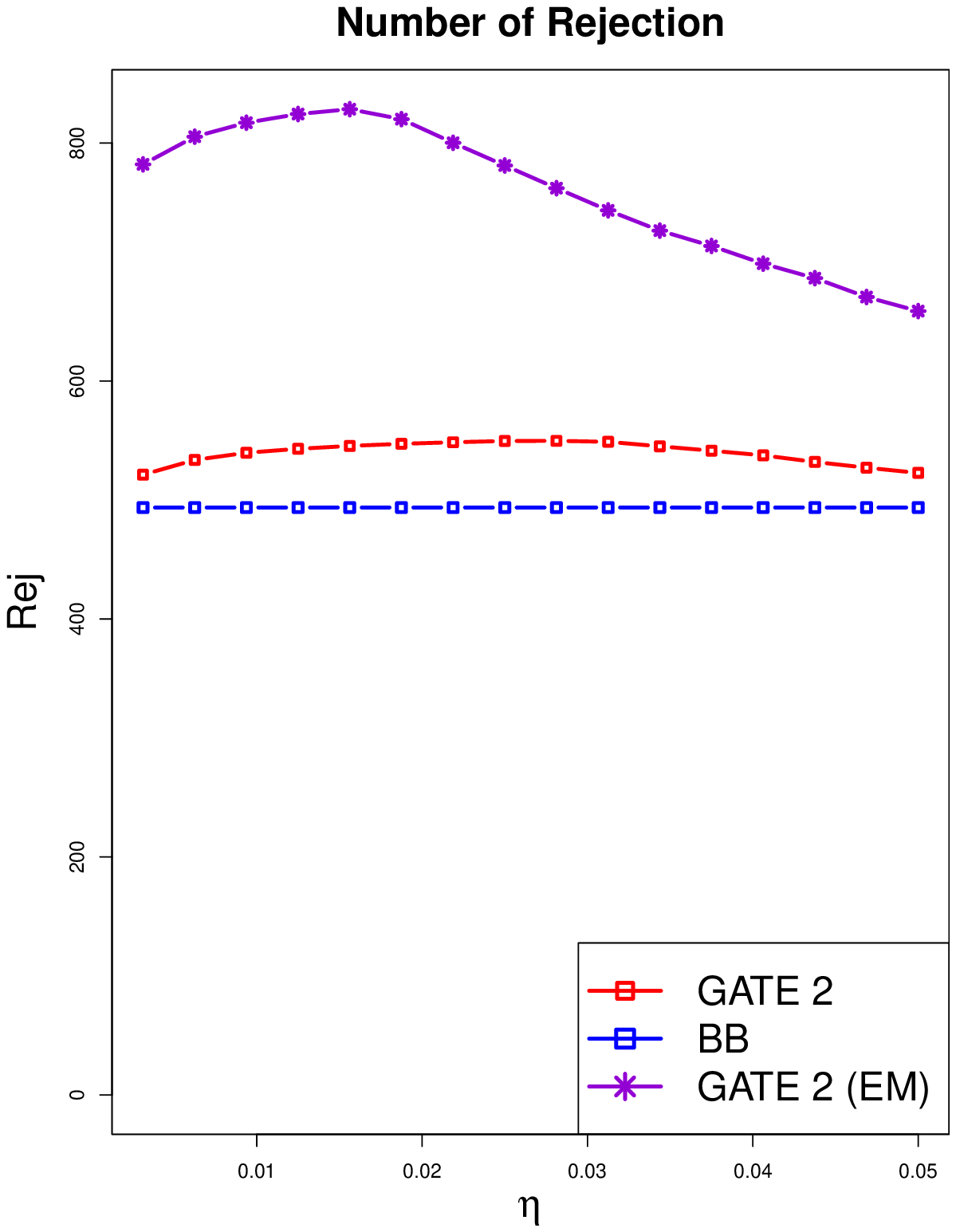}
   \caption{Performance of One-Way GATE 2 when $m=500, n=20, \pi_{1}=0.95$. The top panels correspond to cases when $K=1$ and the bottoms ones correspond to $K=2$.%The four panels correspond to four choices of $\pi_1$, which are 0.05, 0.3, 0.6, and 0.95 respectively.
   }\label{fig:gate4:d}
\end{figure}

\subsection{Real Data Analysis}\label{sec:realdata}
In this section, we revisit the Adequate Yearly Progress (AYP) study of California elementary schools in 2013. This data set has been analyzed in \cite{Liu:Sarkar:Zhao:2016}. The processed data can be loaded from the R package {\it GroupTest} and the reader can read \cite{Liu:Sarkar:Zhao:2016} for a complete description of how the data is prepared. In the final data set, we have 4118 elementary schools and 701 qualified school districts. For each school, there is an associated $z$-statistics, a quantity to compare the success rates in math exams of sociaeconomically advantaged students and sociaeconomically disadvantaged students.

We fit the Bayesian hierarchical model with $K=2$ and $\sigma^2=1$. We ran the Gibbs sampling 20,000 times and chose the first 10,000 as burn-in and select 1 out of every 20 in the remaining sequence. We ran the MCMC algorithm three times with different initial points. Eventually, the estimated parameters are
\[
\hat{\pi}_1= 0.53, \hat{\pi}_2= 0.59, (\hat{\eta}_1,\hat{\eta}_2)=(0.22,0.78), (\hat{\mu}_1,\hat{\mu}_2) = (2.64, -1.88).
\]
Figure \ref{fig:density}, which compares the kernel density estimate based on the available z-values for all the schools with the same based on the z-values produced by the data generated using the Bayes hierarchical model, seems to indicate that this model fits the AYP data quite well.  

\begin{figure}[H]
   \centering
   \includegraphics[height=75mm,width=75mm]{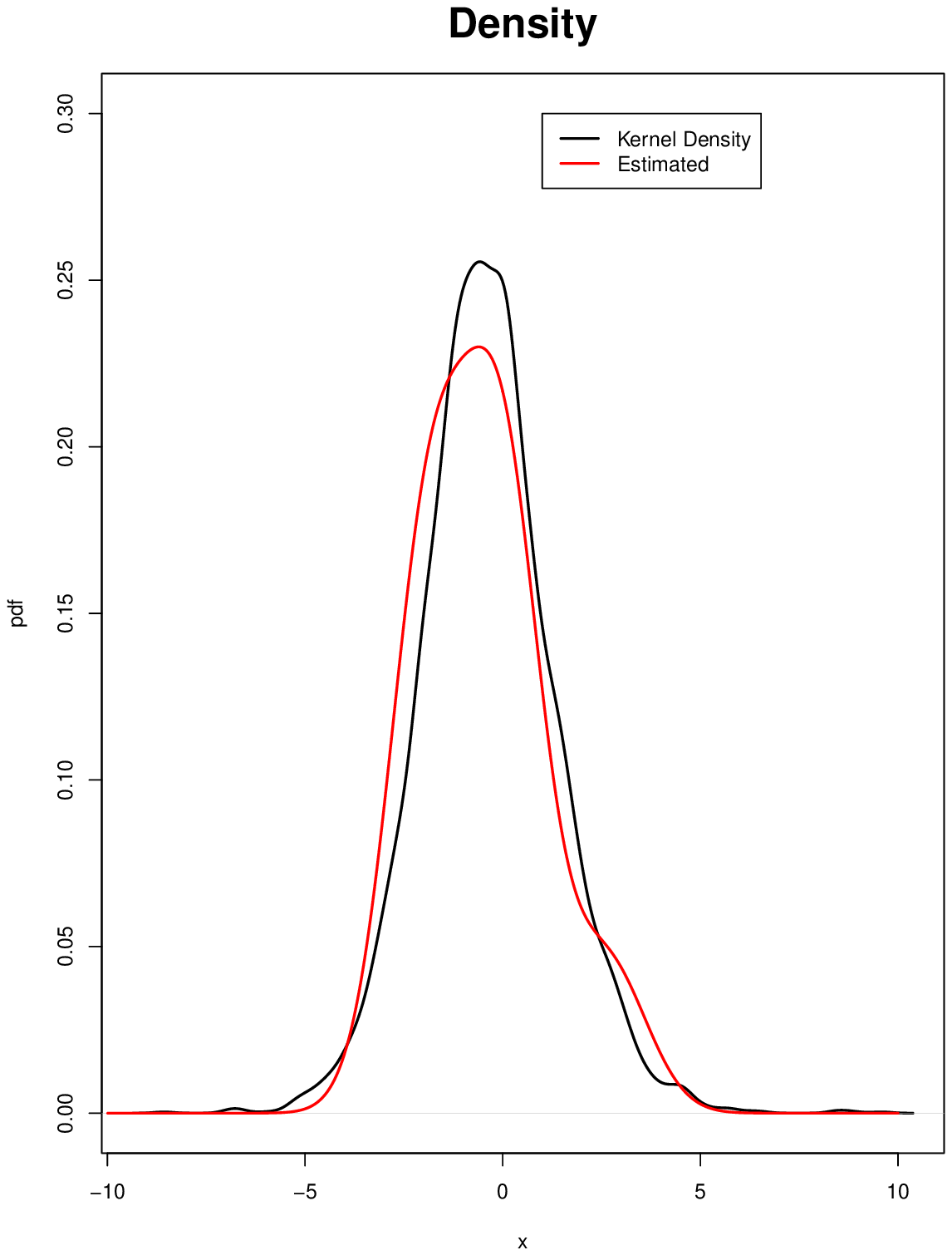}
   \caption{Comparison of the kernel density estimation of the pdf based on the z-statistics of all the data and the density function based on the estimated parameters.
   }\label{fig:density}
\end{figure}

Based on these estimated values, we applied One-Way GATE 1 with $\alpha=0.05$. It rejected 773 schools, located in 209 school districts. The GATE 1 (EM) yielded 687 schools being rejected. The GBH (TST) Method rejected 588 schools. The data-driven version of SC Method yielded 929 rejections. It should be pointed out that  $\lambda_i$ ranges from 0 to 0.78, depending on the number of schools in district $i$, which varies between 1 and 277.

The GATE 1 has the following two features: (i) incorporating the groups can help to allocate errors among the groups appropriately; (ii) group structure can alter the relative importance of different schools and thus change the ranking of the schools. These two features are playing important roles in the analysis of AYP data, which can be seen from the results for the following two districts: 
\begin{verbatim}
New Haven Unified, Guy Jr. Emanuele Elementary: z=3.05
-0.45, -0.38, -0.28, 0.83, 0.92, 3.05, 0.27

Berkeley Unified, Oxford Elementary: z=2.65.
1.77, 3.60, 4.40, -0.14, 1.83, 2.52, 2.65, 1.18, 3.25, 1.41, 0.86
\end{verbatim}

The Emanuele Elementary school is not rejected, while the Oxford Elementary school is rejected, though the corresponding z-values are 3.05 and 2.65 respectively. The New Haven Unified school district is not rejected because most of the schools has a moderate statistic except for the Guy Jr. Emanuele Elementary. On the other hand, the Berkeley Unified school district is rejected due to the fact that many schools within this district have relatively large z-values.

To cross-validate this conclusion, we downloaded the AYP data for the year of 2015 and look into these two school districts. We get the following information:
\begin{verbatim}
New Haven Unified, Guy Jr. Emanuele Elementary: z=-1.45
-1.27, 0.67, -1.51, 0.89, 3.05, -1.45, 1.18

Berkeley Unified, Oxford Elementary: z=3.95
1.49, 3.90, 3.69, 4.50, 2.35, 5.63, 3.95, 4.71, 7.64, 3.27, 2.98
\end{verbatim}
It is seen that our conclusion made based on the data of 2013 agrees with the data in 2015. In Figure \ref{fig:realdata}, we plot the z-values for the AYP data in 2015 against that for the data in 2013. We excluded those schools with the absolute $z$-values in 2013 being less than or equal to 2 because all the methods fail to reject them. We also exclude the schools with the absolute $z$-values being greater than 5 for clear rejection. The four panels in this figure correspond to One-Way GATE 1, One-Way GATE 1 (EM), GBH (TST) and SC from the top-left to the bottom-right. In each panel, the red dots represent the schools which are not rejected based on the data in 2013 and the blue triangles represent the schools which are rejected.

Note that the rejection region of GBH's method is roughly symmetric around zero. However, the rejection regions of SC and GATE methods are not. Due to the effect of the group structure, whether to reject a hypothesis depends on the magnitude of z-statistic and the z-statistic of all the schools within the same school district. This is a unique feature of GATE which could provide some new insights for the local government.

The code of the numerical investigations is made available via github. The link is\\
\url{https://github.com/zhaozhg81/GATE}

\begin{figure}[H]
	\centering
	\includegraphics[width=50mm, height=50mm]{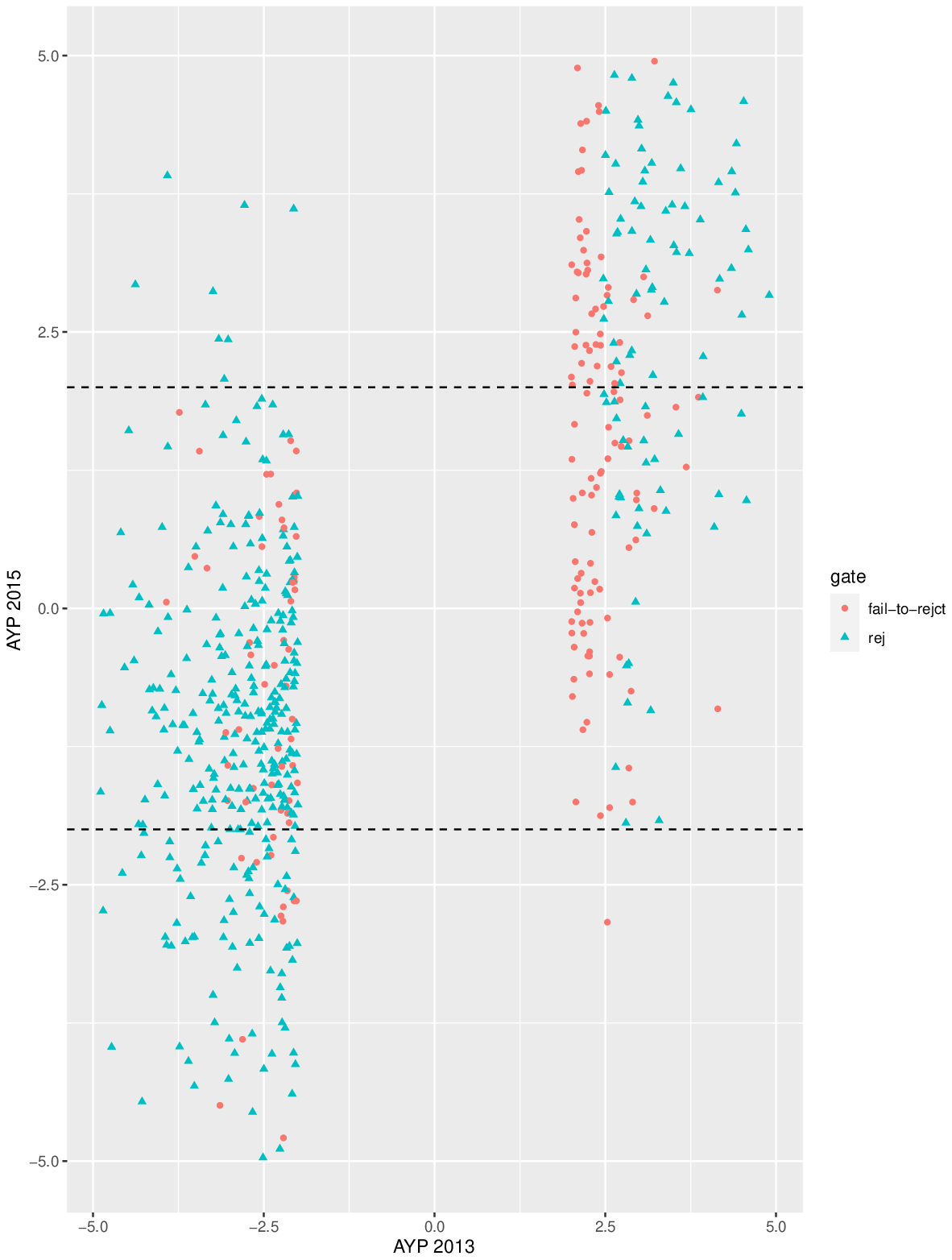}		
	\includegraphics[width=50mm, height=50mm]{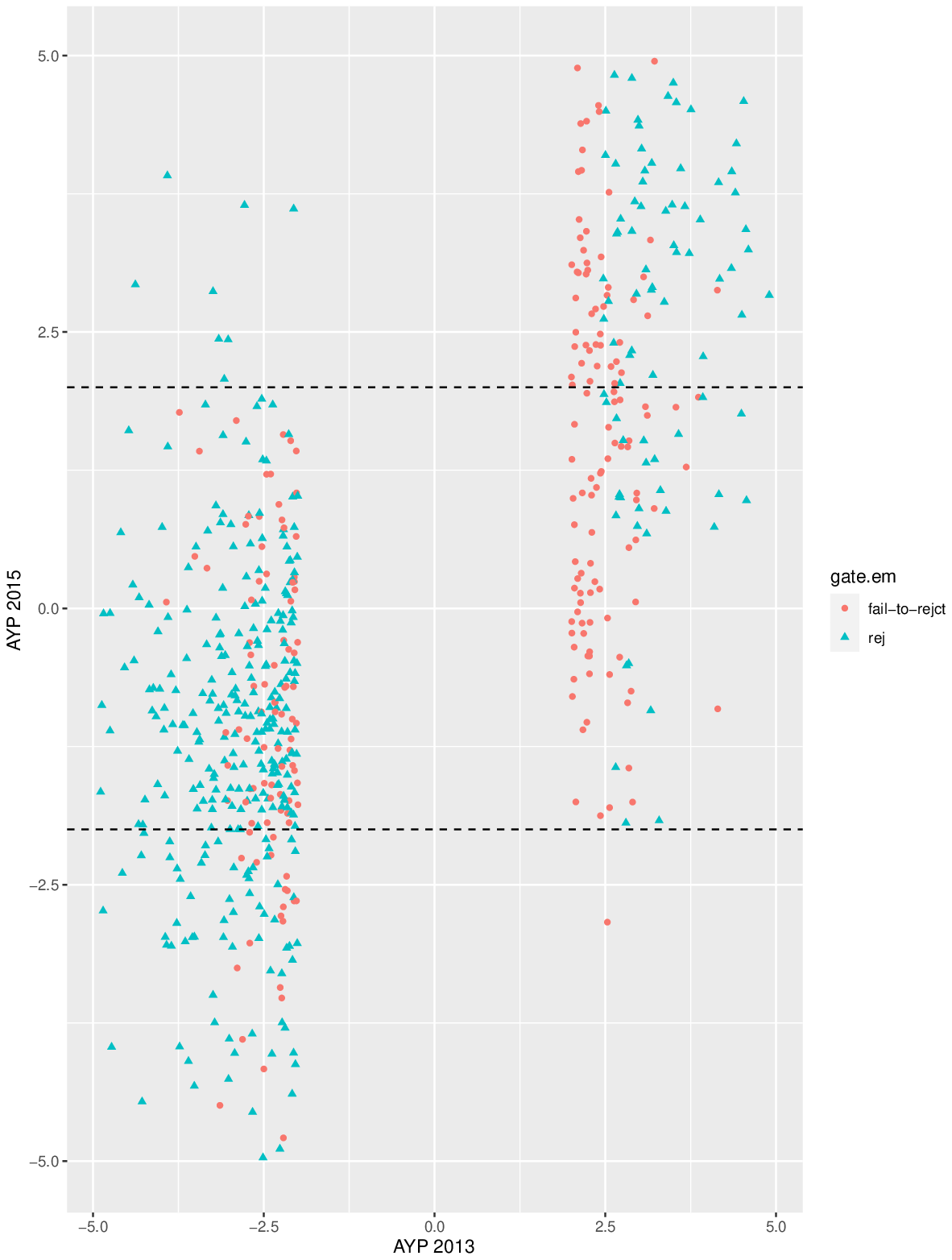}\\
	\includegraphics[width=50mm, height=50mm]{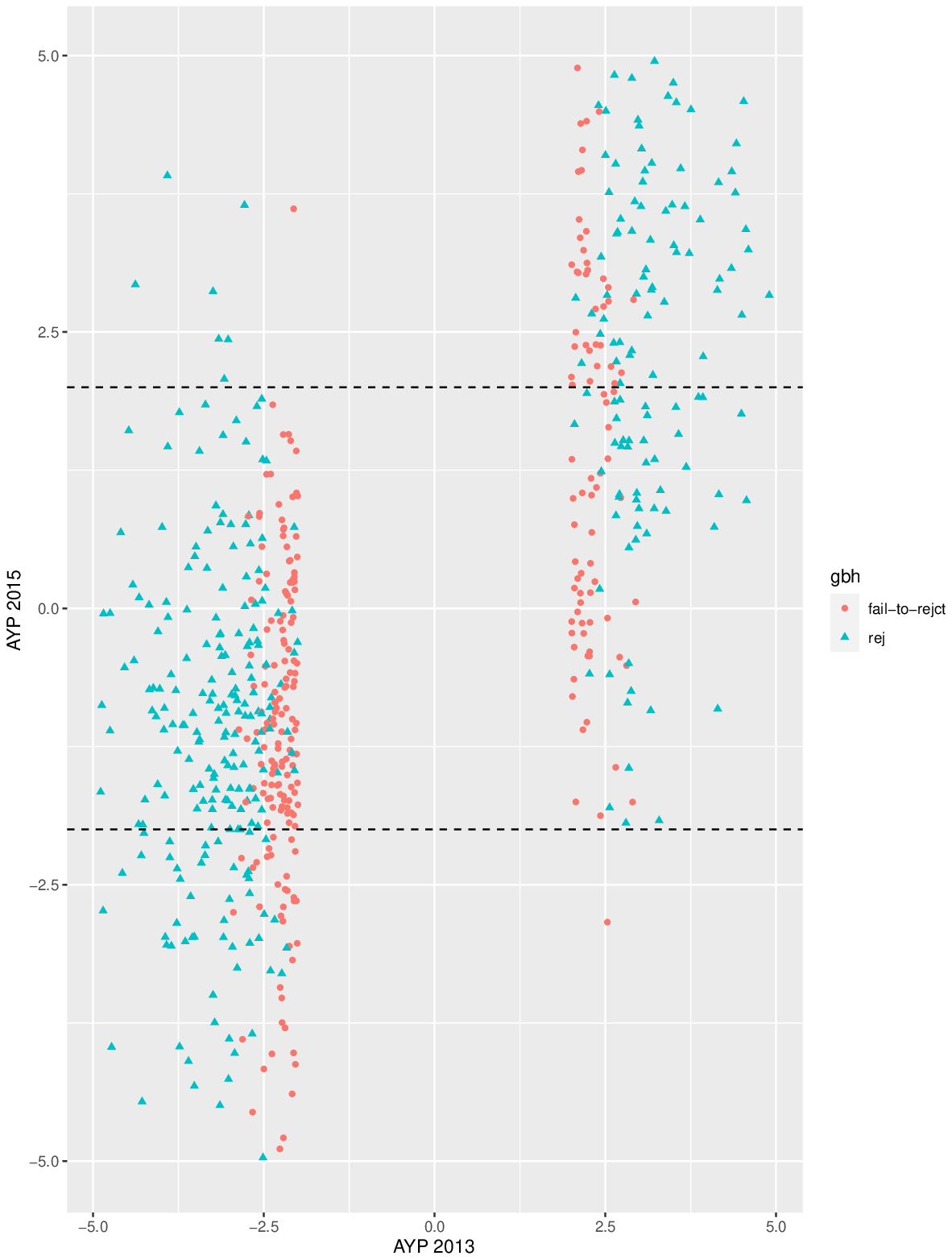}
        \includegraphics[width=50mm, height=50mm]{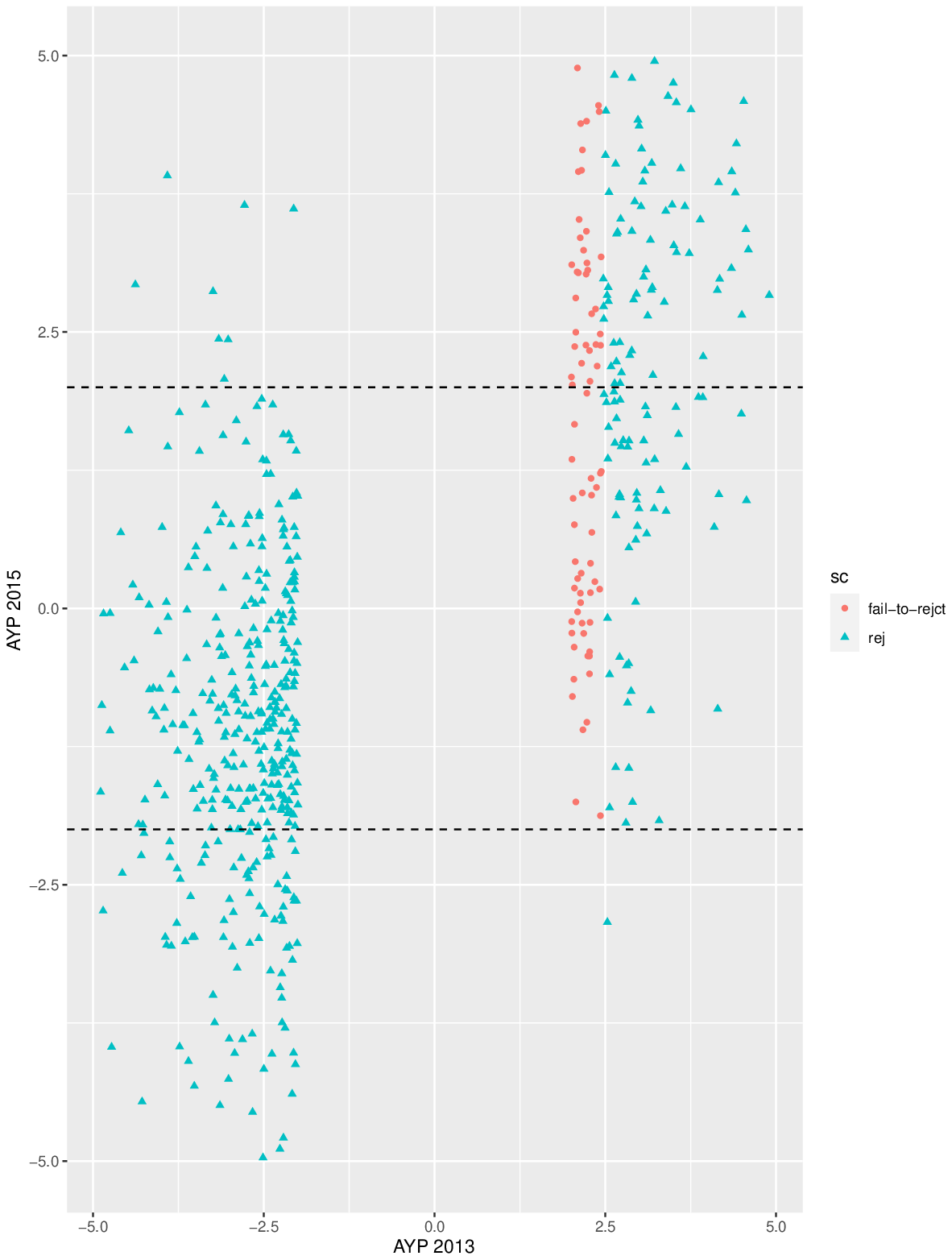}
	\caption{AYP data analysis result. Four panels, from top-left to the bottom-right, correspond to: One-Way GATE 1, One-Way GATE 1 (EM), GBH (TST) Method and SC Method.}\label{fig:realdata}
\end{figure}	

\section{Concluding Remarks}\label{sec:concluding}
The primary focus of this article has been to continue the line of research in \cite {Liu:Sarkar:Zhao:2016} to answer Q1 and Q2 for one-way classified hypotheses. It can provide the ground work for our broader goal of answering these questions in the setting of two-way classified hypotheses. Two-way classified setting is seen to occur in many applications. For instance, in time-course microarray experiment (see, e.g., \cite{Storey:etal:2005, Yuan:Kendziorski:2006, Sun:Wei:2011}), the hypotheses of interest can be laid out in a two-way classified form with `gene' and `time-point' representing the two categories of classification. In multi-phenotype GWAS (\cite{Peterson:etal:2016, Segura:etal:2012}), the families of the hypotheses related to different phenotypes form one level of grouping, while the other level of grouping is formed by the families of hypotheses corresponding to different SNPs. Two-way classified structure of hypotheses occurs also in brain imaging studies (\cite{Liu:etal:2009, Stein:etal:2010, Lin:Calhoun:Wang:2014, Barber:Ramdas:2015}). 
Now that we know the theoretical framework that can successfully capture the underlying group effect and yield powerful approaches to multiple testing in one-way classified setting, proceeding further towards extending it to produce newer and powerful Lfdr based approaches answering Q1 and Q2 in two-way classified setting would be a worthwhile goal.  Some initial progress towards this goal has been made in \cite{Nandi:Sarkar:2018, sarkar2021development}. We intend to expand upon it in our future research.  

We strongly recommend using the One-Way GATE 1 we have proposed in this article for answering Q1, motivated by its desirable theoretical properties (as stated in Theorems 1 and 2) and strong numerical findings in comparison with its natural competitors (discussed in Sections 4.2 and 4.4).  For answering Q2, the proposed One-Way GATE 2 is a better alternative to the BB method when only a few of the selected groups are likely to be important, which happens in many applications.    

\section{Acknowledgment}
The authors greatly appreciate valuable comments from the reviewers.

\setcounter{section}{0}
\setcounter{subsection}{0}
\renewcommand\thesection{\Alph{section}}
\numberwithin{equation}{section}

\section{Appendix}\label{sec:appendix}

\subsection{Proofs of (\ref{eqn:Lfdr:ji}) and (\ref{eqn:Lfdr:idot})} These results, although appeared before in \cite{Liu:Sarkar:Zhao:2016}, will be proved here using different and simpler arguments. They are re-stated, without any loss of generality, for a single group with slightly different notations in the following lemma.

  \begin{lemma}\label{lemma:a1}
    Conditionally given $\theta_{\centerdot} \sim Bern(\pi_1)$, let  $(X_{j}, \theta_{j})$, $j=1, \ldots, n$, be distributed as follows: (i) $X_1, \ldots X_n~|~\theta_1, \ldots, \theta_n \stackrel {ind} \sim (1- \theta_{\centerdot} \cdot \theta_{j})f_{0}(x_j)+ \theta_{\centerdot} \cdot \theta_{j}f_{1}(x_{j})$, and (ii) $\theta_{1}, \ldots, \theta_{n} \sim TPBern(\pi_{2}, n)$. Let Lfdr$_{j}(\pi_2) \equiv Lfdr(x_j;\pi_2) = (1-\pi_2)f_0(x_j)/m(x_j)$, with $m(x)=(1-\pi_2)f_0(x)+\pi_2f_{1}(x)$, for $j=1, \ldots, n$, and Lfdr$_{\centerdot}(\pi_2) = \prod_{j=1}^n  Lfdr_{j}(\pi_2)$. Then,
    \begin{eqnarray}\label{eqn:A1}
      Pr (\theta_j=0|\theta_{\centerdot}=1, X_1=x_1, \ldots, X_n=x_n ) = \frac{Lfdr_{j}(\pi_2) - Lfdr_{\centerdot}(\pi_2)}{1- Lfdr_{\centerdot}(\pi_2)}
    \end{eqnarray}
    and
    \begin {eqnarray}\label{eqn:A2}
      Pr (\theta_{\centerdot}=0|X_1=x_1, \ldots, X_n=x_n ) = \frac{Lfdr_{\centerdot}(\pi_2)}{Lfdr_{\centerdot}(\pi_2) + \lambda [1- Lfdr_{\centerdot}(\pi_2)]},
    \end {eqnarray}
    where $\lambda= \frac{\pi_1}{1-\pi_1} \div \frac{1 -(1-\pi_2)^n}{(1-\pi_2)^n}$.
  \end{lemma}

  \begin{proof}
    First, note that
    \begin{eqnarray}\label{eqn:A3}
      & & (X_1, \ldots, X_n) | \theta_{\centerdot} = 0 \sim  \prod_{j=1}^n f_0(x_j) = \frac{\prod_{j=1}^n m(x_j)}{(1-\pi_2)^n} {\rm Lfdr}_{\centerdot}(\pi_2),
    \end {eqnarray}
    and
    \begin {eqnarray}\label{eqn:A4}
      & & (X_1, \ldots, X_n) | \theta_{\centerdot} = 1 \nonumber \\ & \sim & \frac{1}{1-(1-\pi_2)^n} \sum_{\sum_{j=1}^n \theta_j >0}\left [ \prod_{j=1}^n \{(1-\theta_j)f_0(x_j)+\theta_j f_1(x_j)\} \prod_{j=1}^n \{\pi_2^{\theta_j}(1-\pi_2)^{1-\theta_j} \}\right ] \nonumber \\
& = & \frac{1}{1-(1-\pi_2)^n} \left [ \prod_{j=1}^n m(x_j) - (1-\pi_2)^n \prod_{j=1}^n f_0(x_j) \right ] \nonumber \\
& = & \frac{\prod_{j=1}^n m(x_j)}{1-(1-\pi_2)^n} \left [1- {\rm Lfdr}_{\centerdot}(\pi_2) \right ],
    \end {eqnarray}
    from which we get
    \begin {eqnarray}\label{eqn:A5}
      (X_1, \ldots, X_n) \sim \left \{ \frac{(1-\pi_1){\rm Lfdr}_{\centerdot}(\pi_2)}{(1-\pi_2)^n}  +\frac{\pi_1 [1-{\rm Lfdr}_{\centerdot}(\pi_2)]}{1-(1-\pi_2)^n}  \right \} \prod_{j=1}^n m(x_j).
    \end {eqnarray}
    Formula (\ref{eqn:A2}) follows upon dividing $(1-\pi_1)$ times (\ref{eqn:A3}) by (\ref{eqn:A5}).

When $\theta_j=0$, the conditional distribution of $X_1, \ldots, X_n$ given $\theta_{\centerdot} = 1$ can be obtained similar to that in (\ref{eqn:A4}) as follows:
\begin {eqnarray}\label{eqn:A6}
  & & \frac{(1-\pi_2)f_0(x_j)}{1-(1-\pi_2)^n} \sum_{\sum_{k(\neq j)=1}^n \theta_k >0}  [ \prod_{k(\neq j)=1}^n \{(1-\theta_k)f_0(x_k)+\theta_k f_1(x_k)\} \nonumber\\
  &\times & \prod_{k(\neq j)=1}^n \{\pi_2^{\theta_k}(1-\pi_2)^{1-\theta_k} \} ]  \\
& = & \frac{(1-\pi_2)f_0(x_j)}{1-(1-\pi_2)^n} \left [ \prod_{k(\neq j)=1}^n m(x_k) - (1-\pi_2)^{n-1} \prod_{k(\neq j)=1}^n f_0(x_k) \right ] \nonumber \\
& = & \frac{\prod_{j=1}^n m(x_j)}{1-(1-\pi_2)^n} ) \left [{\rm Lfdr}_{j}(\pi_2)- {\rm Lfdr}_{\centerdot}(\pi_2) \right ].
\end {eqnarray}
Formula (\ref{eqn:A1}) then follows upon dividing (\ref{eqn:A6}) by (\ref{eqn:A4}).

\end{proof}

\noindent {\it Proof of Theorem \ref{thm:1:optimal}.}
For notational simplicity, we will hide $\vX$ in $\delta_{ij}(\vX)$, $\delta_{ij}'(\vX)$, $Lfdr_{ij}(\vX)$.

First, we note the following inequalities:
\begin{equation}\label{eq:proof:1}
  \alpha \sum_{ij} \left( \delta_{ij}-\delta_{ij}' \right) \le \sum_{ij} \left ( \delta_{ij}-\delta_{ij}' \right ) Lfdr_{ij} \le c \sum_{ij} \left( \delta_{ij}-\delta_{ij}' \right),
\end{equation}
the first of which follows from the fact that the PFDR$_T (\delta') \le \alpha = PFDR_T (\delta)$, while the second one follows from $\sum_{ij} \left( \delta_{ij}-\delta_{ij}' \right)\left( c - Lfdr_{ij} \right) \ge 0$, because of the definition of $\delta_{ij}$.

 Since $\alpha = \sum_{ij}\delta_{ij}Lfdr_{ij}/ \max \{\sum_{ij}\delta_{ij}, 1 \} \le c$, we have from (\ref {eq:proof:1}) that $\sum_{ij}(\delta_{ij}-\delta_{ij}')Lfdr_{ij} \ge 0$, that is,
 \begin{equation}\label{eq:proof:2}
  \sum_{ij}(1-\delta_{ij})Lfdr_{ij} \le \sum_{ij}(1-\delta_{ij}')Lfdr_{ij}.
  \end {equation}
  With ${\rm PFNR}_T(\delta)$ and ${\rm PFNR}_T(\delta')$ denoting the PFNR$_T$ of $\delta$ and $\delta'$, respectively, we now note that
    \begin{eqnarray}
  & & c \left [ \frac {{\rm PFNR}_T(\delta)}{1-{\rm PFNR}_T(\delta)} - \frac {{\rm PFNR}_T(\delta')}{1-{\rm PFNR}_T(\delta')} \right ] \nonumber \\ & = &
  c \sum_{ij} \left[ \frac{(1-\delta_{ij})(1-Lfdr_{ij})}{ \sum_{ij}(1-\delta_{ij})Lfdr_{ij}} - \frac{(1-\delta_{ij}')(1-Lfdr_{ij}) }{\sum_{ij}(1-\delta_{ij}')Lfdr_{ij}}\right] \nonumber \\
  & = & \sum_{ij} \left[ \frac{1-\delta_{ij}}{ \sum_{ij}(1-\delta_{ij})Lfdr_{ij}} - \frac{1-\delta_{ij}'}{\sum_{ij}(1-\delta_{ij}')Lfdr_{ij}}\right] \left [ c(1-Lfdr_{ij})-(1-c)Lfdr_{ij} \right ] \nonumber \\ & \le & 0, \nonumber
  \end {eqnarray}
  with the inequality holding due to the definition of $\delta_{ij}$ and the inequality in (\ref{eq:proof:2}). Thus,  we have  \[
  PFNR_T(\delta) \le PFNR_T(\delta'),\] which proves the theorem.

\bibliographystyle{plainnat}
\bibliography{zhaozhg}

\newpage
\subsection{EM Algorithm}
In this section, we provide steps of the EM algorithm.
%\subsection{}
To better present the result, define $\pi_1^1=\pi_1, \pi_1^0=1-\pi_1, \pi_{2}^1=\pi_{2}$ and $\pi_{2}^0=1-\pi_{2}$. Let $A_i=\frac{1}{1-(\pi_2^0)^{n_i}}$. Consider $(\vx, \btheta)$ as the complete data. Then the complete log-likelihood function can be written as:
\begin{align*}
&l(\vx,\btheta)\\
&=\sum_i \sum_{l=0}^{1} I(\theta_i=l)(log\pi_{1}^{l}+logf(\vx_{i}|\theta_i=l))\\
&=\sum_i\left\{I(\theta_i=0)\left[log\pi_{1}^{0}+\sum_{j=1}^{n_i}logf(x_{ij}|\theta_i=0)\right]+I(\theta_i=1)\left[log\pi_{1}^{1}+logf(\vx_{i}|\theta_i=1)\right]\right\}\\
&=\sum_i\sum_{l=0}^{1}I(\theta_{i}=l)log\pi_{1}^{l}\\
&+\sum_i\sum_{j=1}^{n_i}\sum_{k=1}^{K}\left[ I(\theta_i=1,\theta_{j|i}=1,m_{j|i}=k)log( c_k \pi_{2}^1/A_i) + I(\theta_i=1,\theta_{j|i}=0)log( \pi_{2}^0/A_i) \right] \\
&+\sum_{i}\left[I(\theta_{i}=0)\sum_{j=1}^{n_i}logf_{0}(x_{ij})+I(\theta_{i}=1)\sum_{j=1}^{n_i}I(\theta_{j|i}=0)logf_{0}(x_{ij})\right]\\
&+\sum_iI(\theta_{i}=1)\sum_{j=1}^{n_i}\sum_{k=1}^{K}I(\theta_{j|i}=1,m_{j|i}=k)logf_{k}(x_{ij}|\theta_{j|i}=1),\\
\end{align*}
where  $m_{j|i}=k$ implies that $x_{ij}$ is generated from $N(\mu_k,\sigma_k^2)$.

The expected value of the complete-data log-likelihood $l(\vx,\btheta)$ with respect to the unknown $\theta_i,\theta_{j|i}$, given the observed data $\vx$ and the current value $\bbeta'$ of the parameter is:
\begin{align*}
&Q(\bbeta,\bbeta')=E\left[ l (\vx,\btheta)|\vx,\bbeta'\right]\\
&=\sum_i\sum_{l=0}^{1}log\pi_{1}^{l}P(\theta_{i}=l|\vx,\bbeta')\\
%%&\qquad{}+\sum_{g}\sum_{j=1}^{n_i}\sum_{h=0}^{1}\sum_{l=1}^{L}log \pi_{2|1}^{hl}P(\theta_{i}=1,\theta_{j|i}=h,m_{j|i}=l|\vx,\bbeta')\\
&\qquad{}+ \sum_i\sum_{j=1}^{n_i}\sum_{l=0}^1\log(\pi_{2}^l/A_i)P(\theta_i=1,\theta_{j|i}=l|\vx,\bbeta') \\
&\qquad{}+ \sum_{i}\sum_{j=1}^{n_i}\sum_{k=1}^{K}logc_kP(\theta_{i}=1,\theta_{j|i}=1,m_{j|i}=k|\vx,\bbeta')\\
&\qquad{}+\sum_{i}\sum_{j=1}^{n_i}logf_{0}(x_{ij})P(\theta_{i}=0|\vx,\bbeta')+\sum_{g}\sum_{j=1}^{n_i}logf_{0}(x_{ij})P(\theta_{i}=1|\vx,\bbeta')\\
&\qquad{}+\sum_i\sum_{j=1}^{n_i}\sum_{k=1}^{K}logf_{k}(x_{ij})P(\theta_{i}=1,\theta_{j|i}=1,m_{j|i}=k|\vx,\bbeta').\\
\end{align*}

Note that
\[
P(\theta_i=1, \theta_{j|i}=1|\vx, \beta' ) = (1- fdr_{j|i}(\bbeta'))(1-fdr_i(\bbeta')),
\]
and
\[
P(\theta_{i}=1,\theta_{j|i}=1,m_{j|i}=k|\vx,\bbeta') = P(\theta_i=1, \theta_{j|i}=1|\vx, \beta' ) * \frac{ c_k\frac{1}{\sigma_k}\phi(\frac{x_{ij}-\mu_k}{\sigma_{k}}) }{ \sum_{k=1}^K  c_k\frac{1}{\sigma_k}\phi(\frac{x_{ij}-\mu_k}{\sigma_{k}}) }.
\]

We want to maximize the $Q$ function which can be realized by maximizing each of these parts to get the estimates of $\pi_{1},\pi_{2}, c_{k}$ and $\mu_{k}, \sigma_{k}^2$, since these parts are not related. To maximize the first part with the restriction that $\pi_1^0+\pi_1^1=1$, using the Lagrange multipliers, we can find the maximizer for $\pi_1^1$ as

\begin{equation*}
\pi_{1}^{new}=1-\frac{\sum_iP(\theta_{i}=0|\vx,\bbeta')}{m}=1-\frac{\sum_{i}fdr_i(\bbeta')}{m},
\end{equation*}
The parameter $\pi_2^0$ is updated as 
\[
\pi_2^{new} =argmax\left\{\pi_2^0: \sum_i\sum_{j=1}^{n_i}\sum_{l=0}^1\log(\pi_{2}^l/A_i)P(\theta_i=1,\theta_{j|i}=l|\vx,\bbeta')  \right\}.
\]
The parameter $c_k$ is updated as 
\begin{align*}
%%&\pi_{2}^{new}=\frac{\sum_g\sum_{j=1}^{m_g}P(\theta_{j|i}=1,\theta_{i}=1|\vx,\bbeta')}{\sum_g\sum_{j=1}^{m_g}P(\theta_{i}=1|\vx,\bbeta')}\\
&c_{k}^{new}=\frac{\sum_i\sum_{j=1}^{n_i}P(\theta_{j|i}=1,\theta_{i}=1,m_{j|i}=k|\vx,\bbeta')}{\sum_i\sum_{j=1}^{n_i}P(\theta_{i}=1,\theta_{j|i}=1|\vx,\bbeta')}.
\end{align*}
        
For the last part of $Q$ function, we know that $f_{0}(x)\sim N(0,1)$, $f_{k}(x) \sim N(\mu_{k},\sigma_{k}^2)$ with probability $c_{k}$ . Therefore, for each $k$, we need to find the MLEs for $\mu_{k}$ and $\sigma_{k}^2$ by maximizing the following log-likelihood function:
\begin{align*}
&\sum_{i}\sum_{j=1}^{n_i}\sum_{k=1}^{K}logf_{k}(x_{ij})P(\theta_{i}=1,\theta_{j|i}=1,m_{j|i}=k|\vx,\bbeta')\\
&=\sum_{i}\sum_{j=1}^{n_i}\sum_{k=1}^{K}\left[-\frac{1}{2}log\sigma_{k}^2-\frac{1}{2\sigma_{k}^2}(x_{ij}-
\mu_{k})^2\right]P(\theta_{i}=1,\theta_{j|i}=1,m_{j|i}=k|\vx,\bbeta').
\end{align*}

Taking derivatives with respect to $\mu_{l}$ and $\sigma_{1}^2$ and equating them to zero, we can get:
\begin{align*}
\mu_{k}^{new}=\frac{\sum_i\sum_{j=1}^{n_i}x_{ij}P(\theta_i=1,\theta_{j|i}=1,m_{j|i}=k|\vx,\bbeta')}{\sum_{i}\sum_{j=1}^{n_i}P(\theta_i=1,\theta_{j|i}=1,m_{j|i}=k|\vx,\bbeta')},\\
\sigma_{k}^{2new}=\frac{\sum_i\sum_{j=1}^{n_i}(x_{ij}-\mu_l)^2P(\theta_i=1,\theta_{j|i}=1,m_{j|i}=k|\vx,\bbeta')}
{\sum_i\sum_{j=1}^{n_i}P(\theta_i=1,\theta_{j|i}=1,m_{j|i}=k|\vx,\bbeta')}.\\
\end{align*}

\subsection{More simulation results on One-Way GATE 1}
In this section we list more simulation results on One-Way GATE 1.

\begin{figure}[H]
  \centering
  \includegraphics[height=50mm,width=50mm]{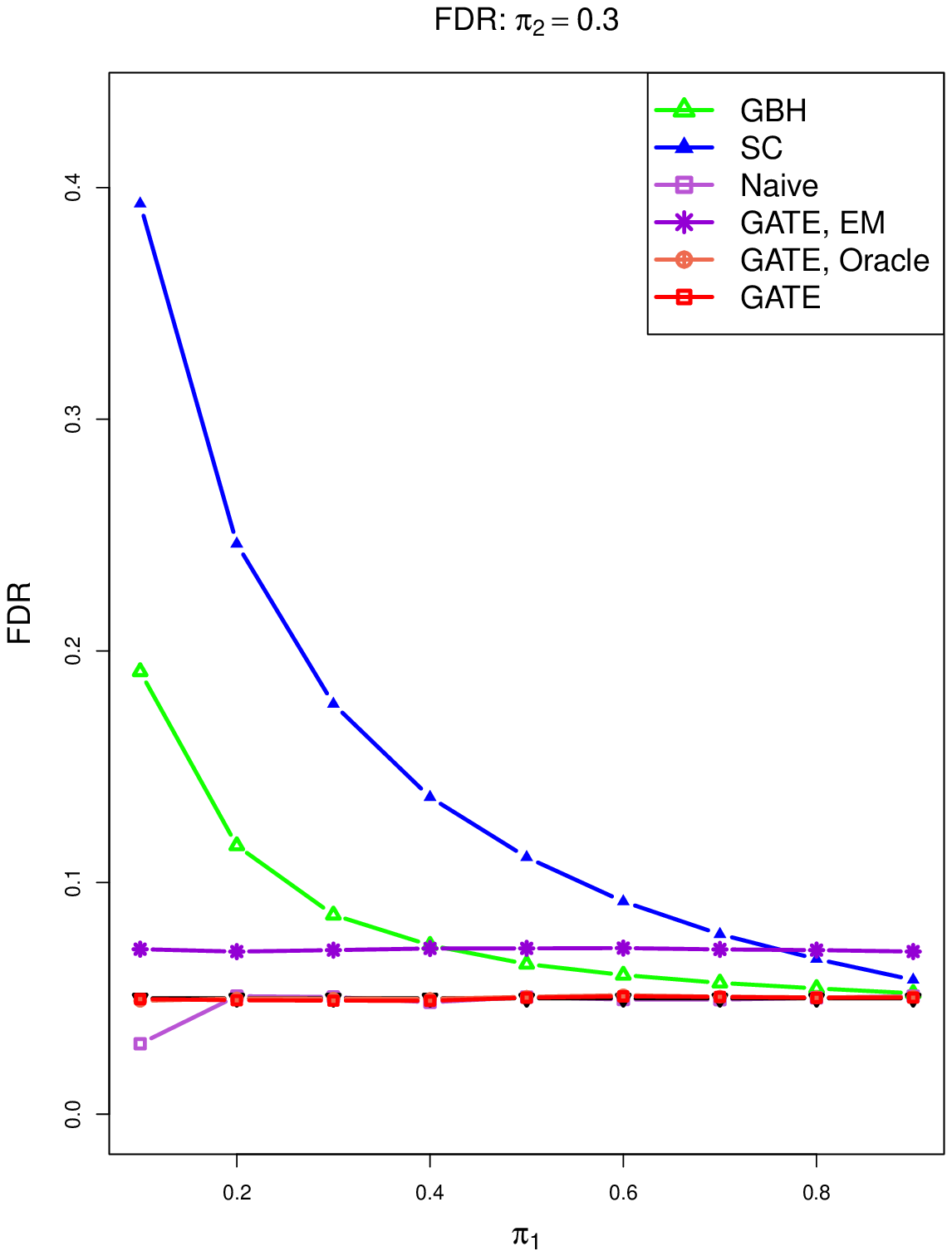}
  \includegraphics[height=50mm,width=50mm]{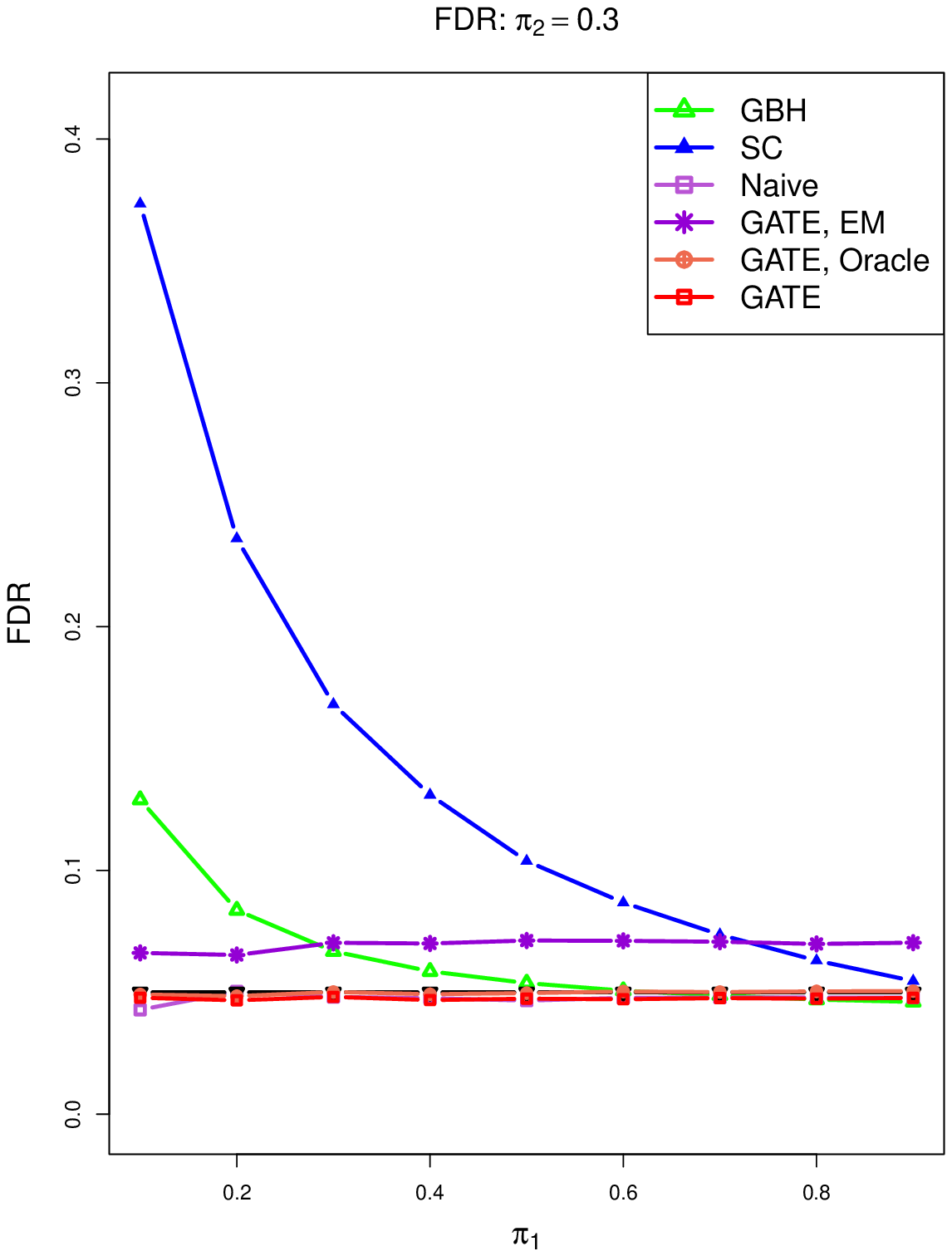}\\
    \includegraphics[height=50mm,width=50mm]{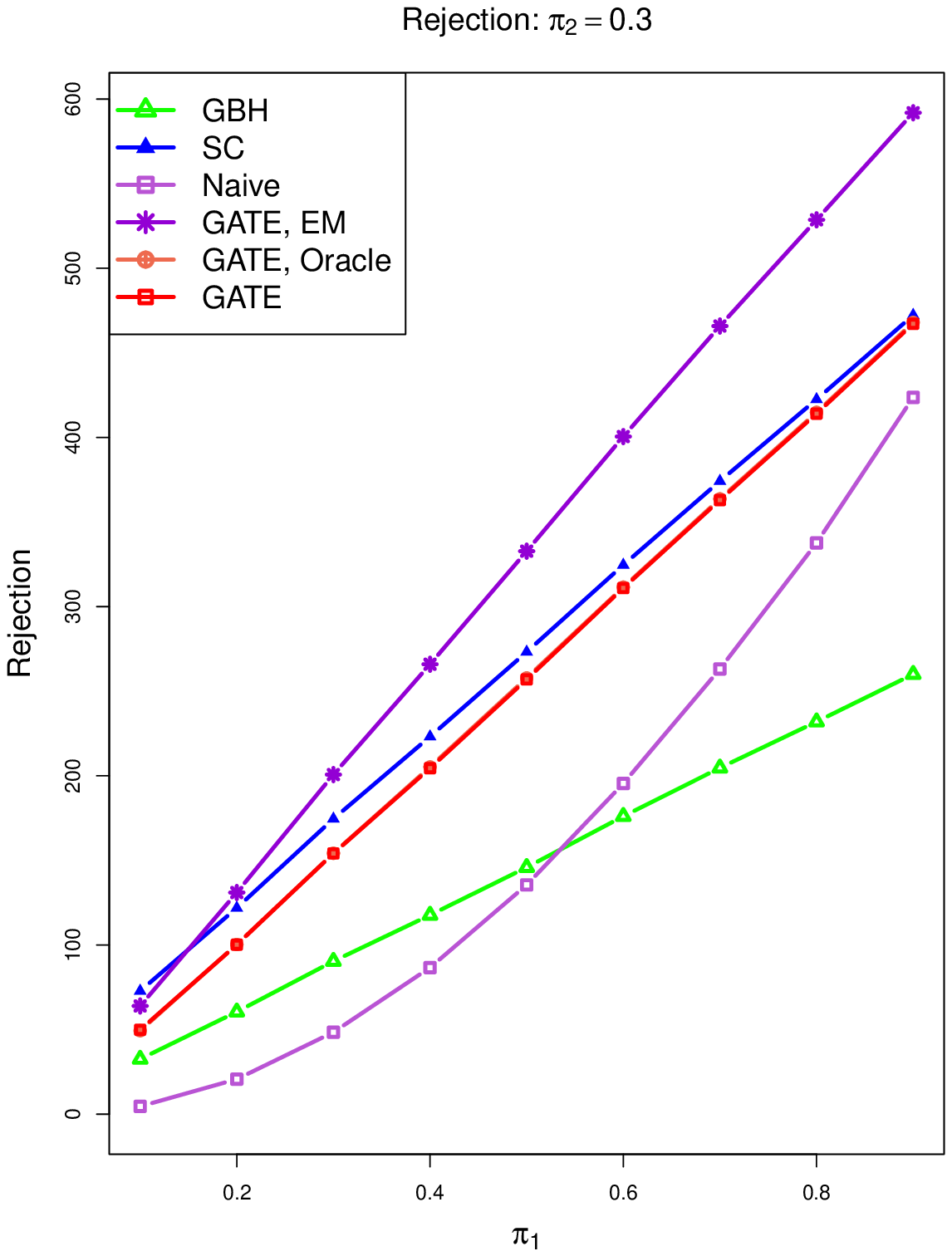}
  \includegraphics[height=50mm,width=50mm]{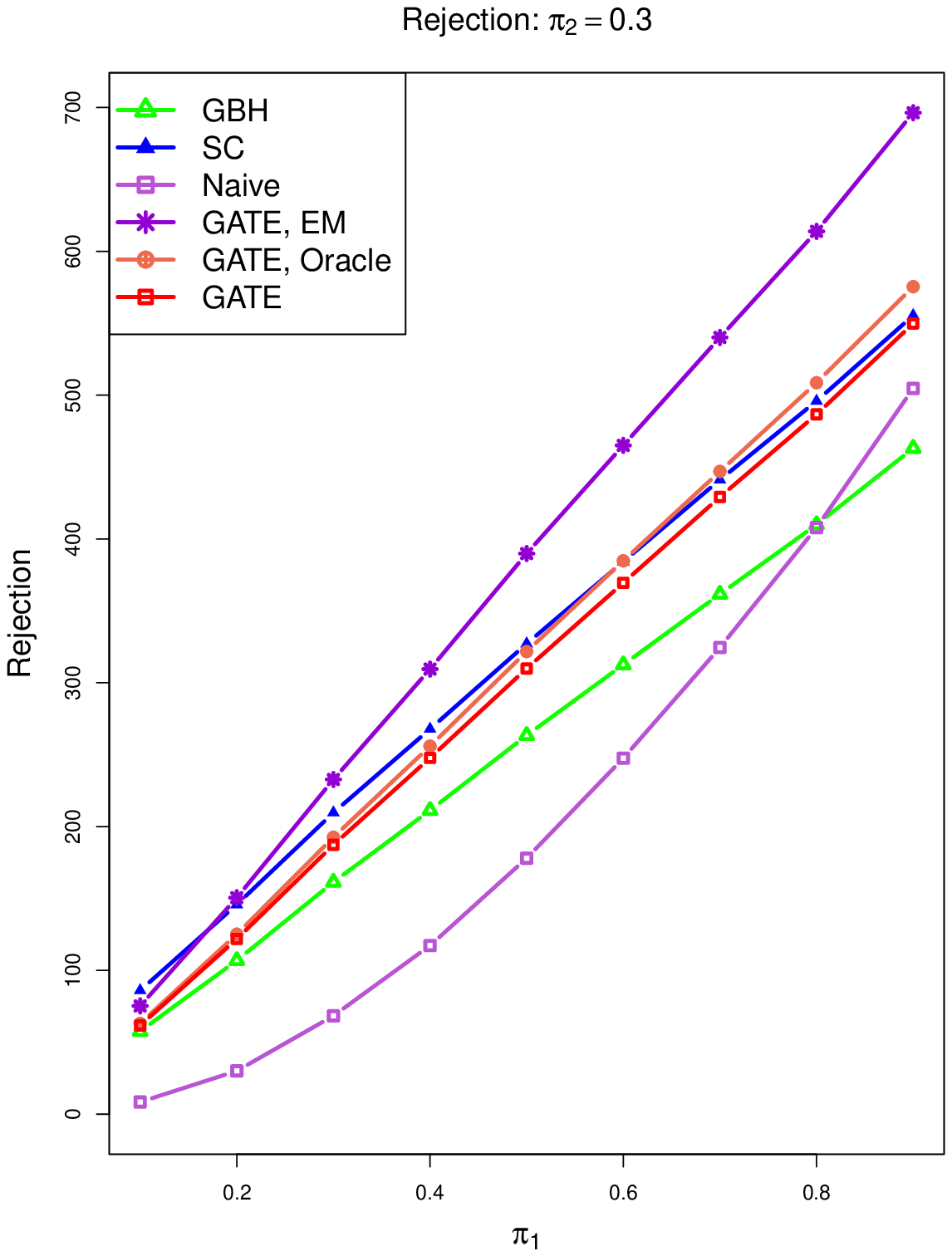}
  \caption{One-Way GATE 1: $m=100, n=50, \pi_2=0.3$. The  left and right panels correspond to $K=1$ and $2$, respectively.
  }\label{fig:gate1:fdr:s1}
\end{figure}

\iffalse 
\begin{figure}[H]
  \centering
  \includegraphics[height=50mm,width=50mm]{./figure/Bayes_rej_gate1_G_100_mg_50_L_1.eps}
  \includegraphics[height=50mm,width=50mm]{./figure/Bayes_rej_gate1_G_100_mg_50_L_2.eps}
  \caption{One-Way GATE 1: $m=100, n=50, \pi_2=0.3$. The  left and right panels correspond to $K=1$ and $2$, respectively.
  }\label{fig:gate1:fdr:s2}
\end{figure}
\fi

\begin{figure}[H]
  \centering
  \includegraphics[height=50mm,width=50mm]{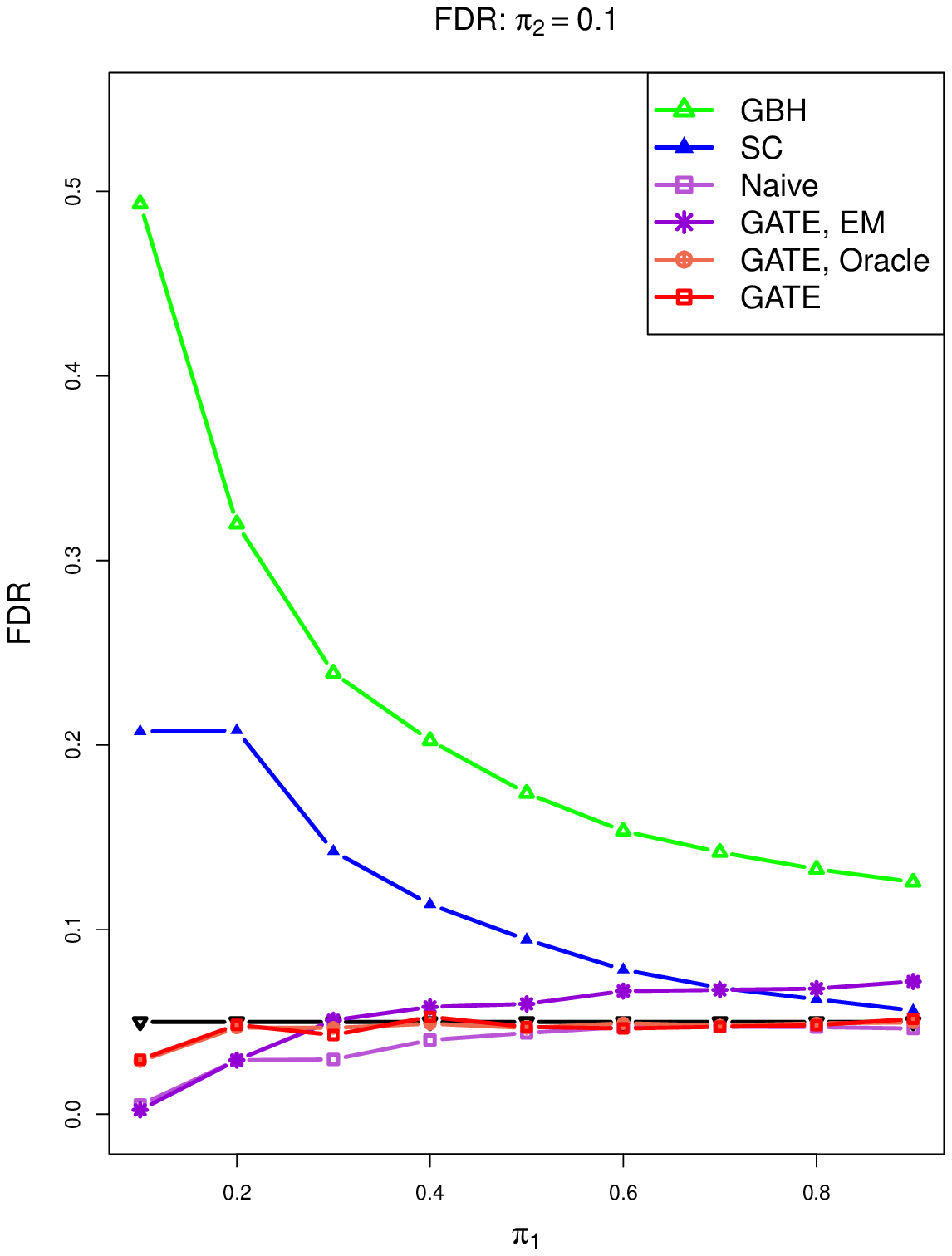}
  \includegraphics[height=50mm,width=50mm]{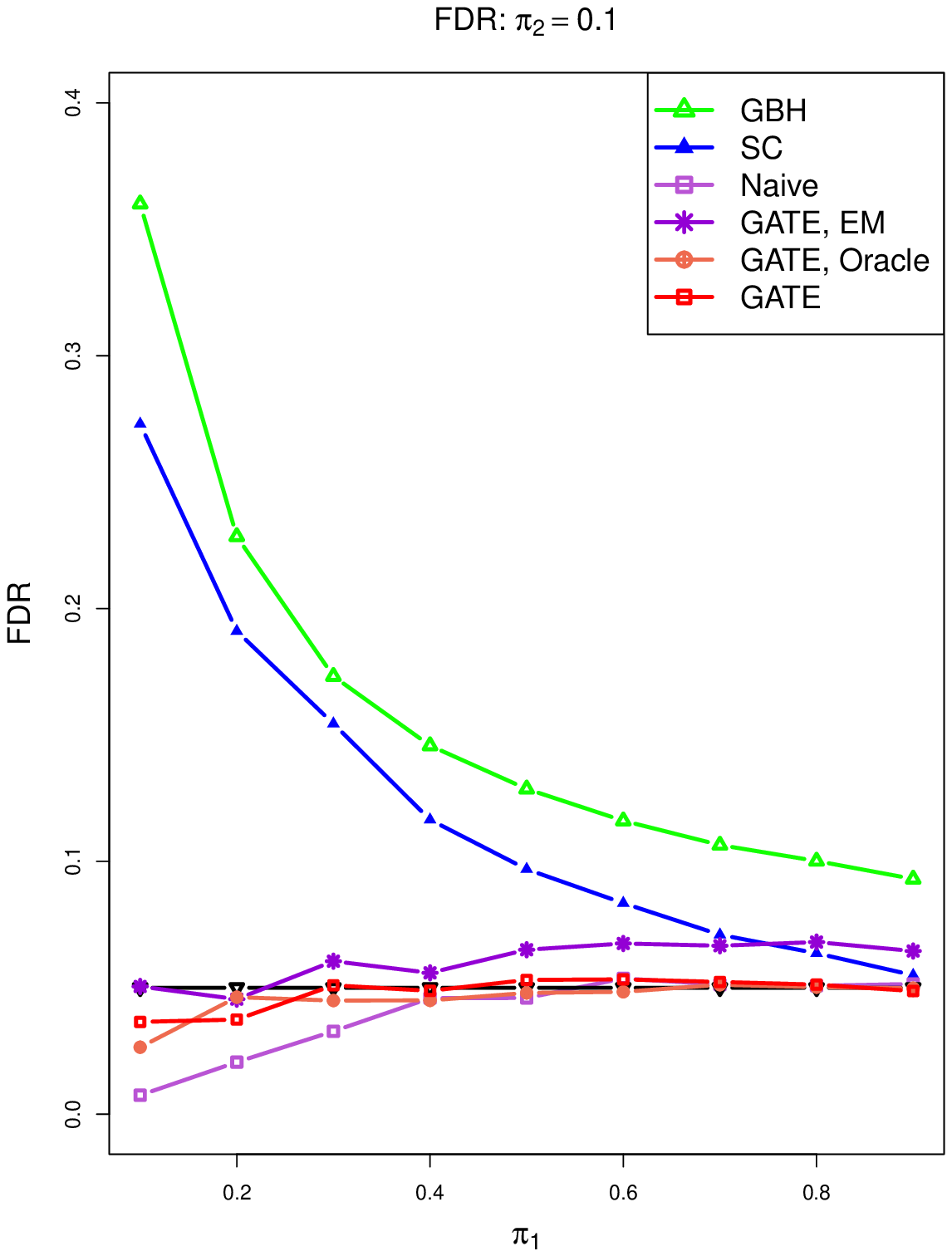}    \\
    \includegraphics[height=50mm,width=50mm]{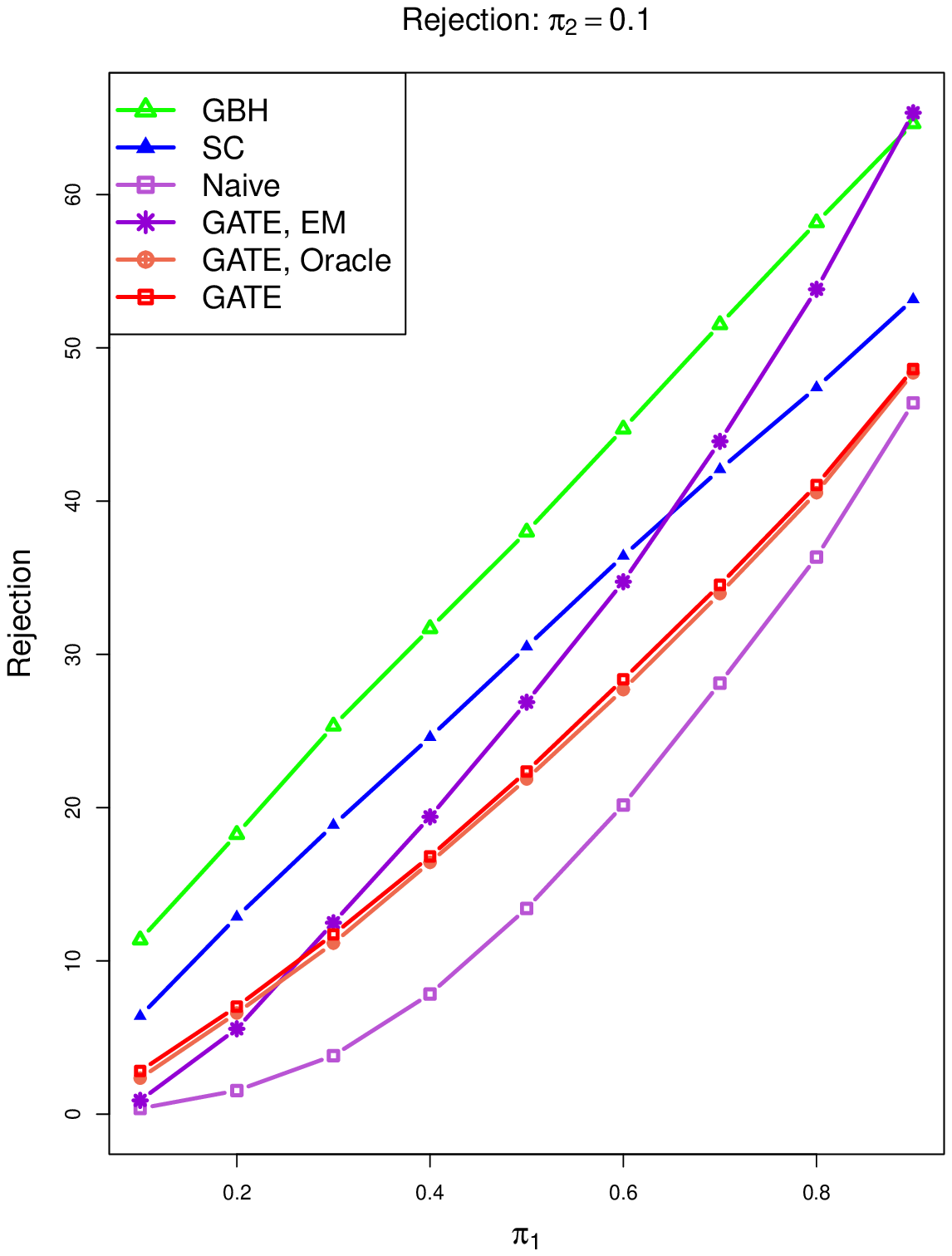}
  \includegraphics[height=50mm,width=50mm]{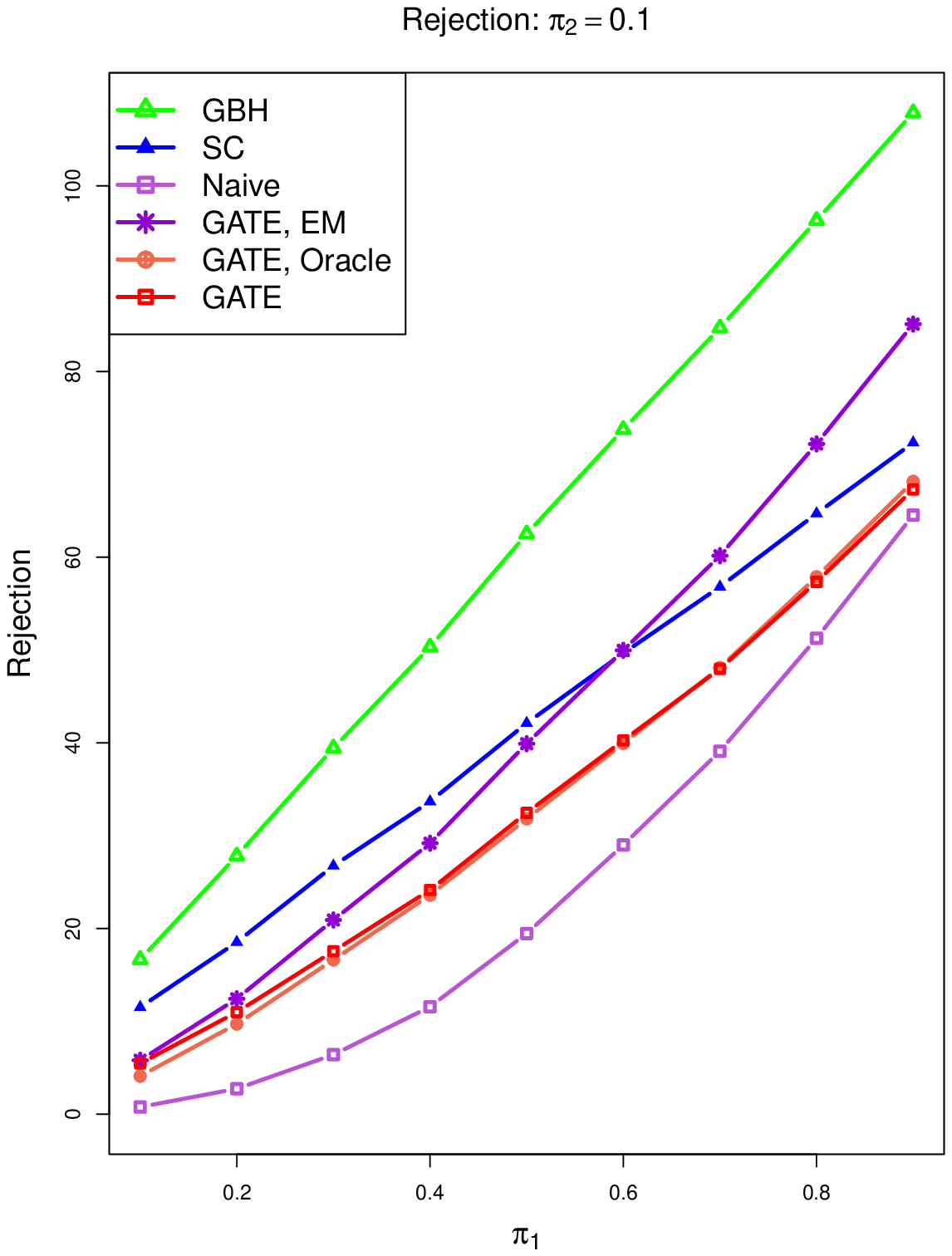}
  \caption{One-Way GATE 1: $m=100, n=50, \pi_2=0.1$. The  left and right panels correspond to $K=1$ and $2$, respectively.
  }\label{fig:gate1:fdr:s3}
\end{figure}

\begin{figure}[H]
  \centering
  \includegraphics[height=50mm,width=50mm]{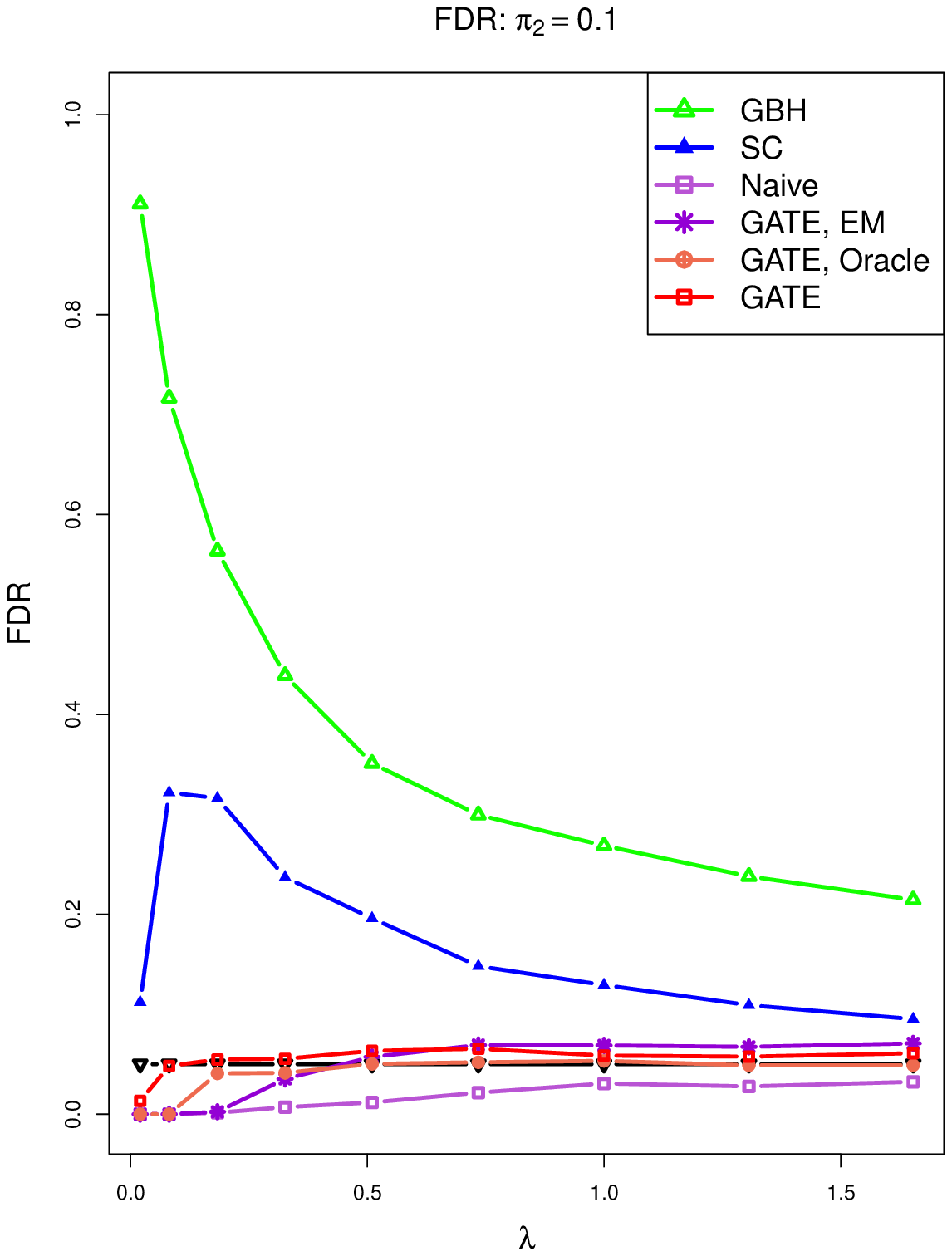}
  \includegraphics[height=50mm,width=50mm]{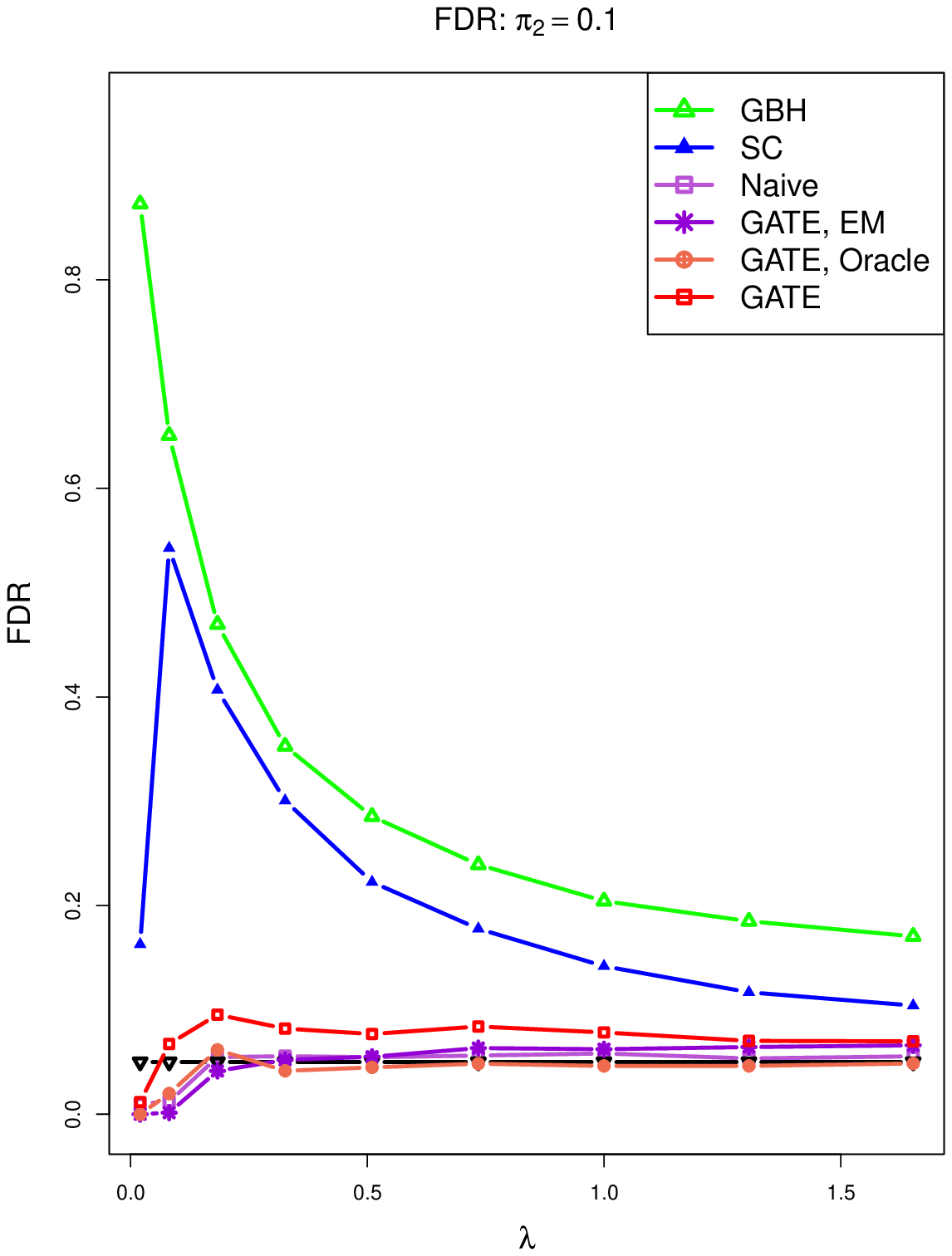}\\
  \includegraphics[height=50mm,width=50mm]{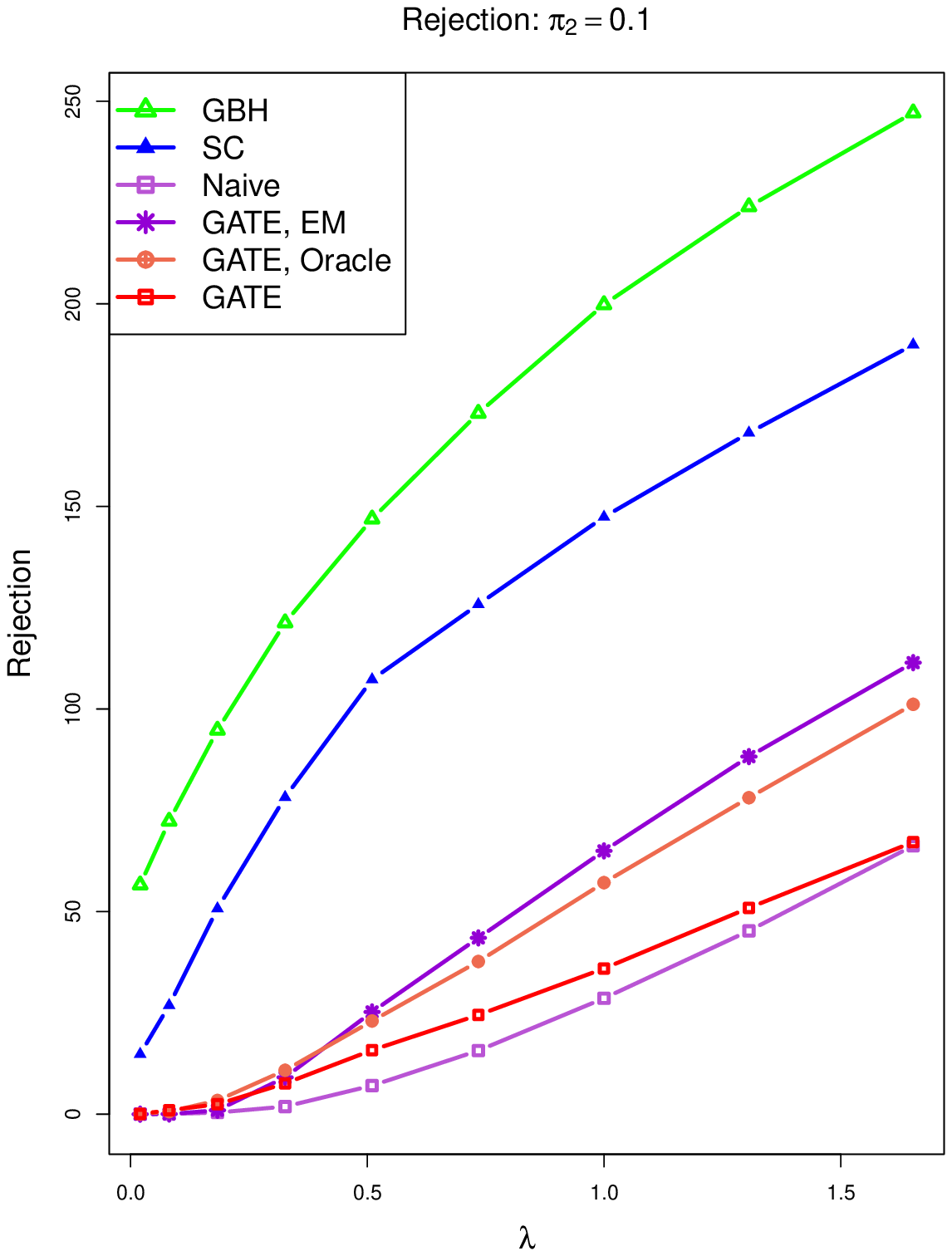}
  \includegraphics[height=50mm,width=50mm]{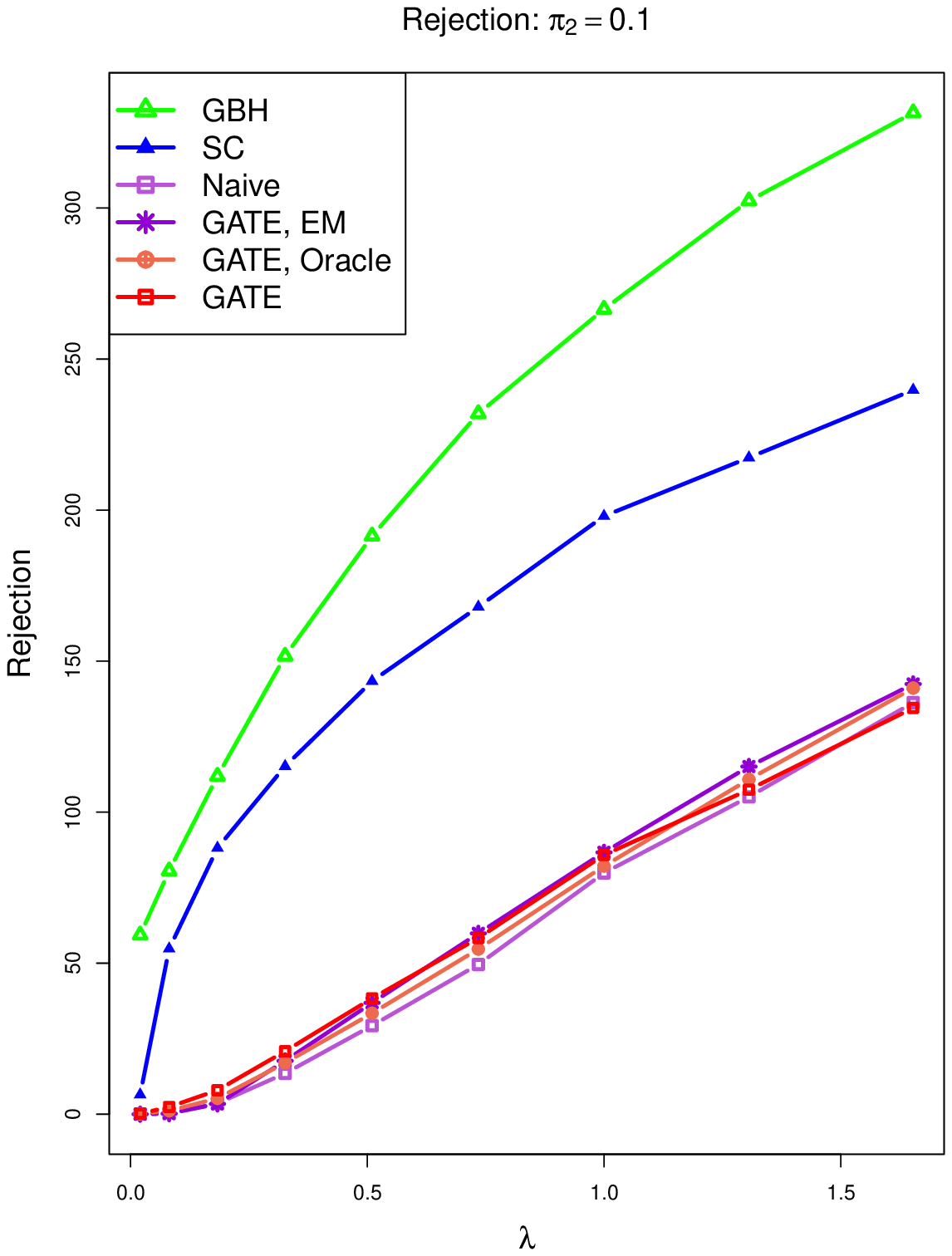}
  \caption{One-Way GATE 1: $m=1000, n=5, \pi_2=0.1$. The  left and right panels correspond to $K=1$ and $2$, respectively.
  }\label{fig:gate1:fdr:s4}
\end{figure}

\begin{figure}[H]
  \centering
  \includegraphics[height=50mm,width=50mm]{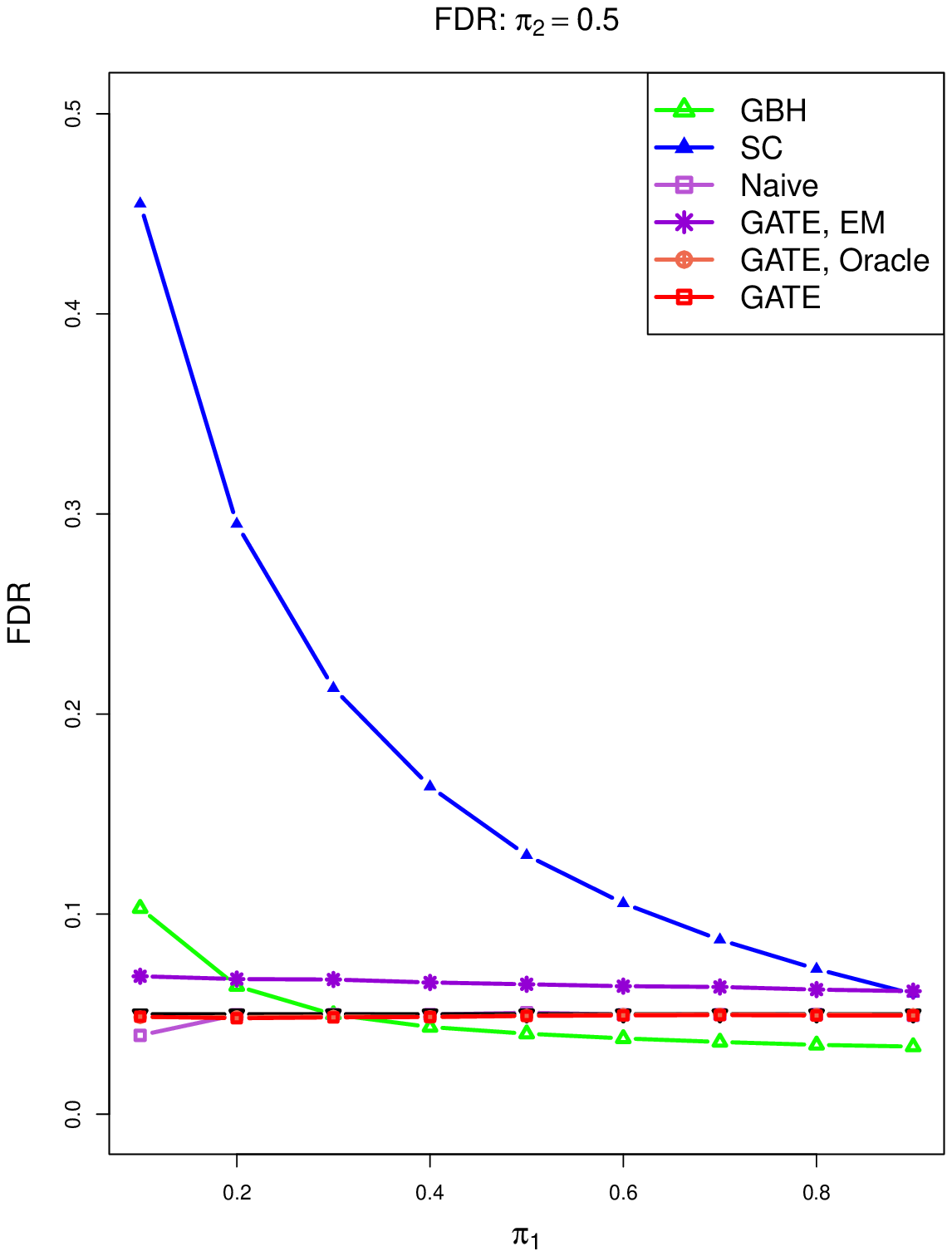}
  \includegraphics[height=50mm,width=50mm]{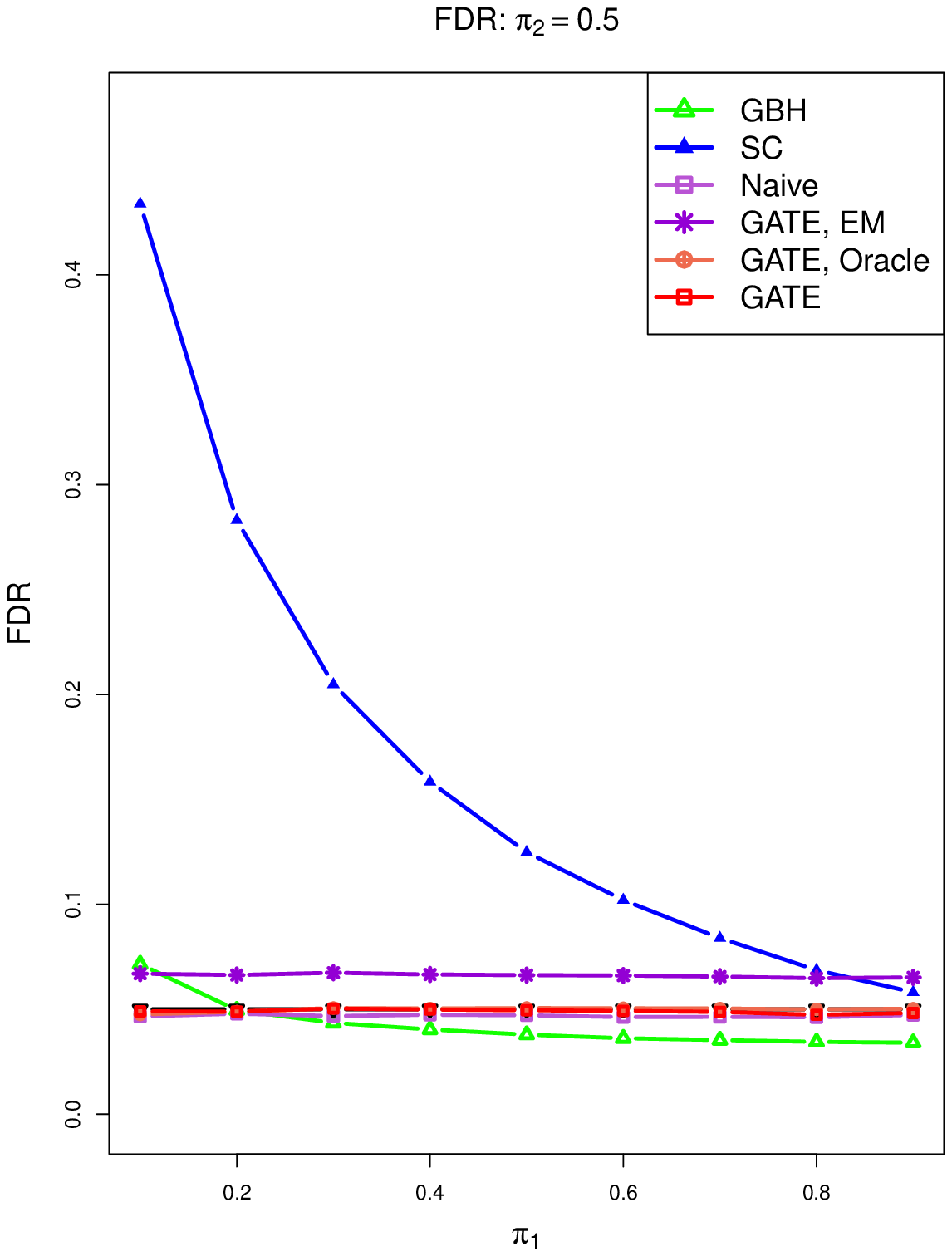}  \\
    \includegraphics[height=50mm,width=50mm]{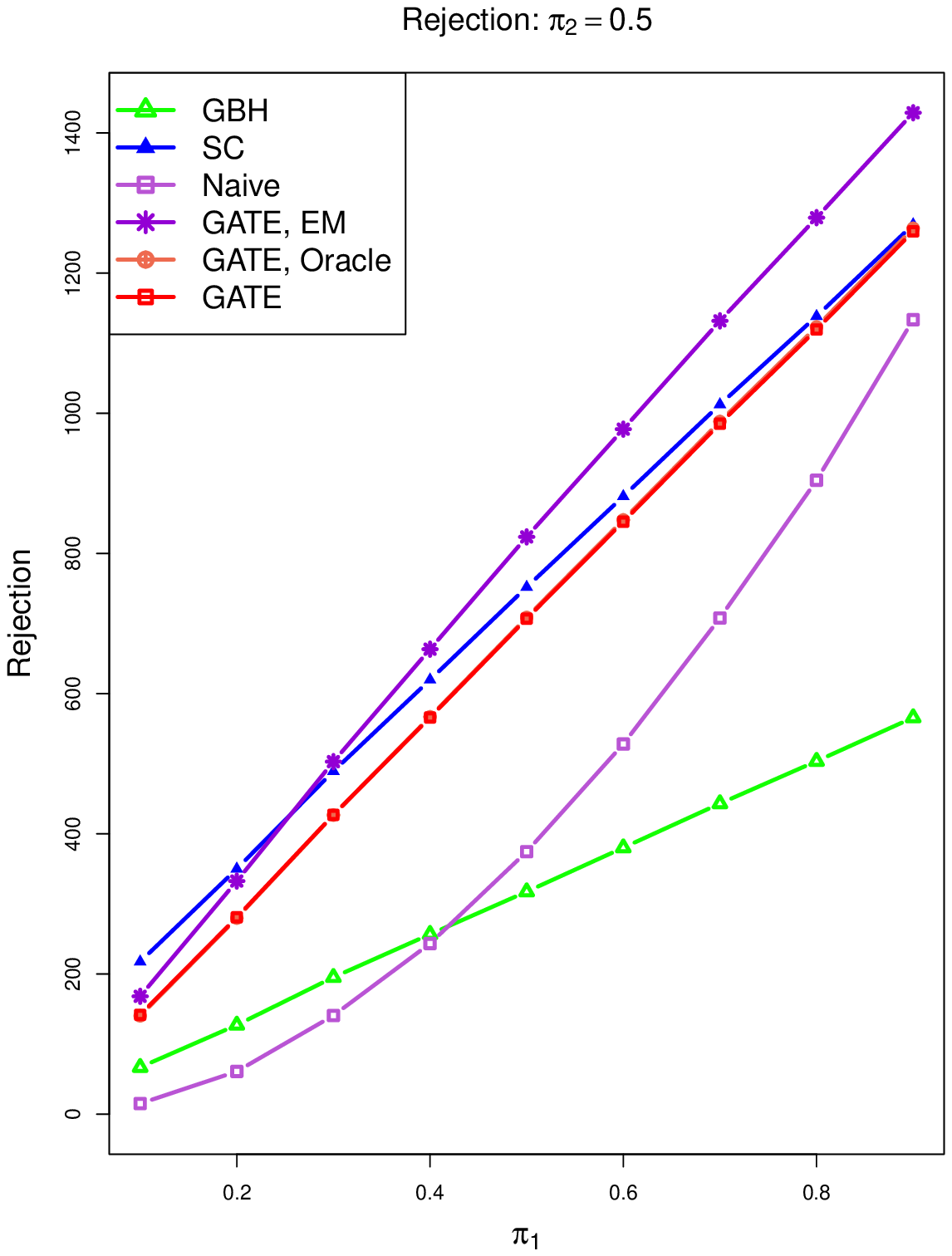}
  \includegraphics[height=50mm,width=50mm]{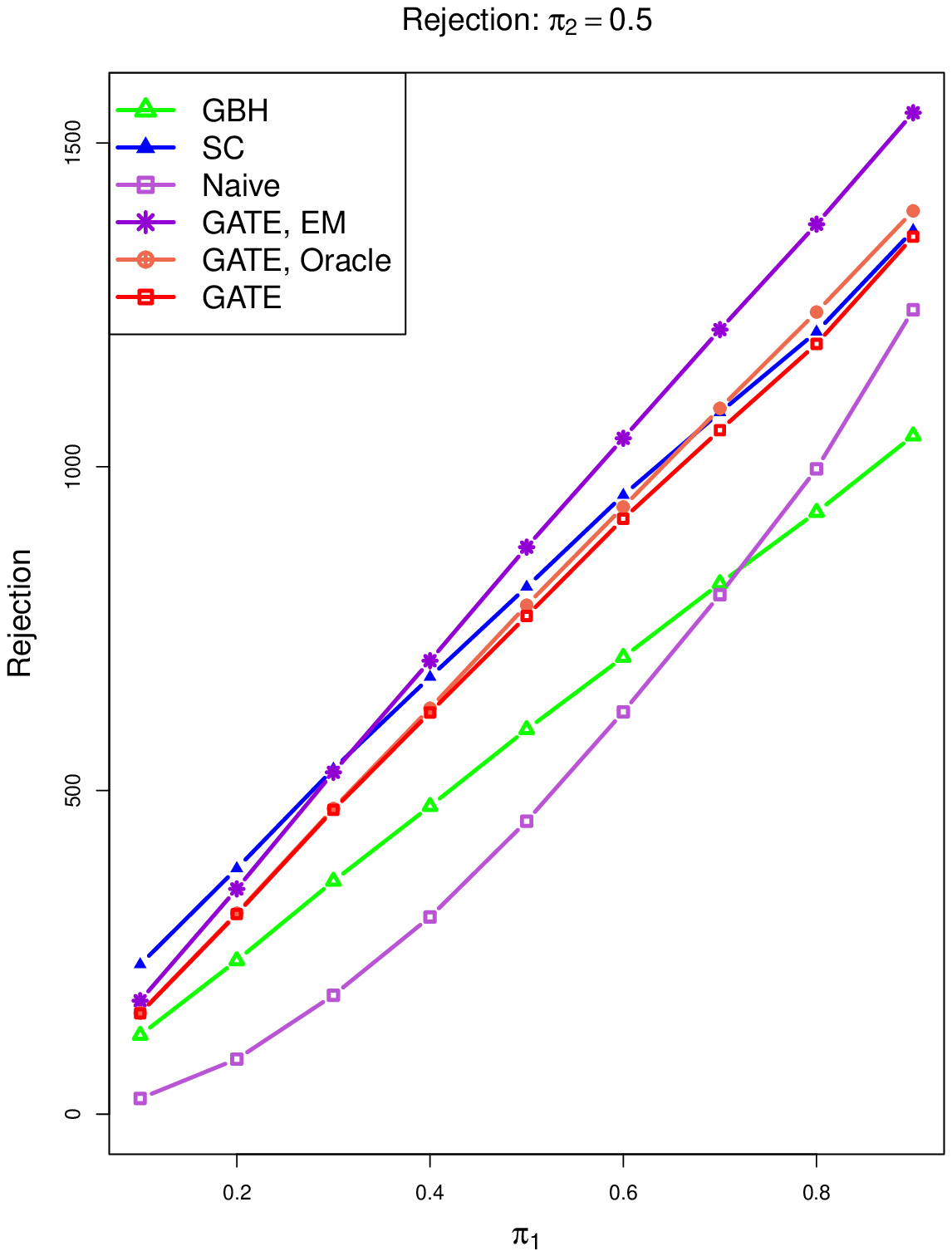}
  \caption{One-Way GATE 1: $m=100, n=50, \pi_2=0.5$. The  left and right panels correspond to $K=1$ and $2$, respectively.
  }\label{fig:gate1:fdr:s5}
\end{figure}

\begin{figure}[H]
  \centering
   
  \includegraphics[height=50mm,width=50mm]{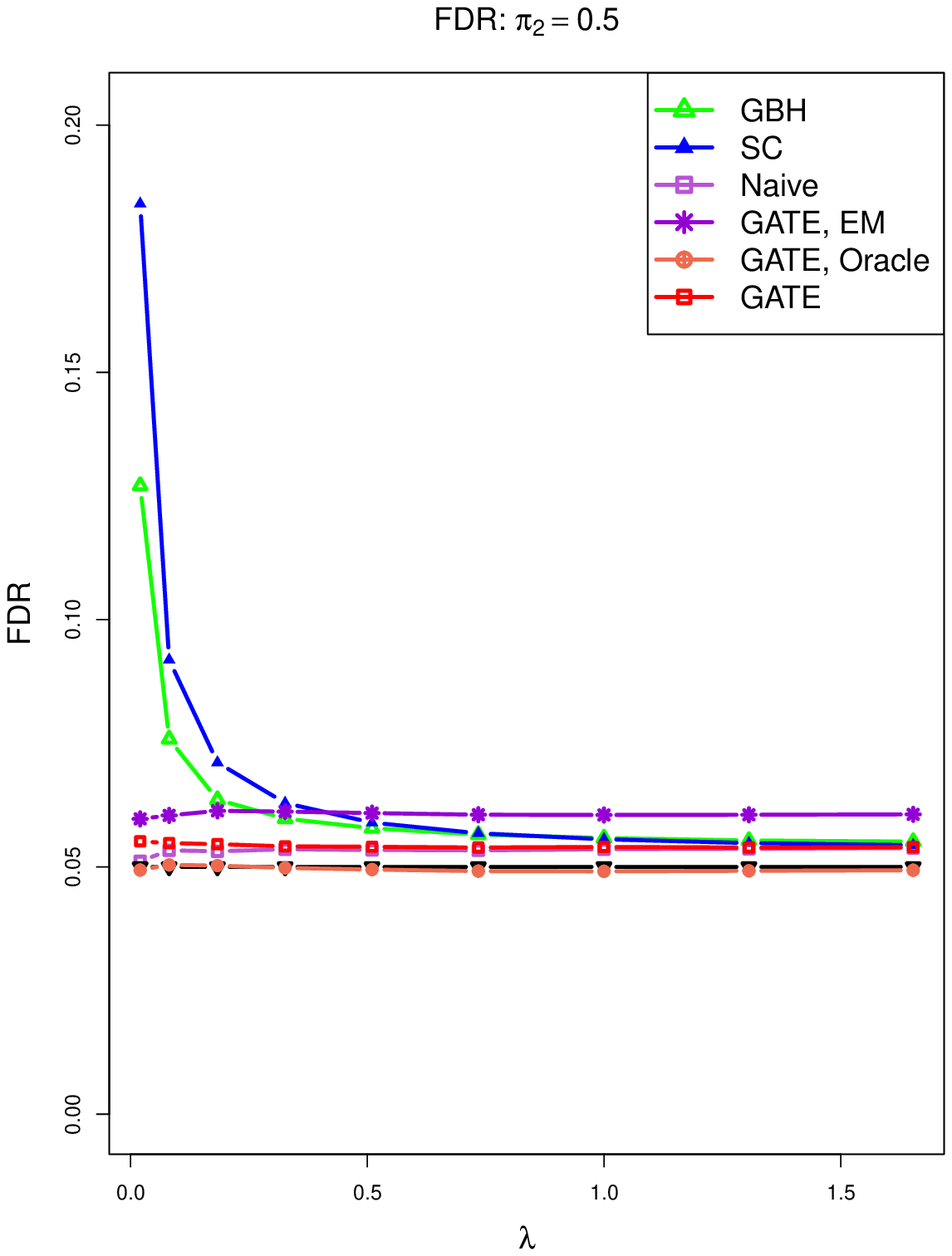}
  \includegraphics[height=50mm,width=50mm]{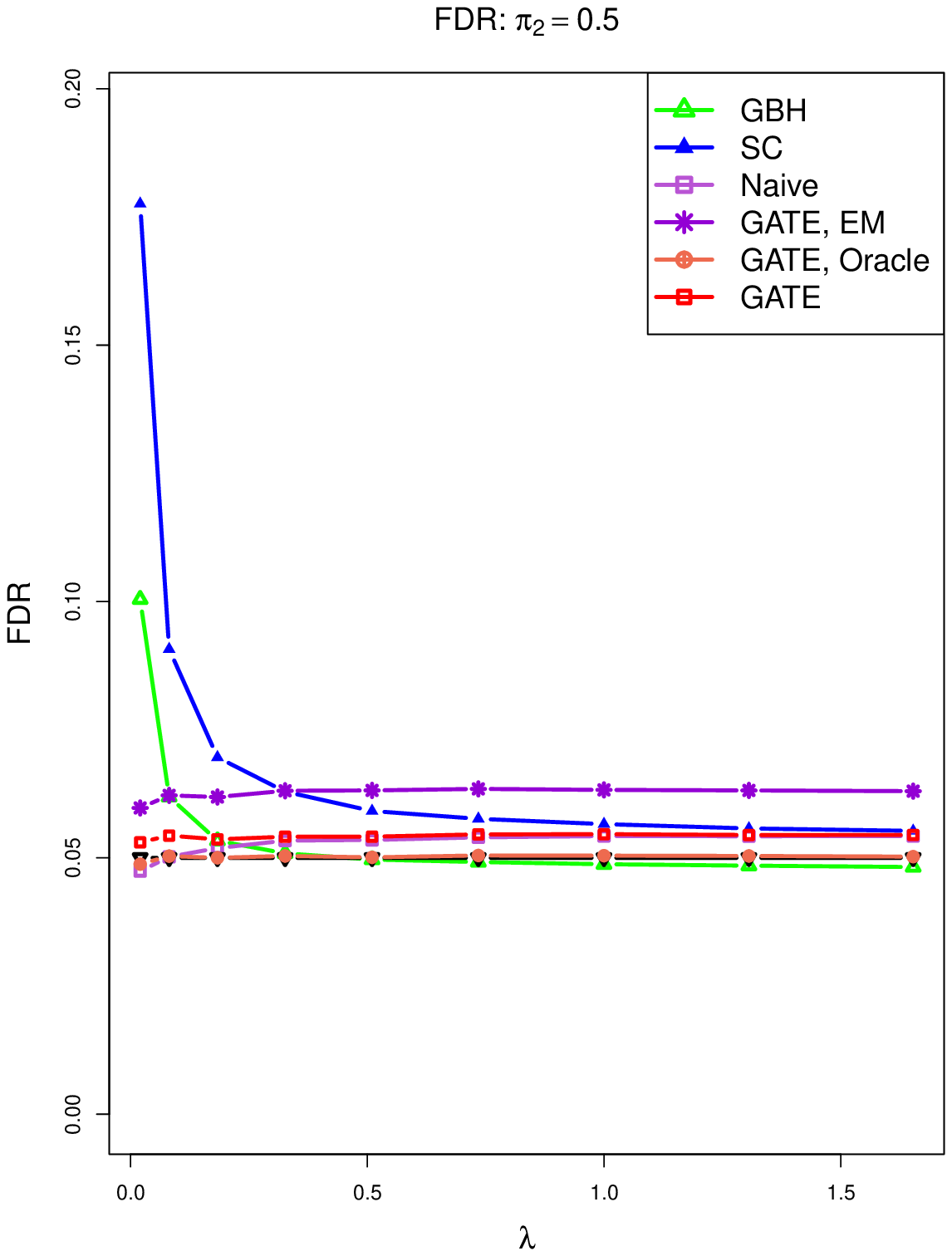}\\
  \includegraphics[height=50mm,width=50mm]{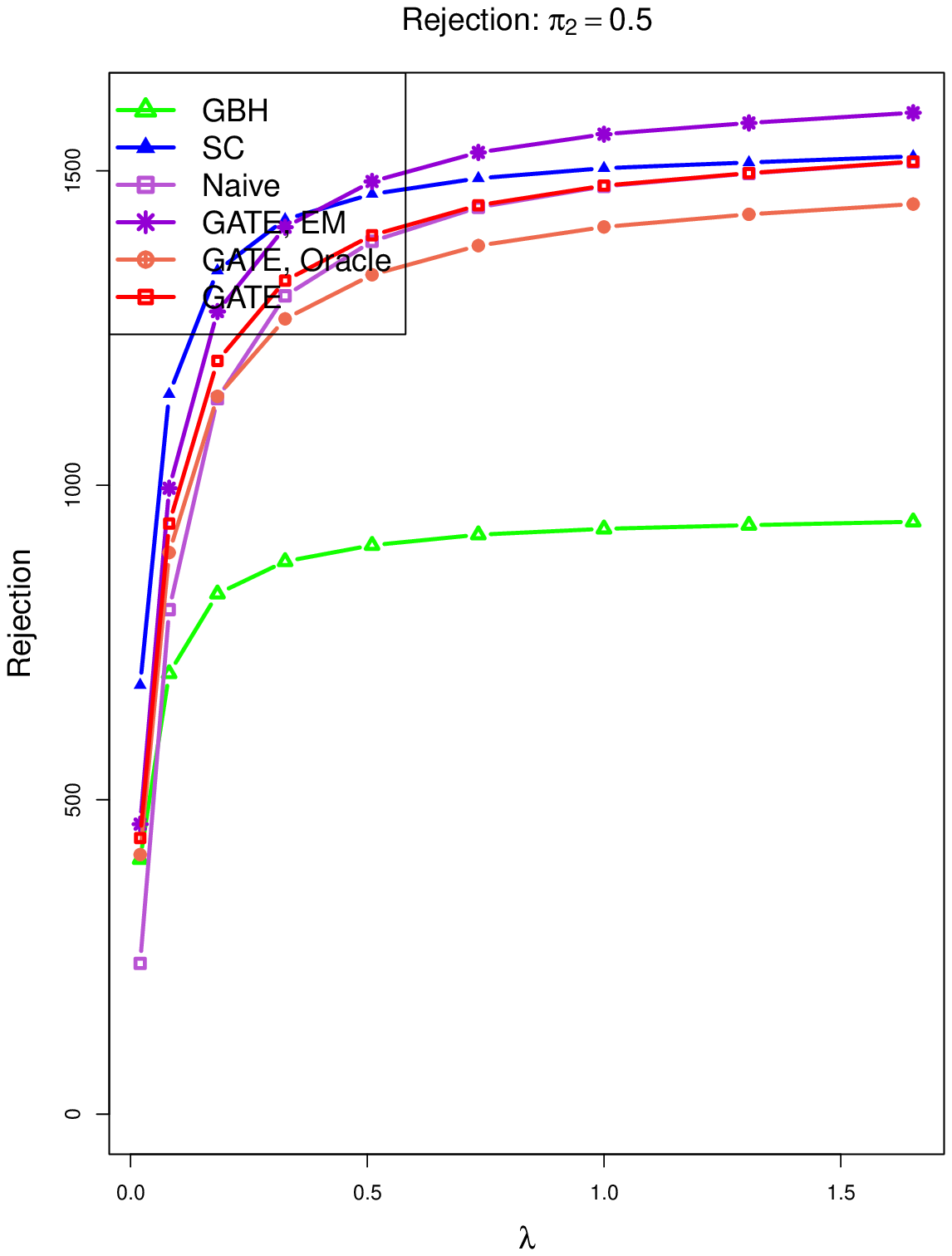}
  \includegraphics[height=50mm,width=50mm]{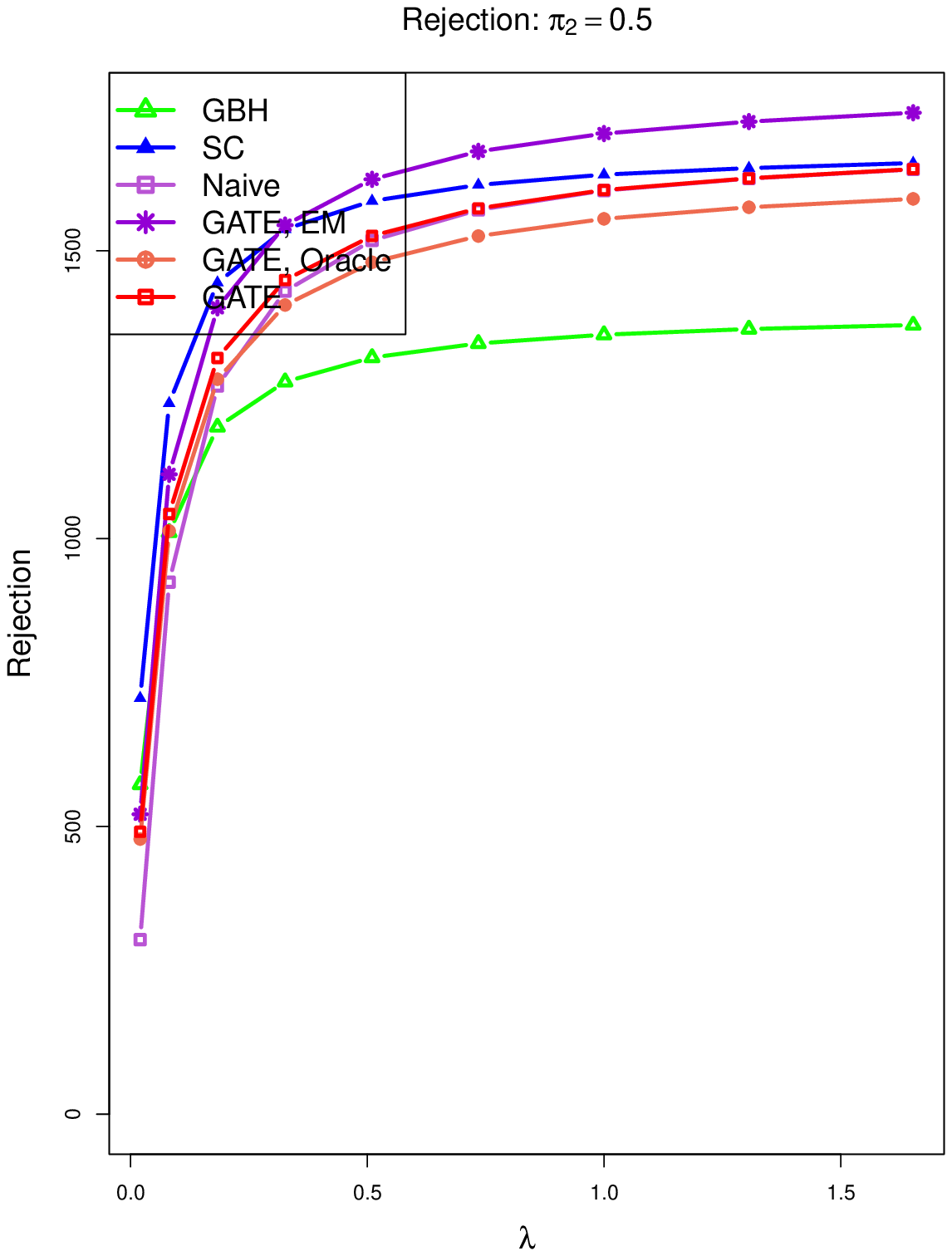}
  \caption{One-Way GATE 1: $m=1000, n=5, \pi_2=0.5$. The  left and right panels correspond to $K=1$ and $2$, respectively.
  }\label{fig:gate1:fdr:s6}
\end{figure}

\begin{figure}[H]
  \centering
  \includegraphics[height=50mm,width=50mm]{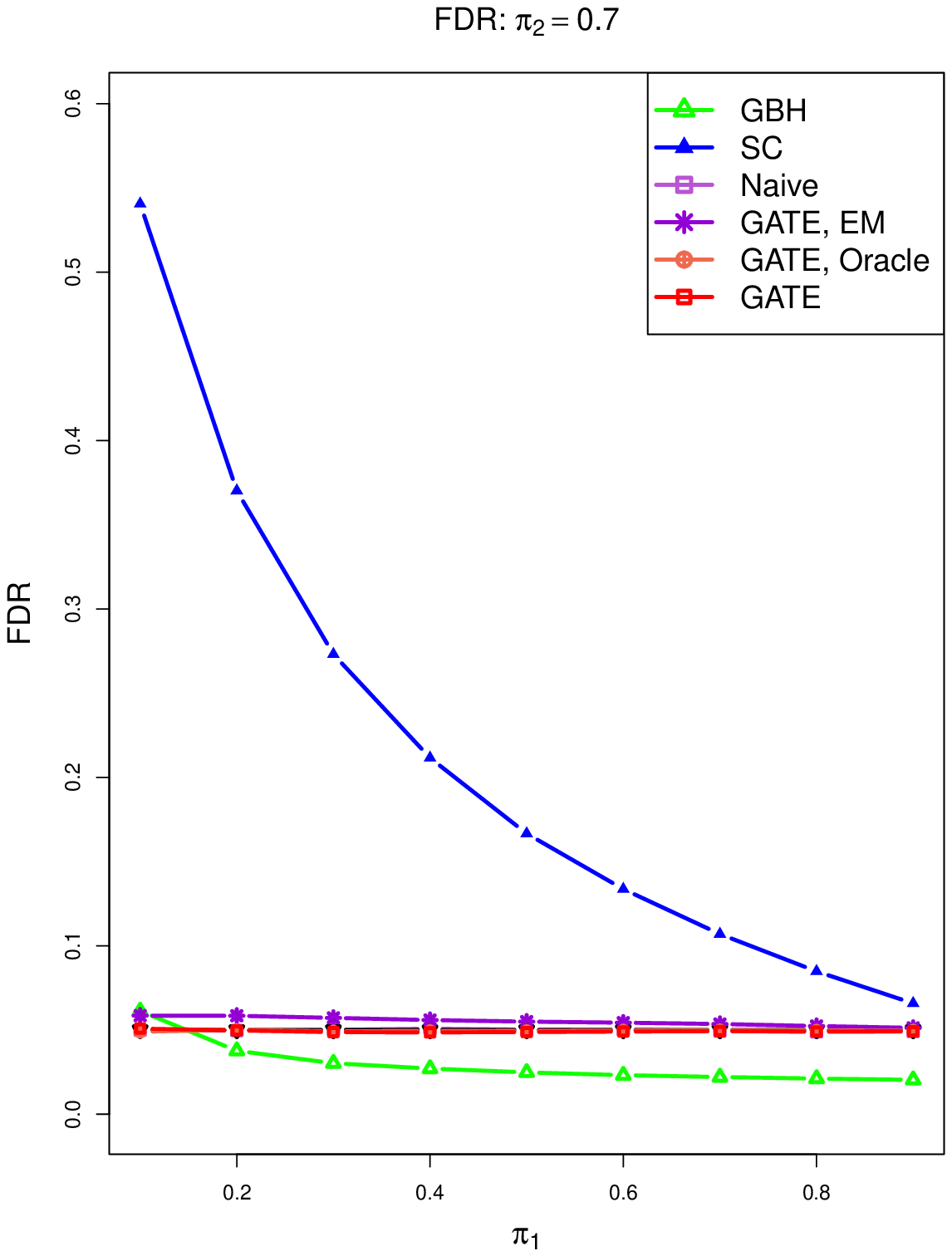}
  \includegraphics[height=50mm,width=50mm]{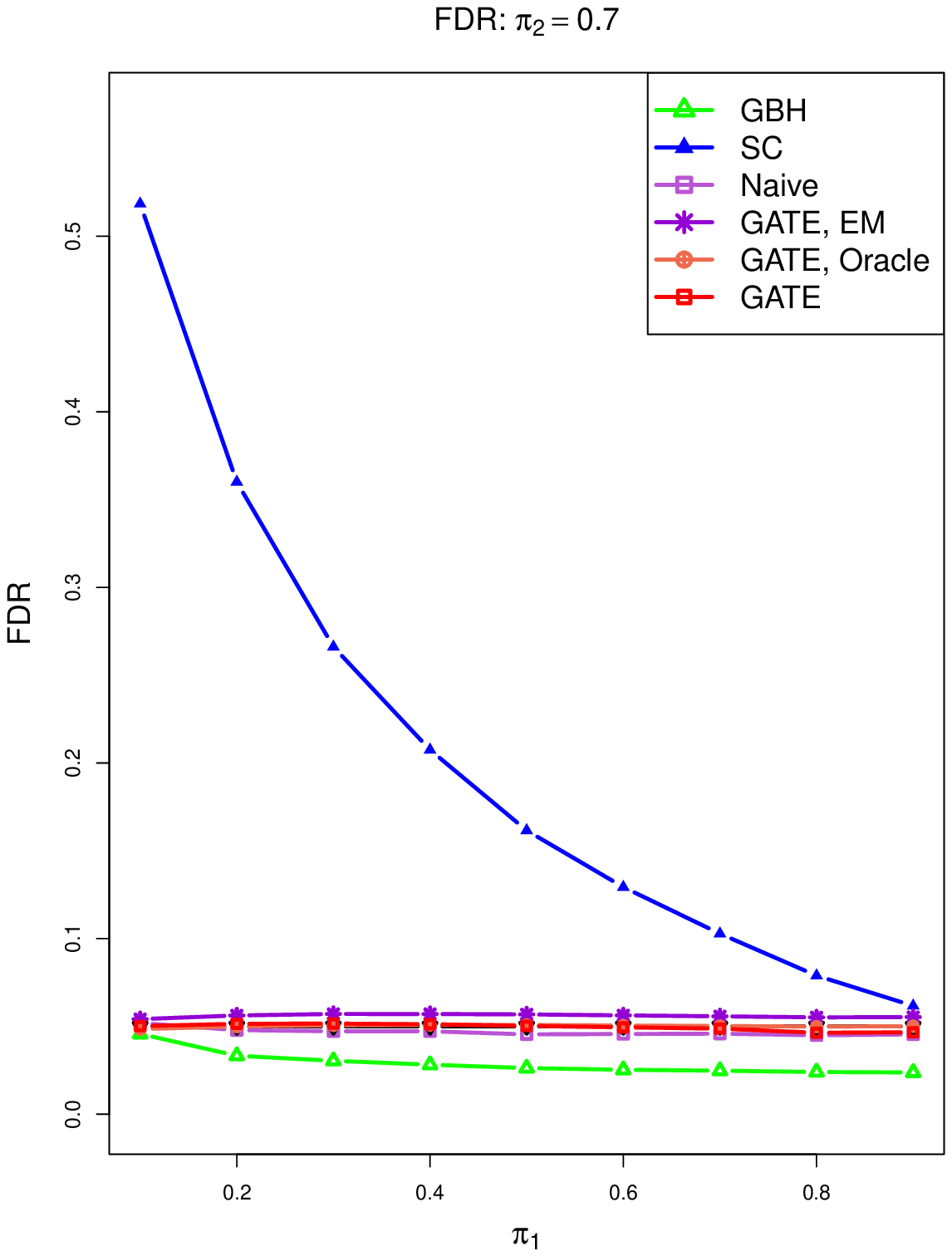}    \\
  \includegraphics[height=50mm,width=50mm]{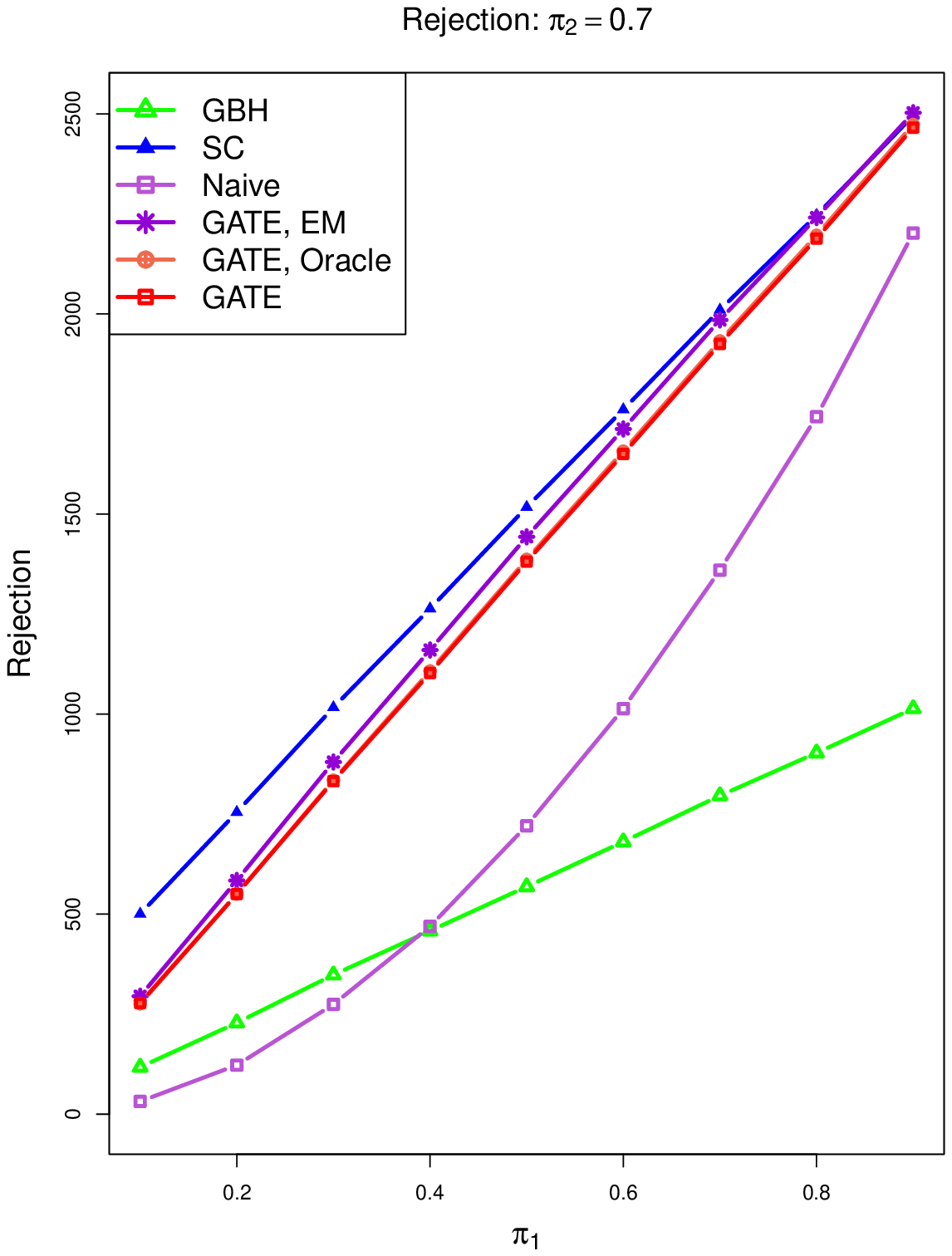}
  \includegraphics[height=50mm,width=50mm]{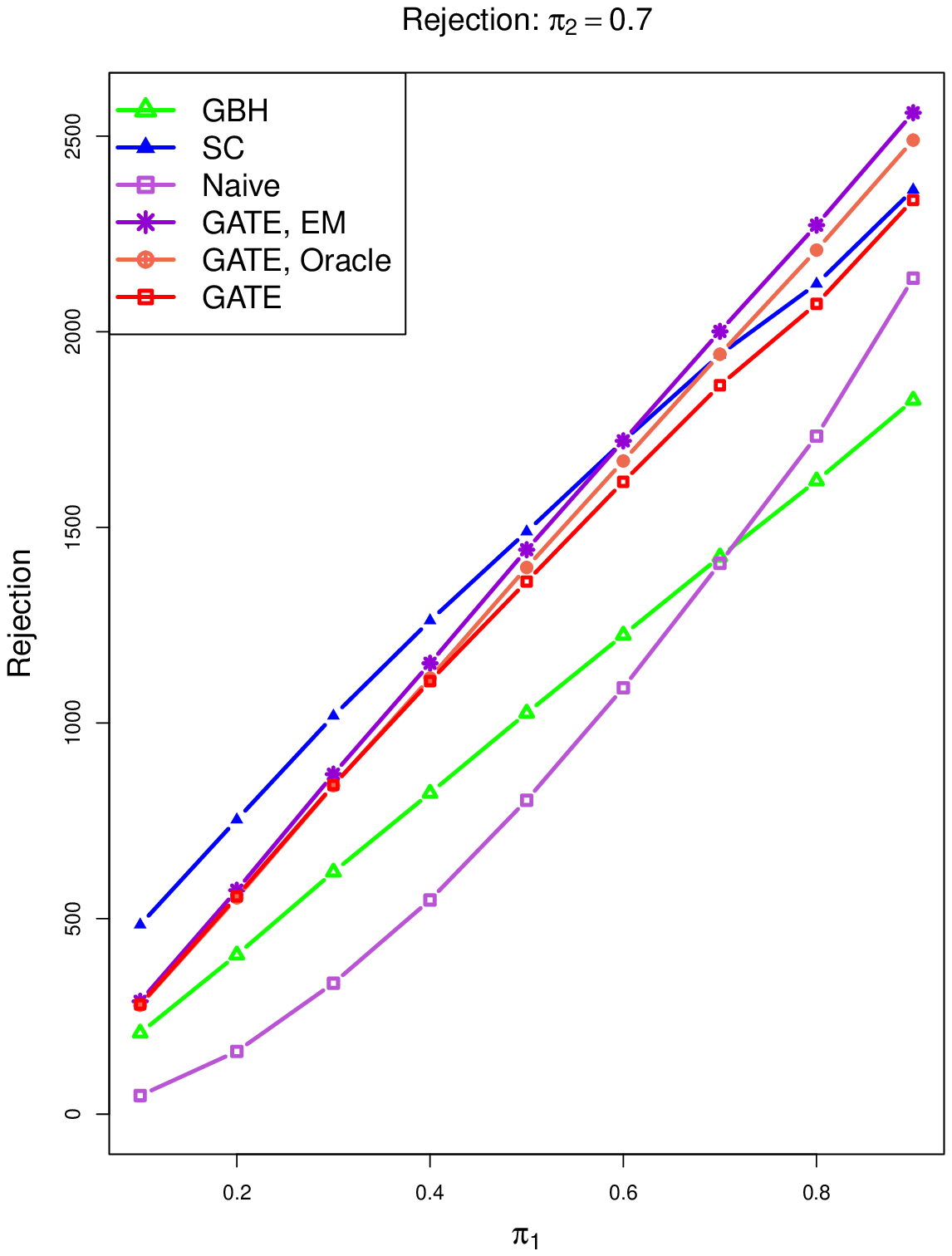}   
  \caption{One-Way GATE 1: $m=100, n=50, \pi_2=0.7$. The  left and right panels correspond to $K=1$ and $2$, respectively.
  }\label{fig:gate1:fdr:s7}
\end{figure}

\begin{figure}[H]
  \centering
  \includegraphics[height=50mm,width=50mm]{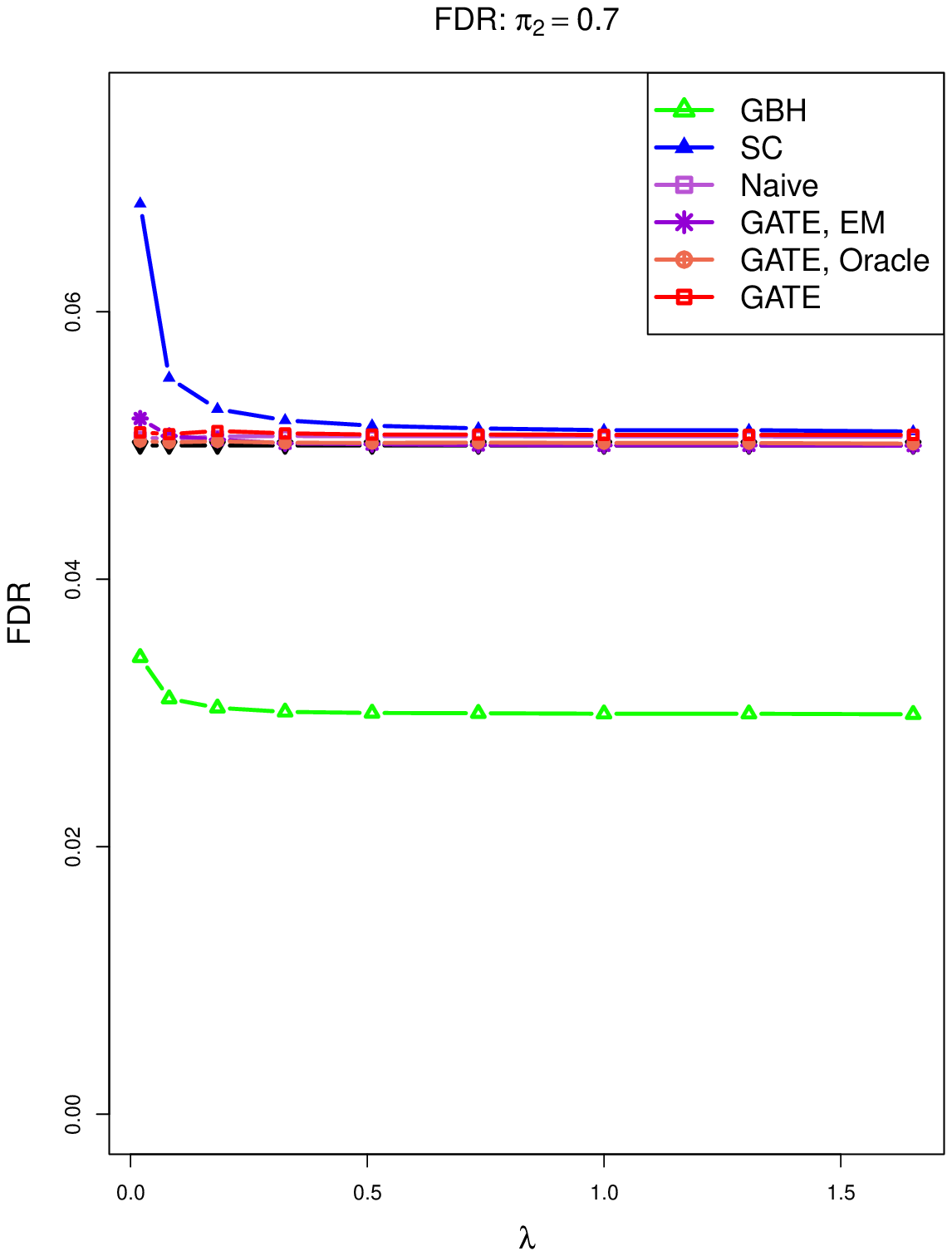}
  \includegraphics[height=50mm,width=50mm]{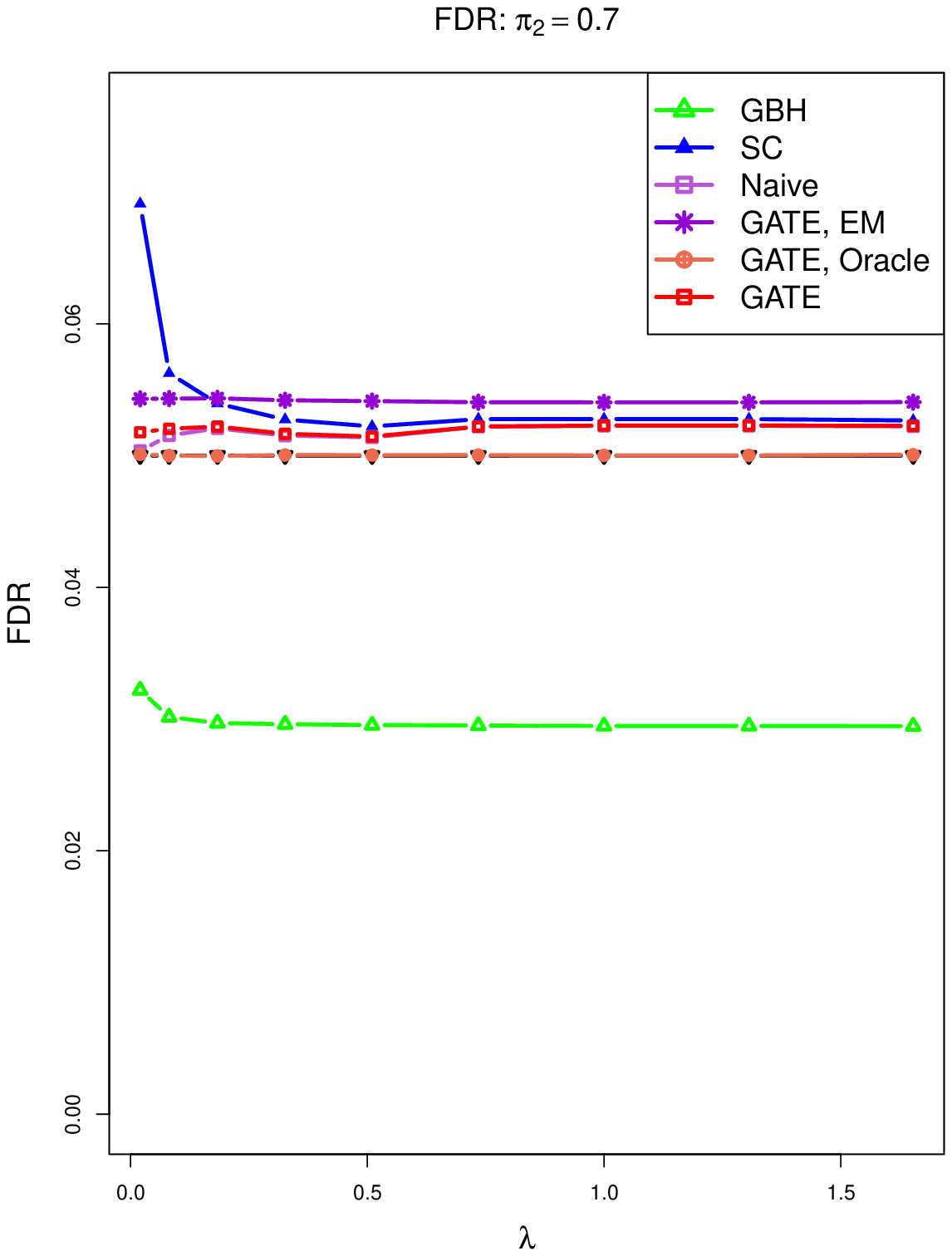}\\
  \includegraphics[height=50mm,width=50mm]{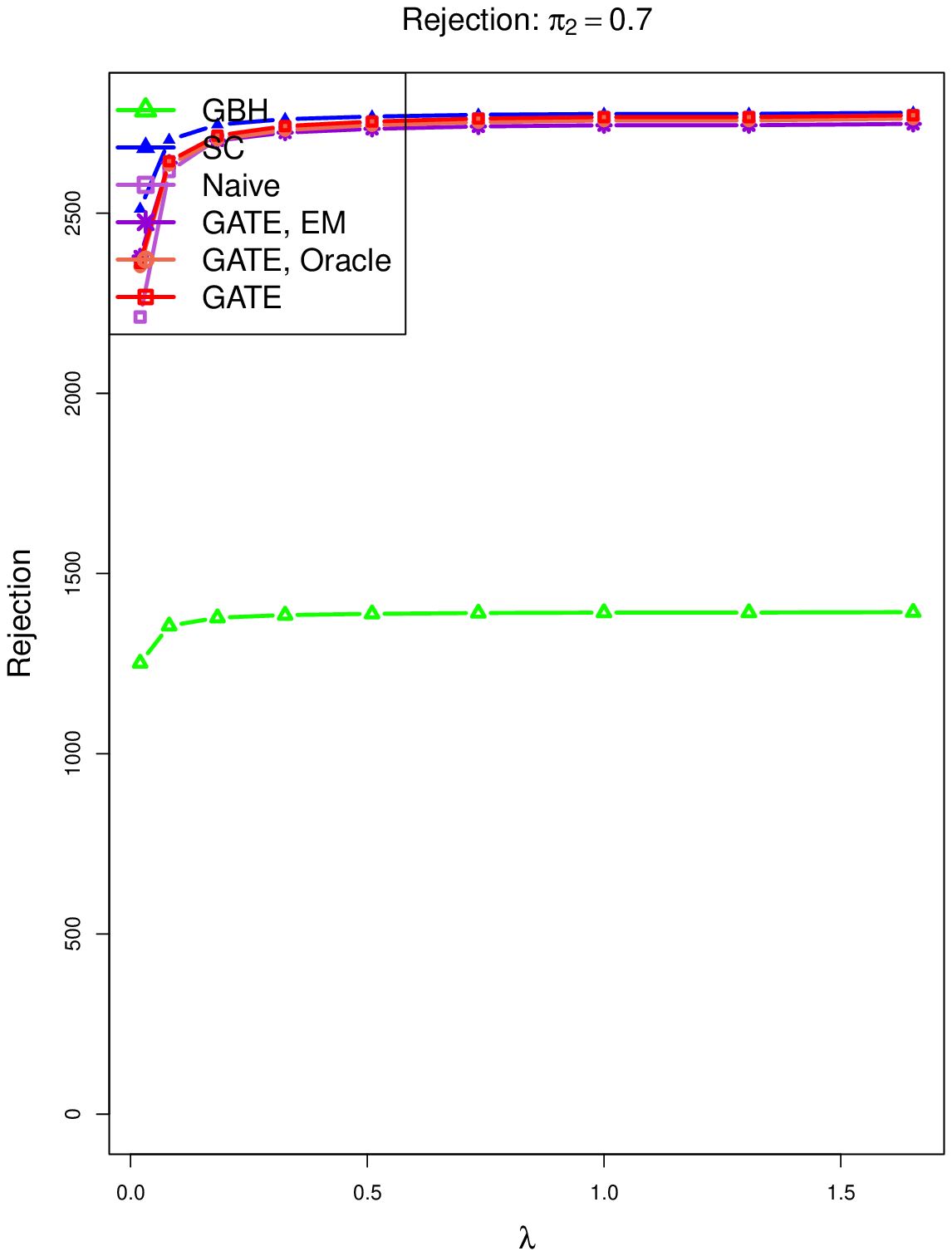}
  \includegraphics[height=50mm,width=50mm]{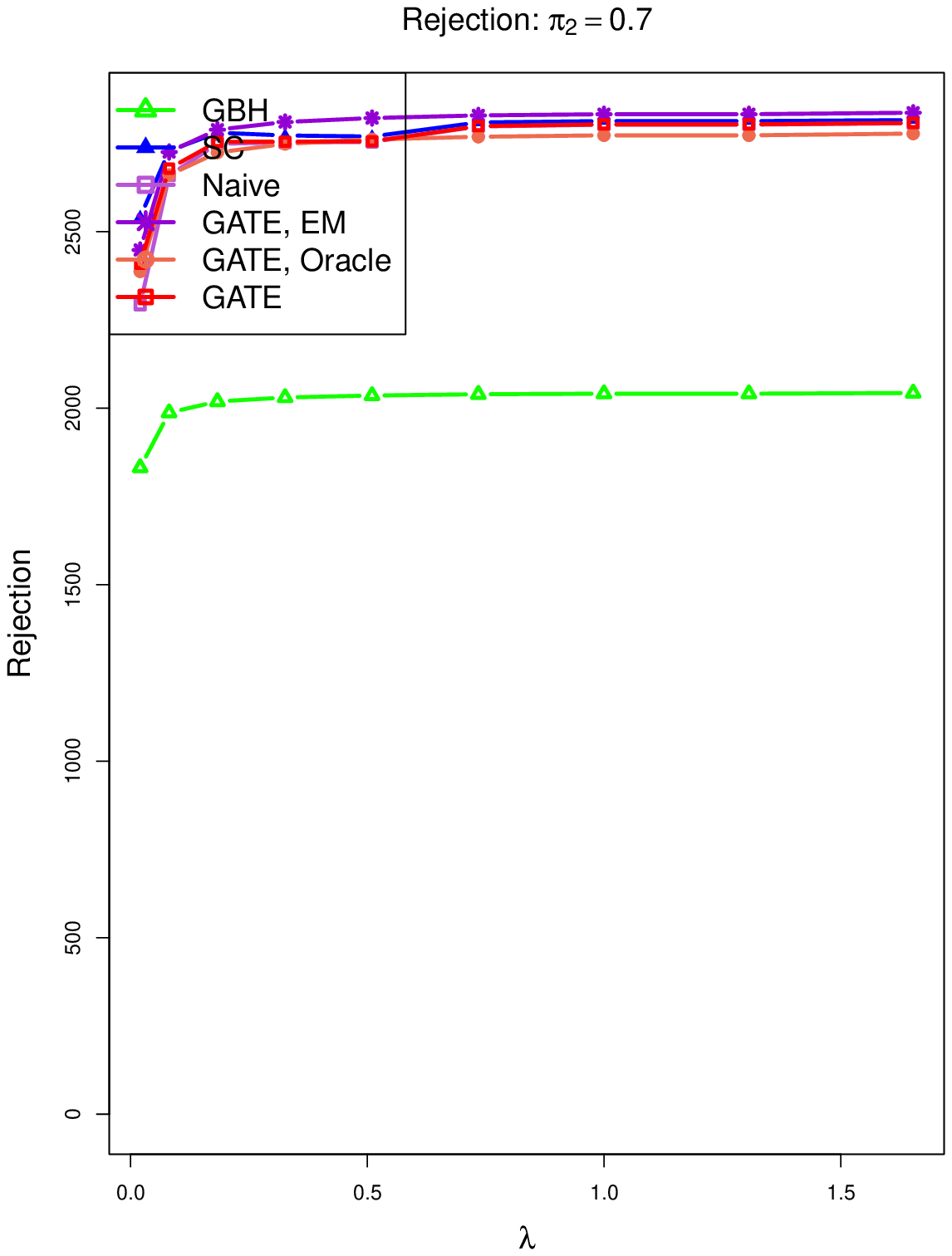}
  \caption{One-Way GATE 1: $m=1000, n=5, \pi_2=0.7$. The  left and right panels correspond to $K=1$ and $2$, respectively.
  }\label{fig:gate1:fdr:s8}
\end{figure}

\subsection{More simulation results on One-Way GATE 2}
In this section we list more simulation results on One-Way GATE 2.

\begin{figure}[H]
   \centering
   \includegraphics[height=40mm,width=40mm]{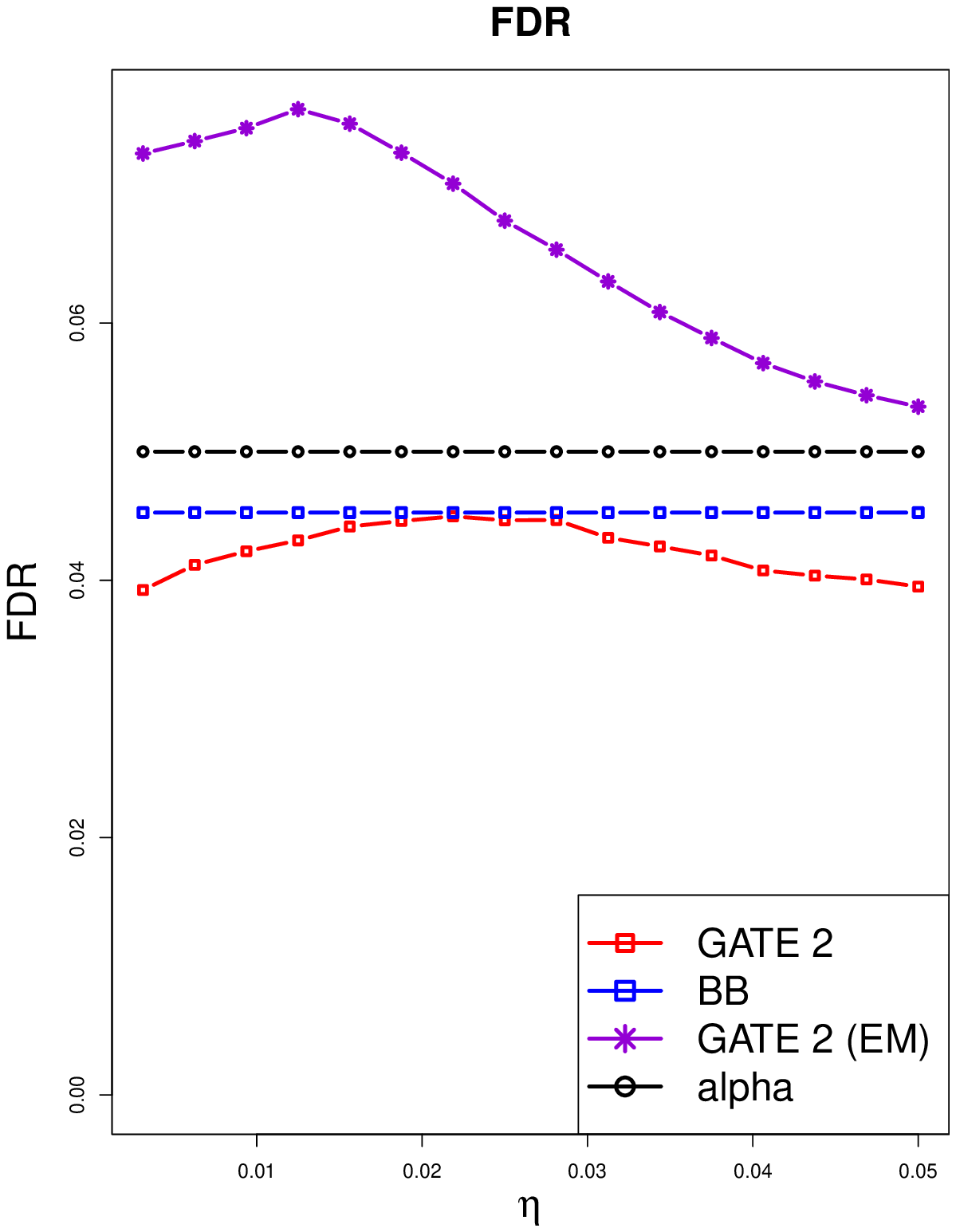}
   \includegraphics[height=40mm,width=40mm]{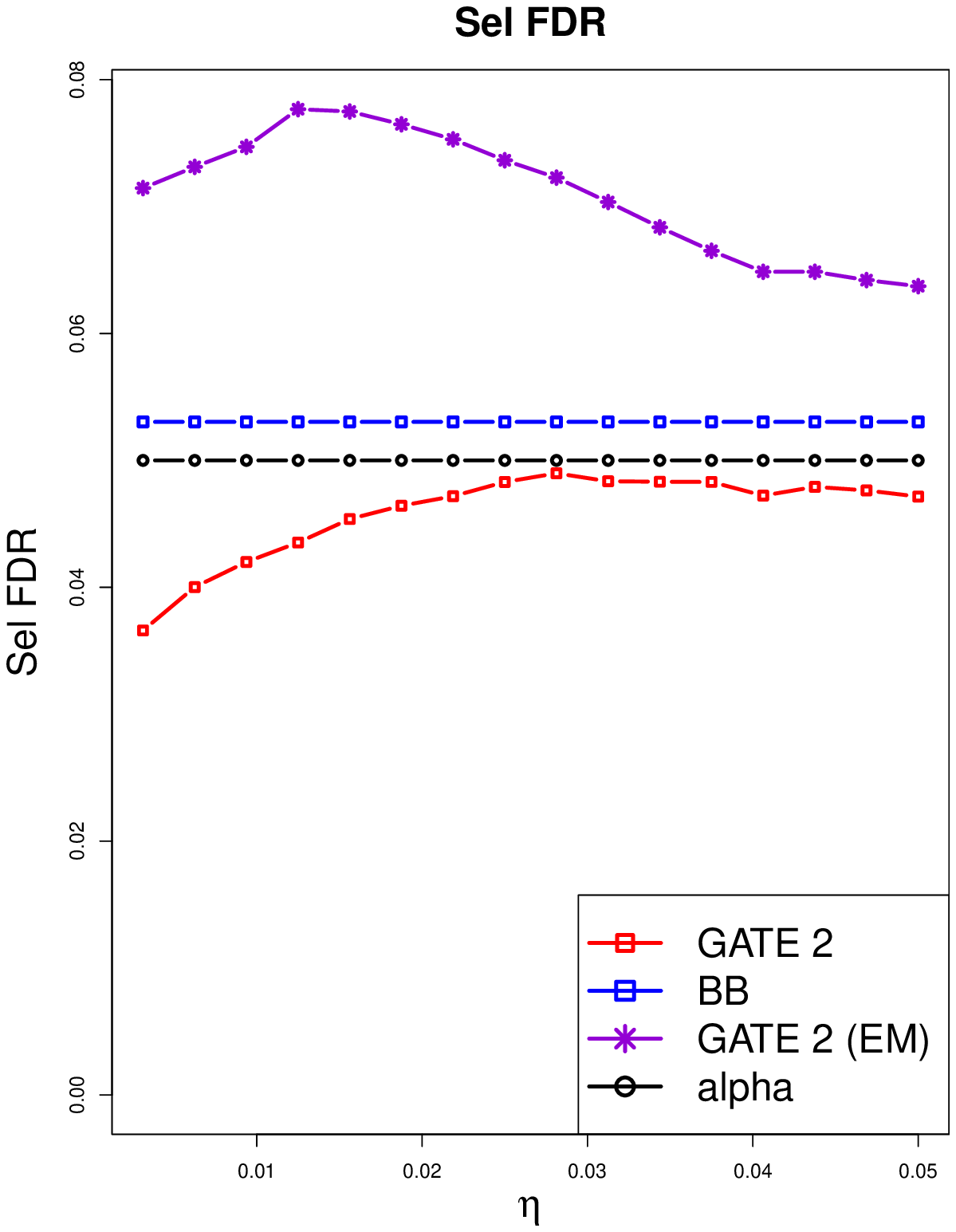}
   \includegraphics[height=40mm,width=40mm]{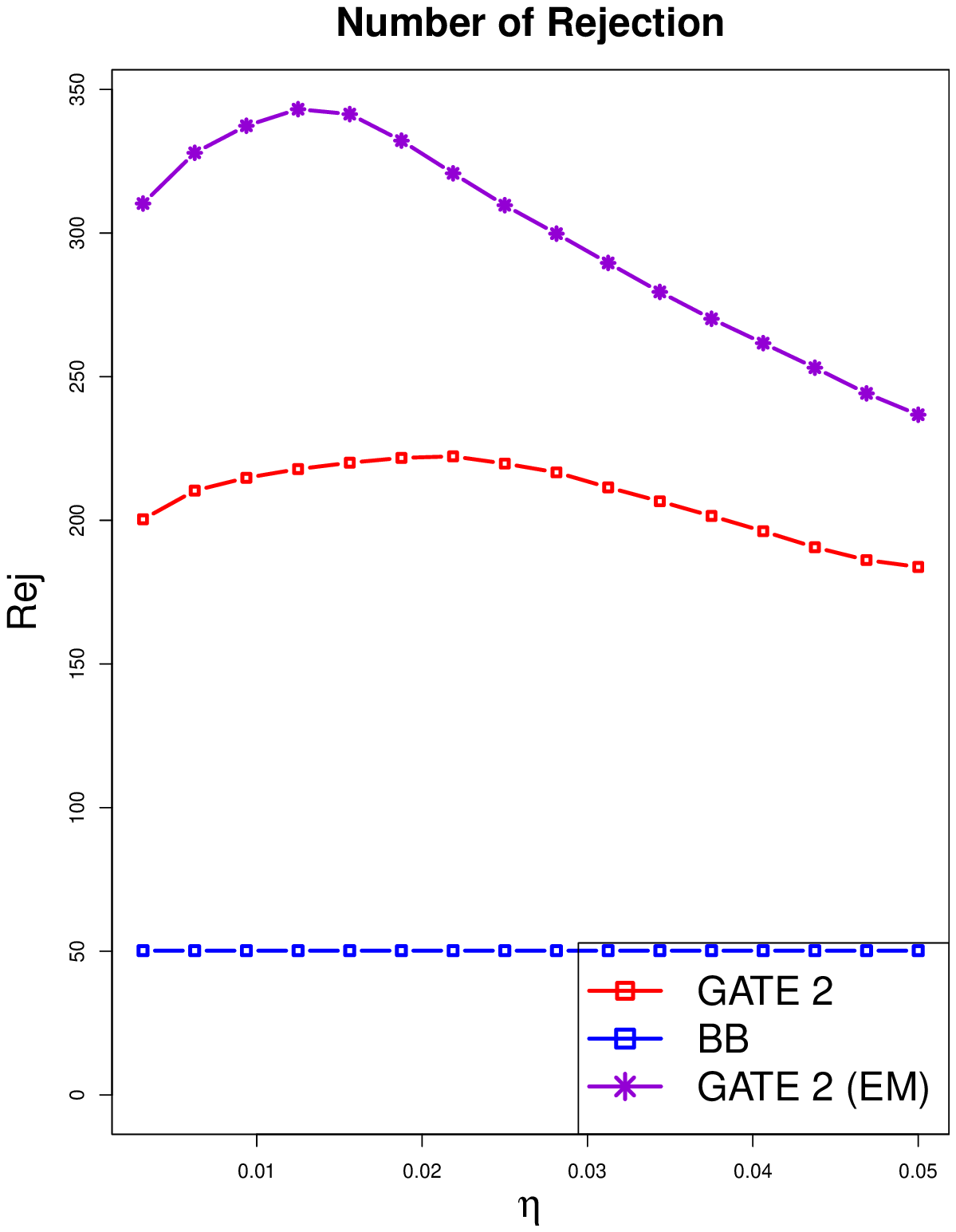}\\
   \includegraphics[height=40mm,width=40mm]{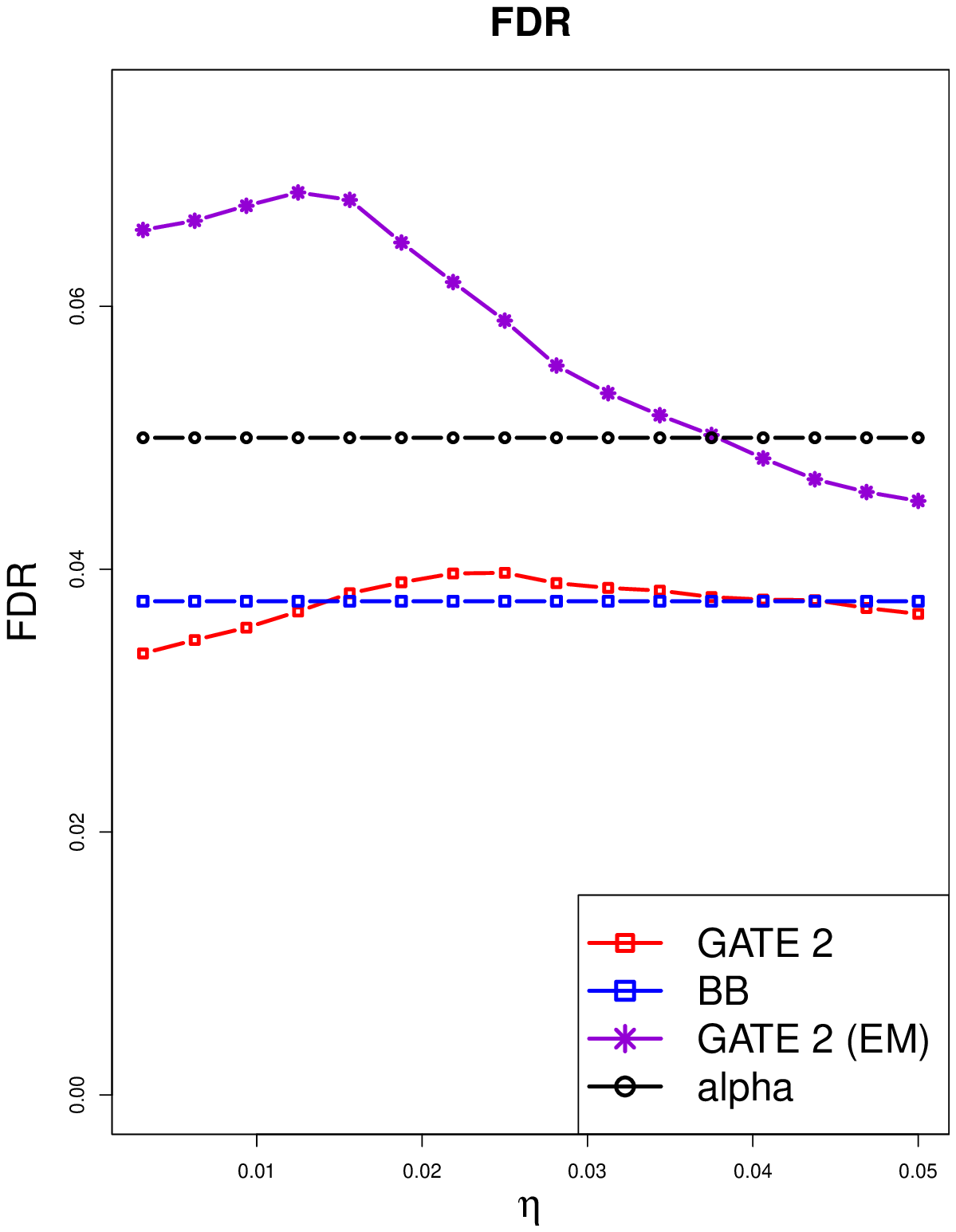}
   \includegraphics[height=40mm,width=40mm]{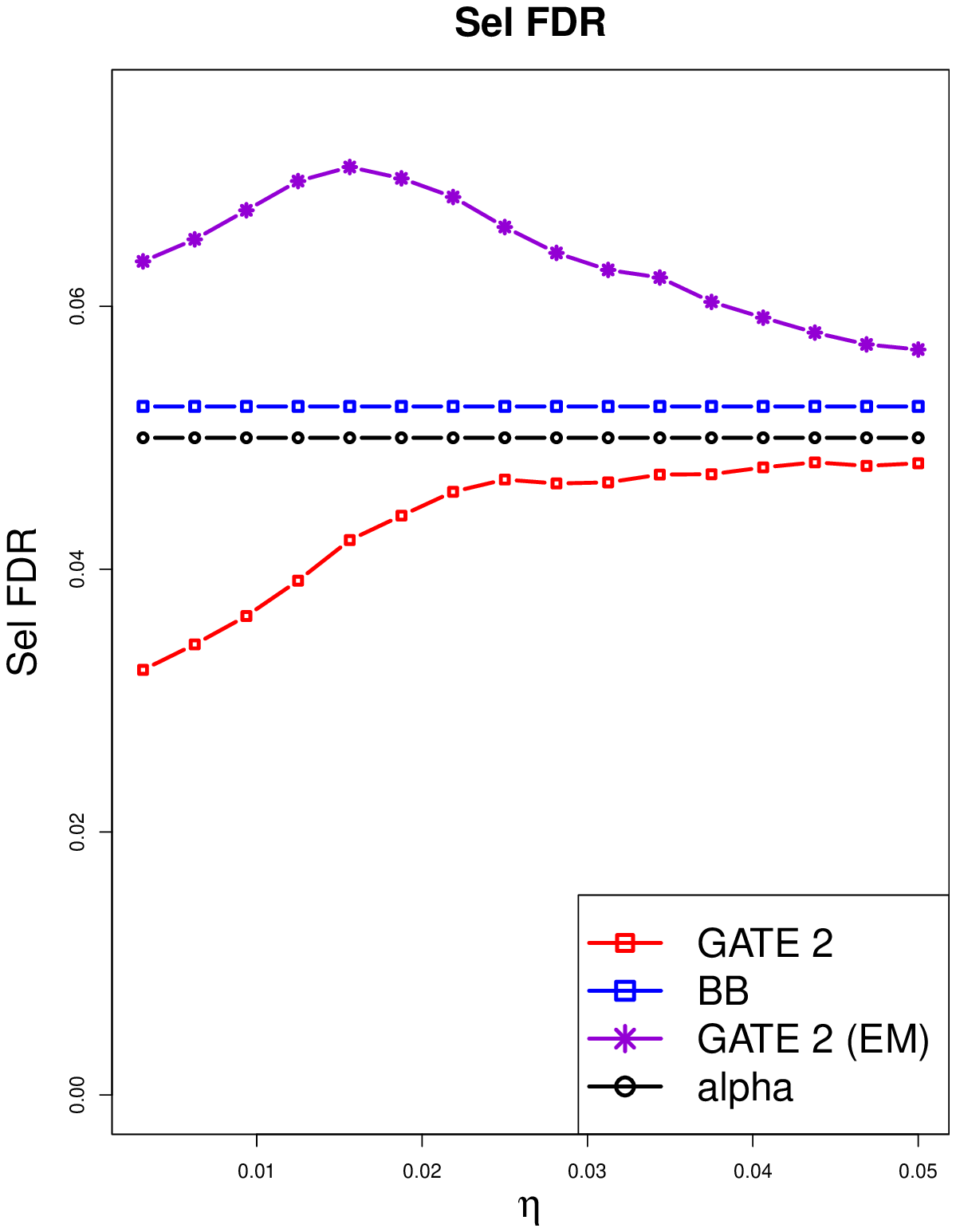}
   \includegraphics[height=40mm,width=40mm]{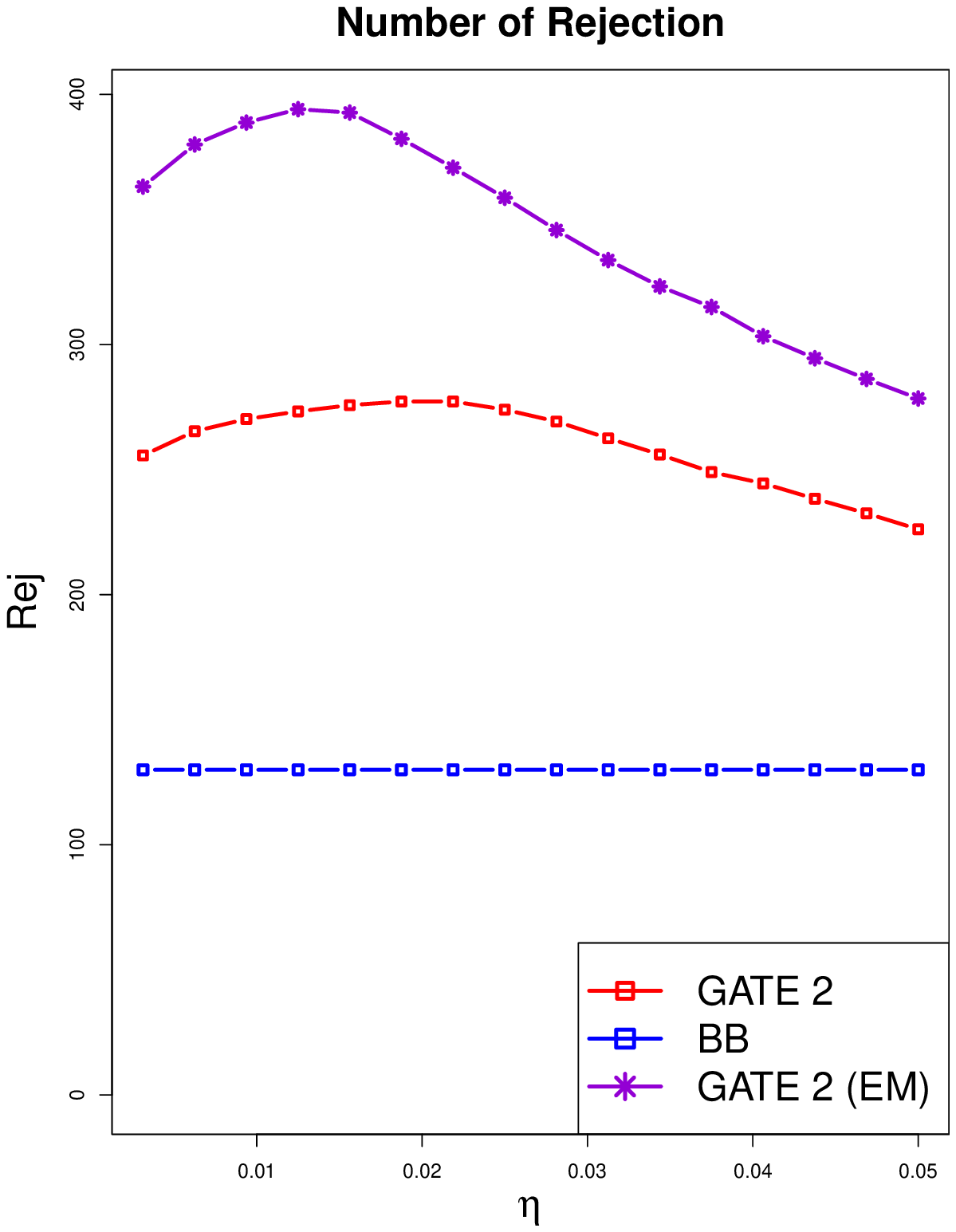}
   \caption{Performance of One-Way GATE 2 when $m=500, n=20, \pi_{1}=0.48$. The top panels correspond to cases when $K=1$ and the bottoms ones correspond to $K=2$.%The four panels correspond to four choices of $\pi_1$, which are 0.05, 0.3, 0.6, and 0.95 respectively.
   }\label{fig:gate4:s2}
\end{figure}

\begin{figure}[H]
   \centering
   \includegraphics[height=40mm,width=40mm]{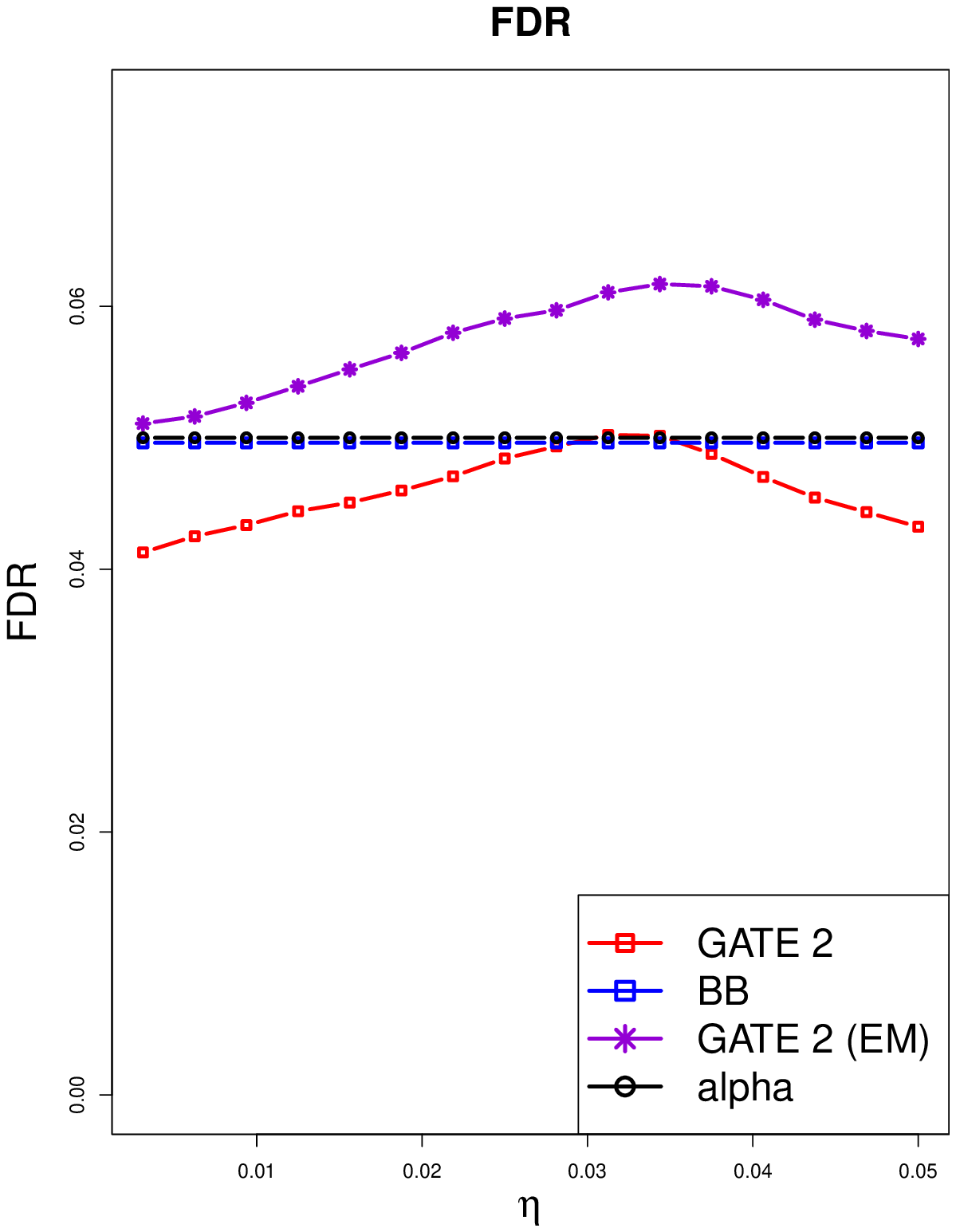}
   \includegraphics[height=40mm,width=40mm]{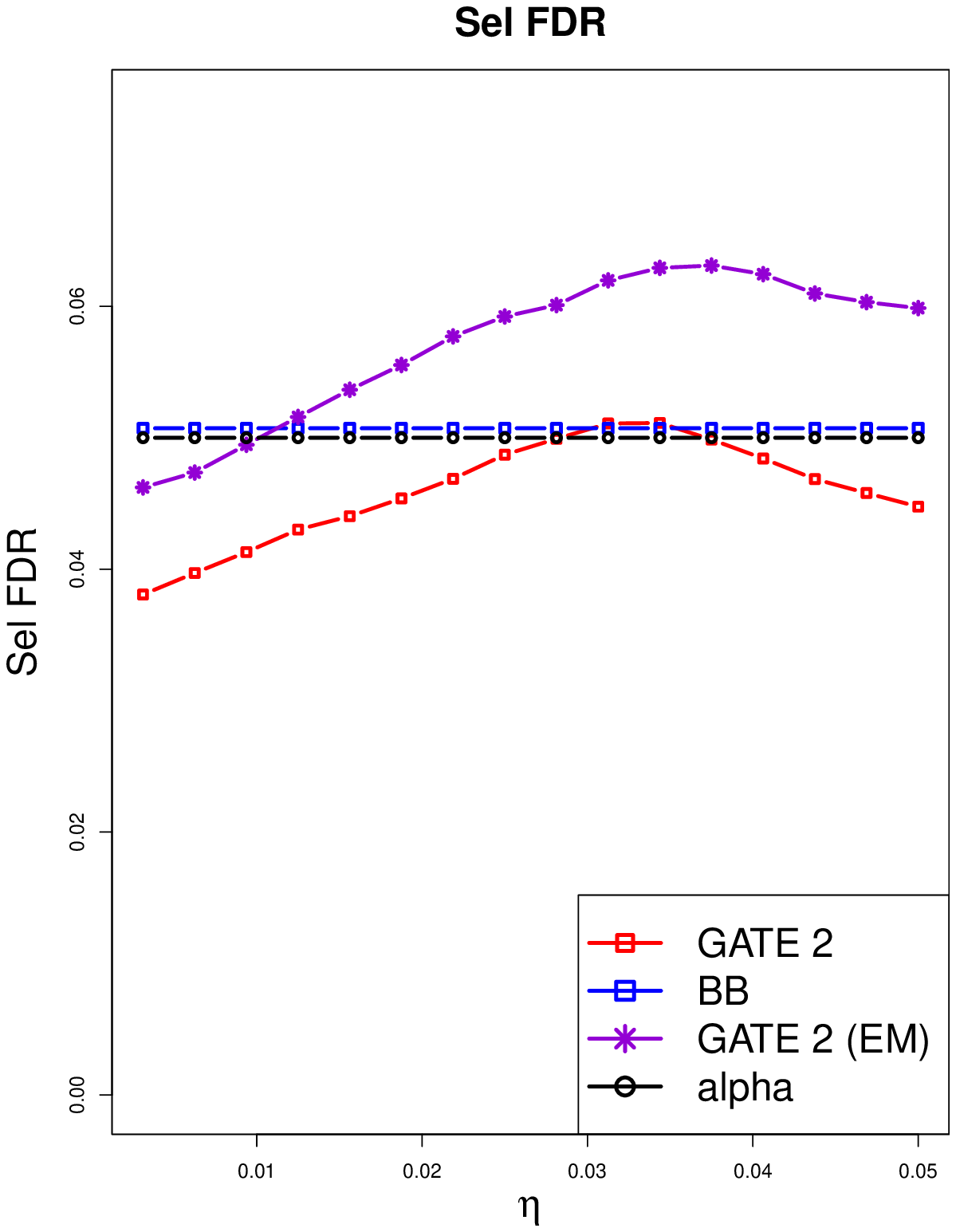}
   \includegraphics[height=40mm,width=40mm]{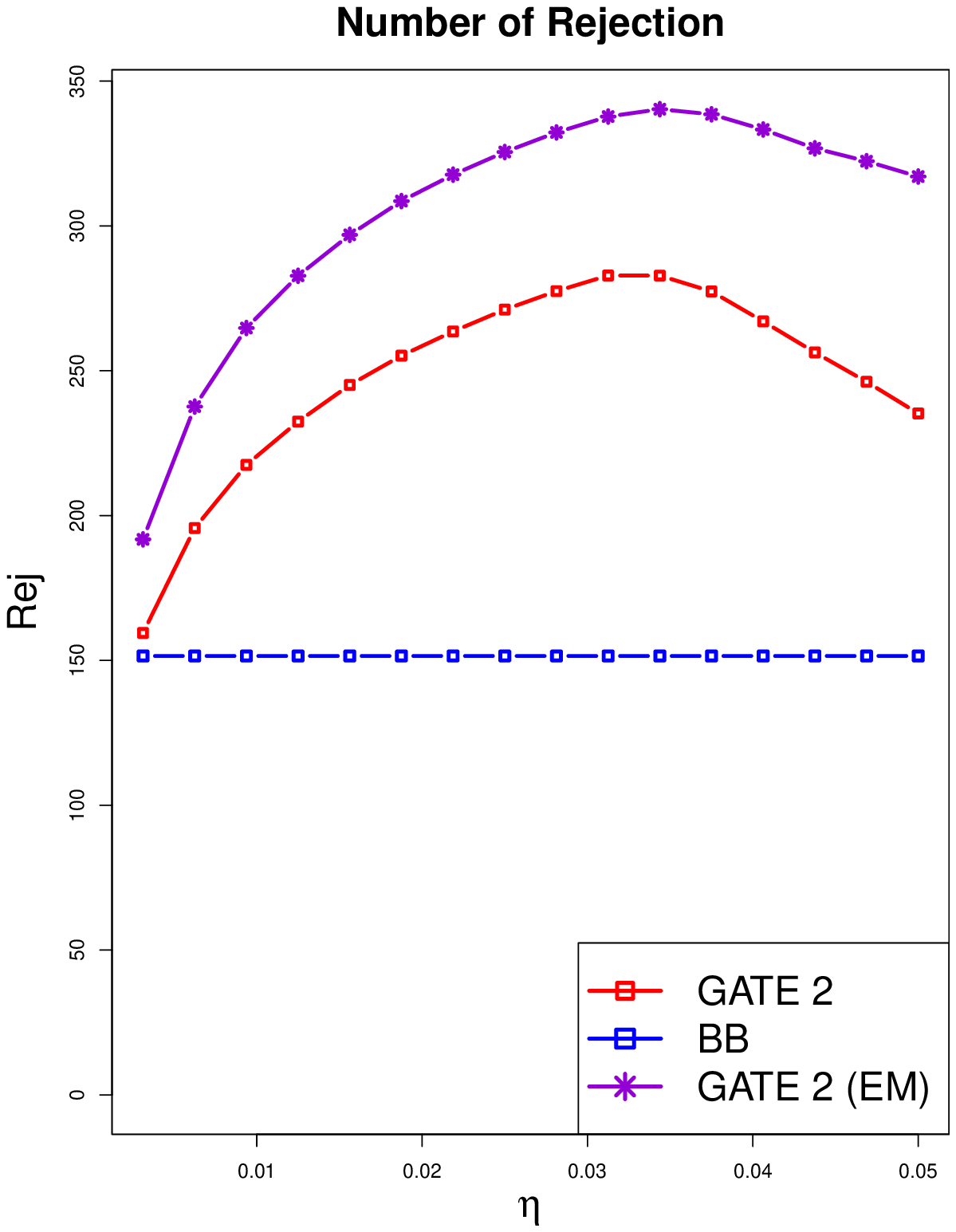}\\
   \includegraphics[height=40mm,width=40mm]{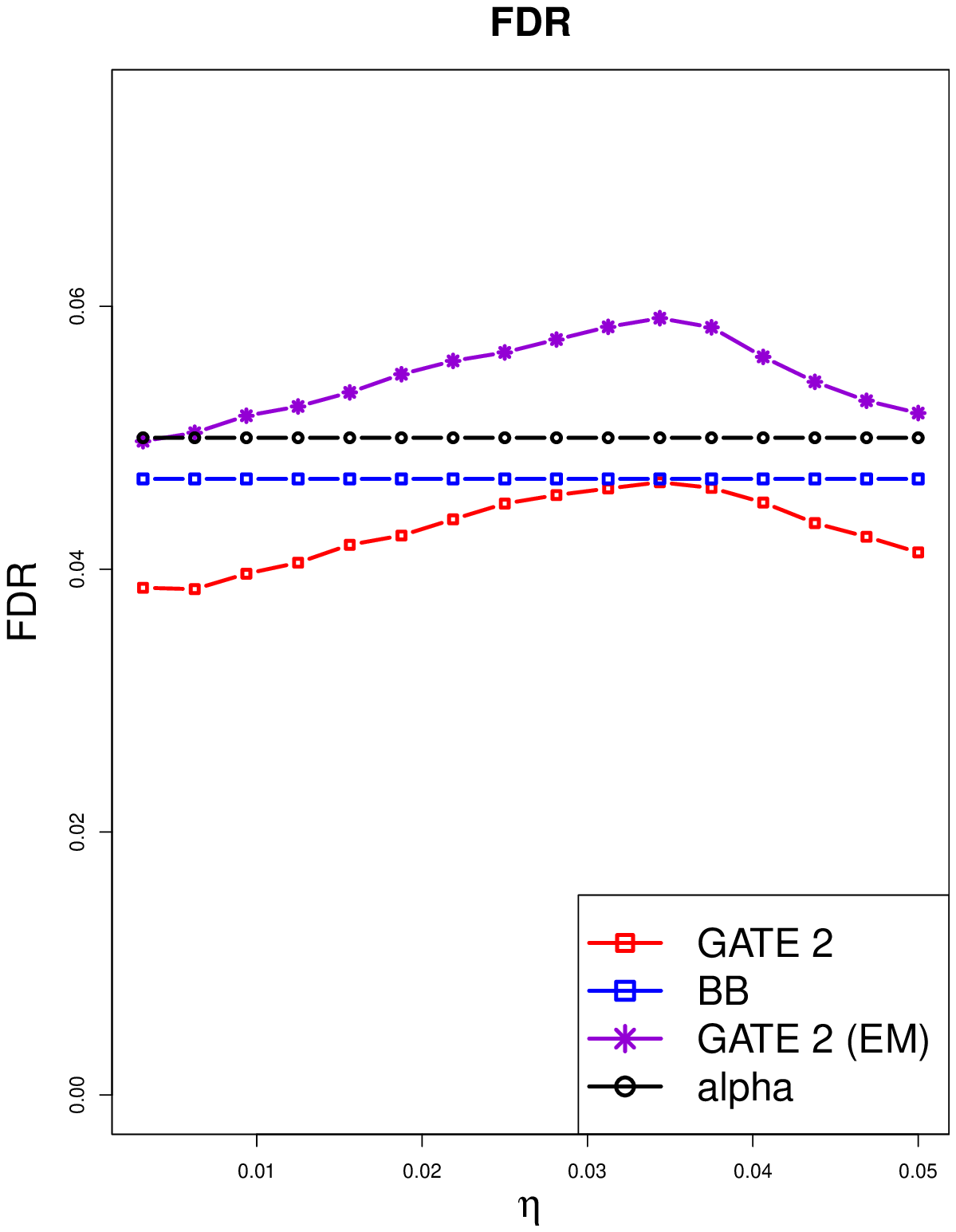}
   \includegraphics[height=40mm,width=40mm]{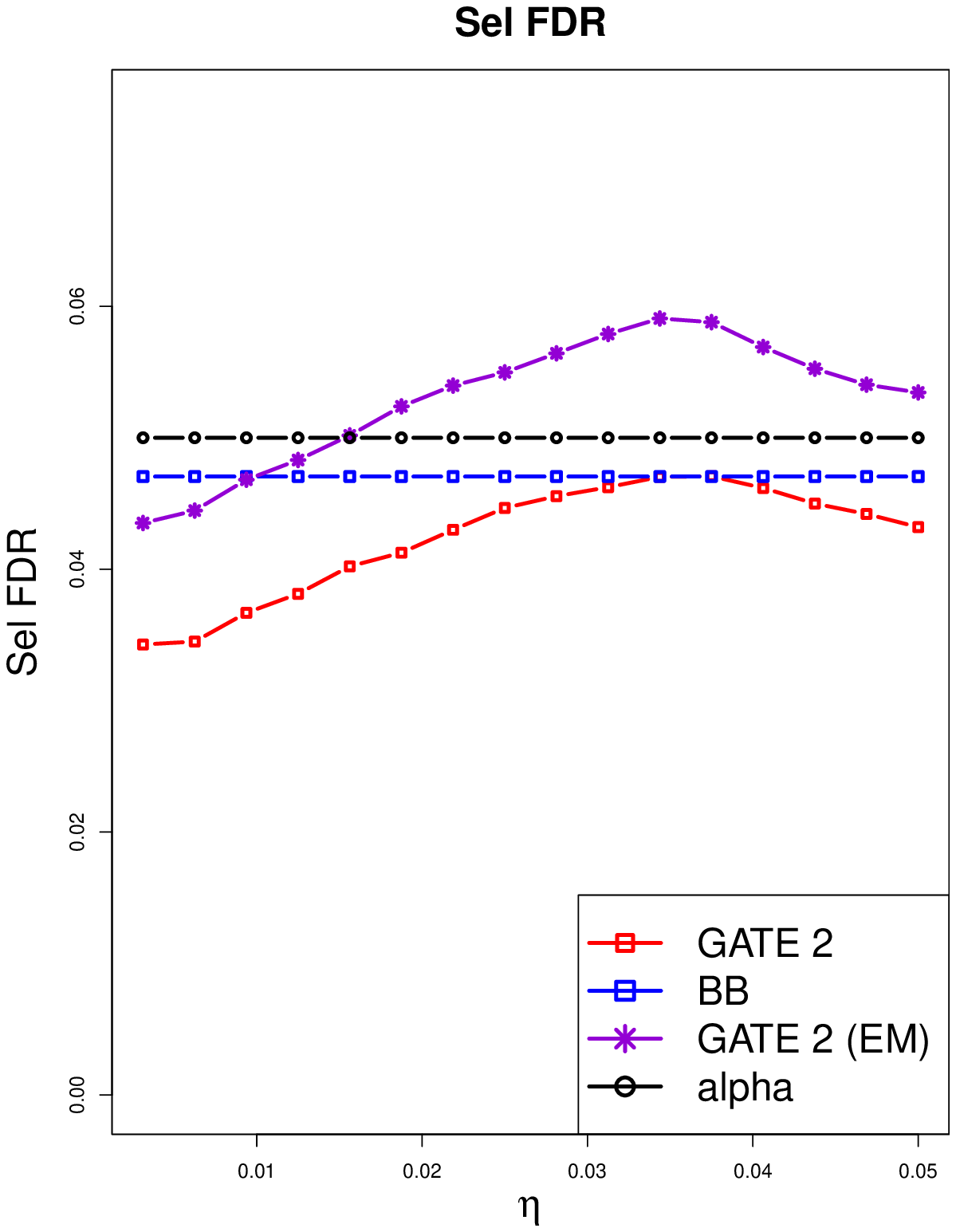}
   \includegraphics[height=40mm,width=40mm]{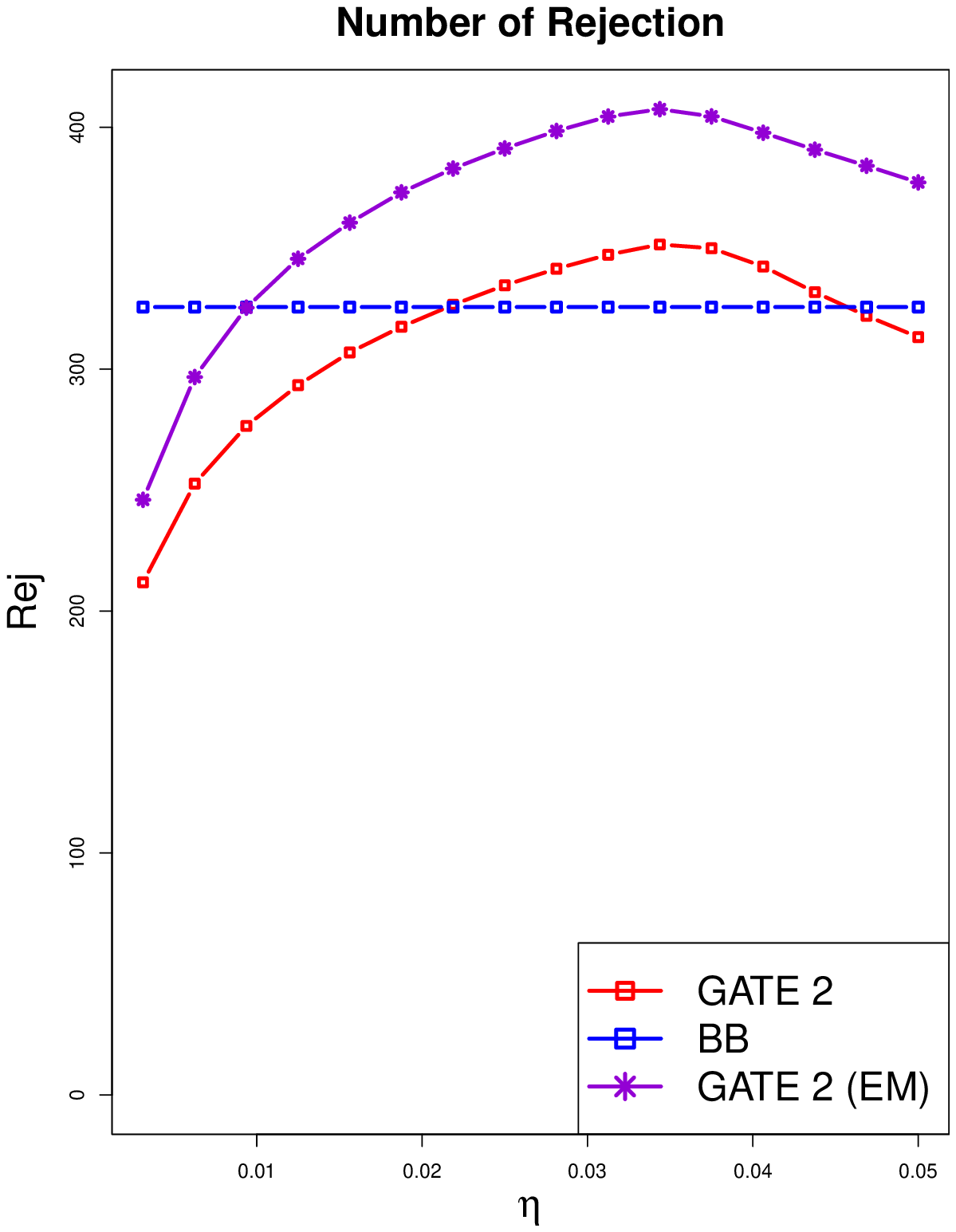}
   \caption{Performance of One-Way GATE 2 when $m=1000, n=5, \pi_{1}=0.165$. The top panels correspond to cases when $K=1$ and the bottoms ones correspond to $K=2$.%The four panels correspond to four choices of $\pi_1$, which are 0.05, 0.3, 0.6, and 0.95 respectively.
   }\label{fig:gate4:s3}
\end{figure}

\begin{figure}[H]
   \centering
   \includegraphics[height=40mm,width=40mm]{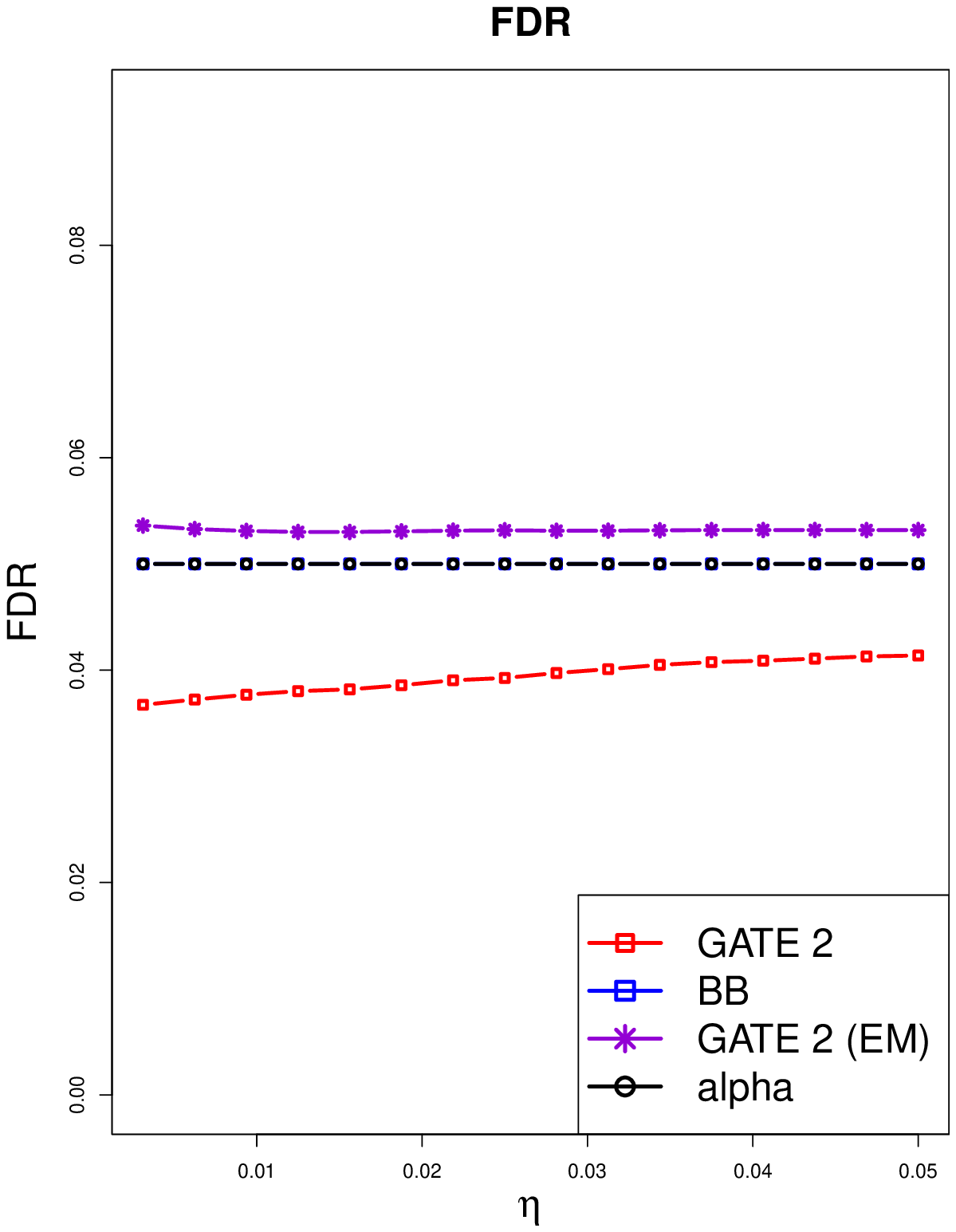}
   \includegraphics[height=40mm,width=40mm]{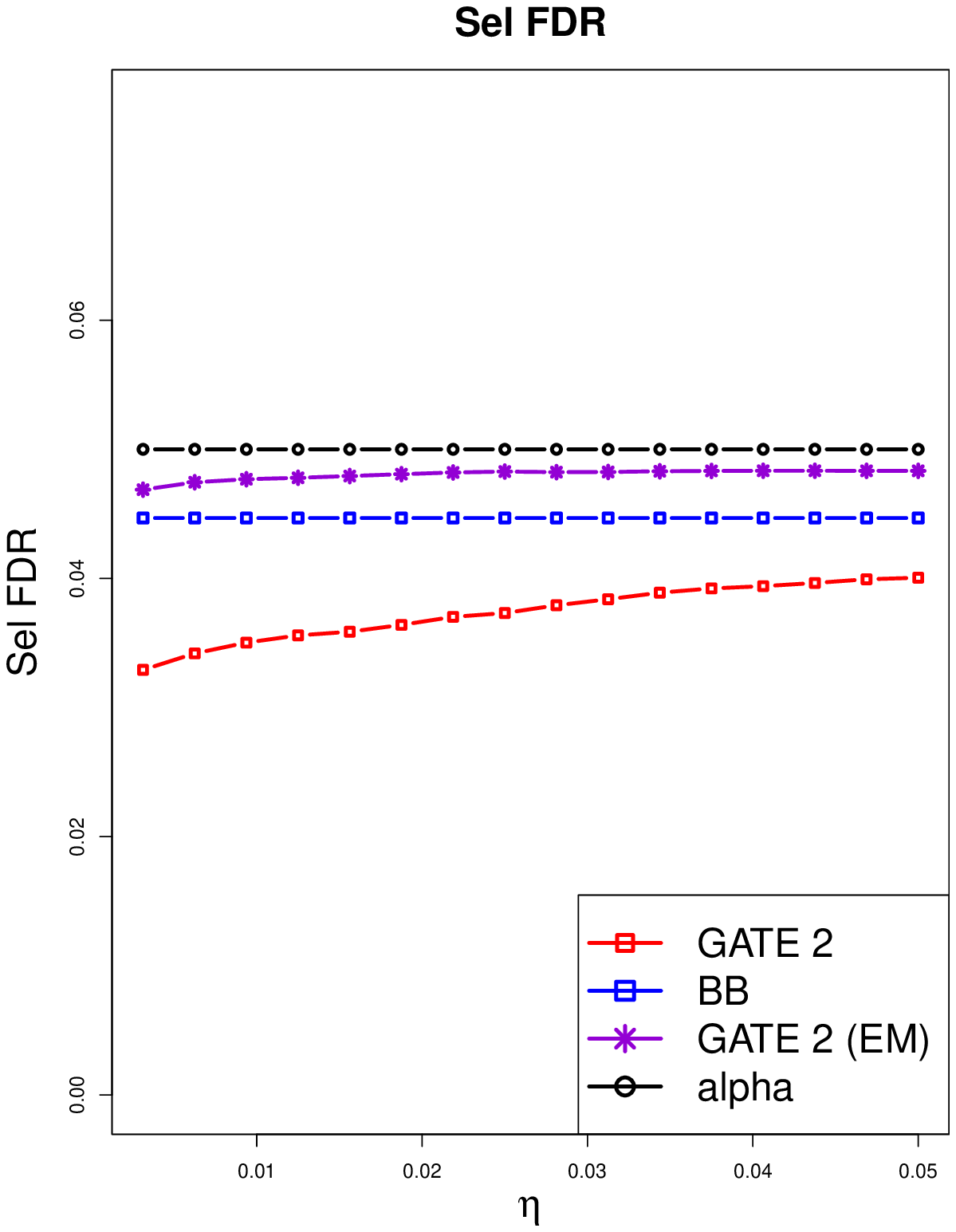}
   \includegraphics[height=40mm,width=40mm]{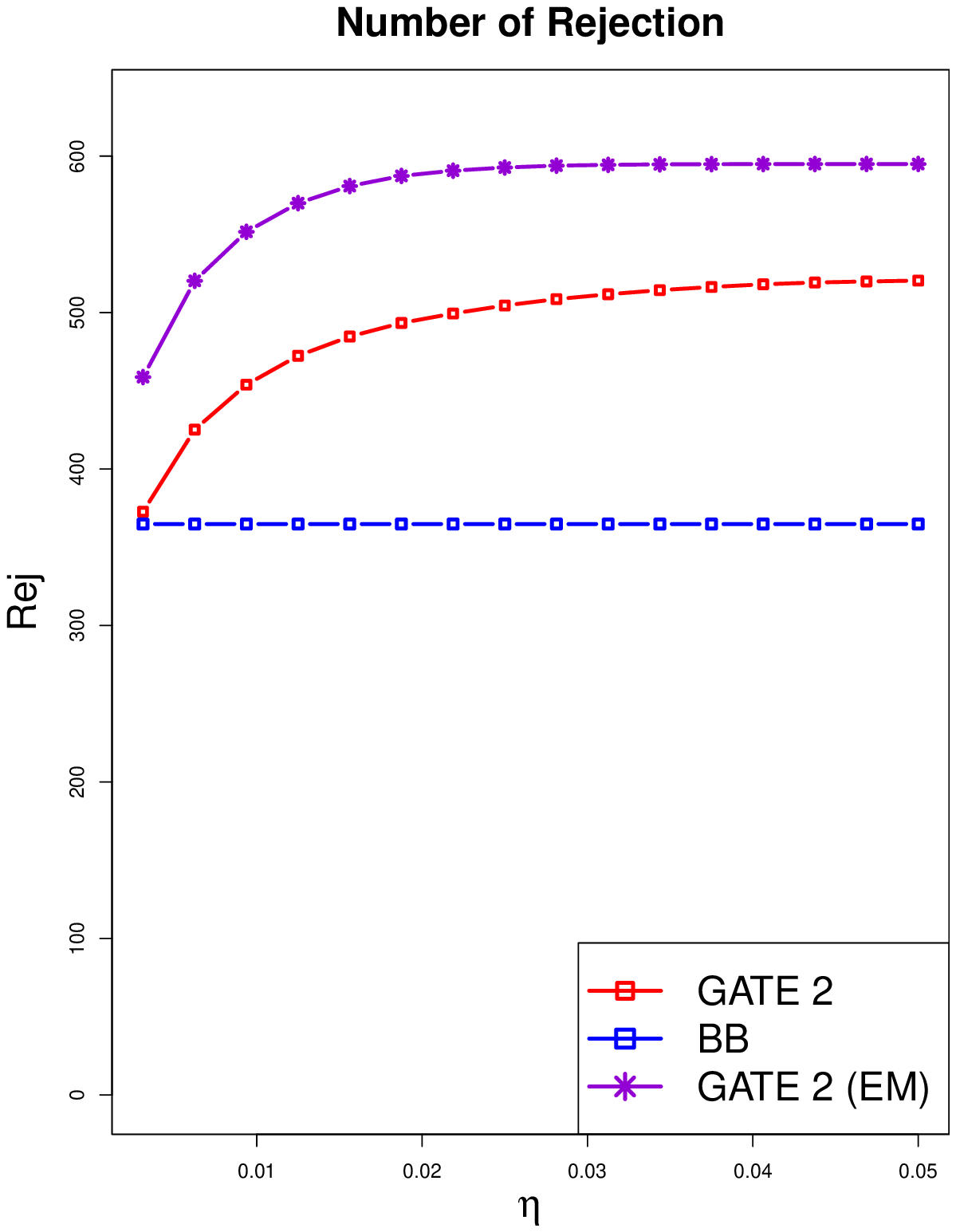}\\
   \includegraphics[height=40mm,width=40mm]{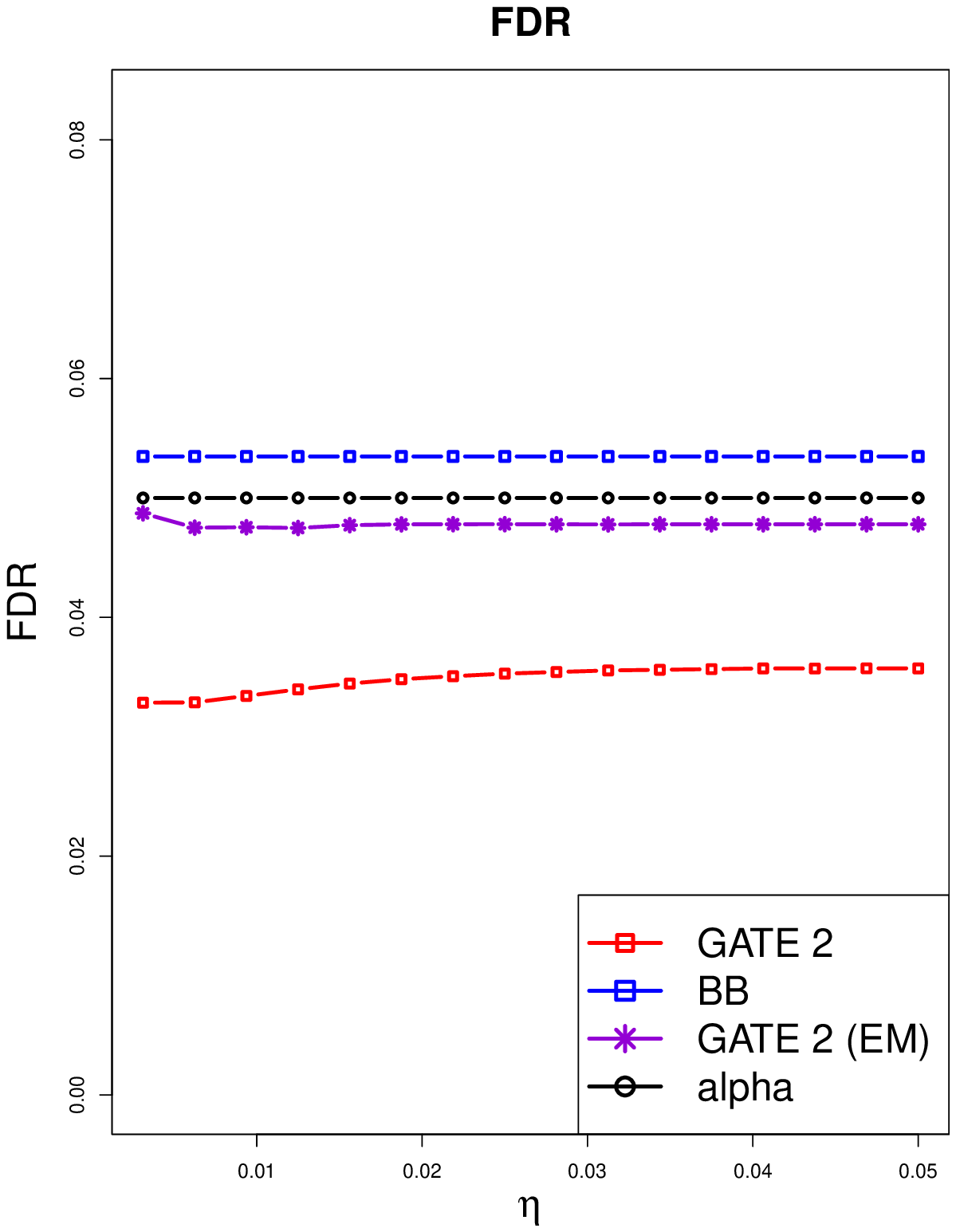}
   \includegraphics[height=40mm,width=40mm]{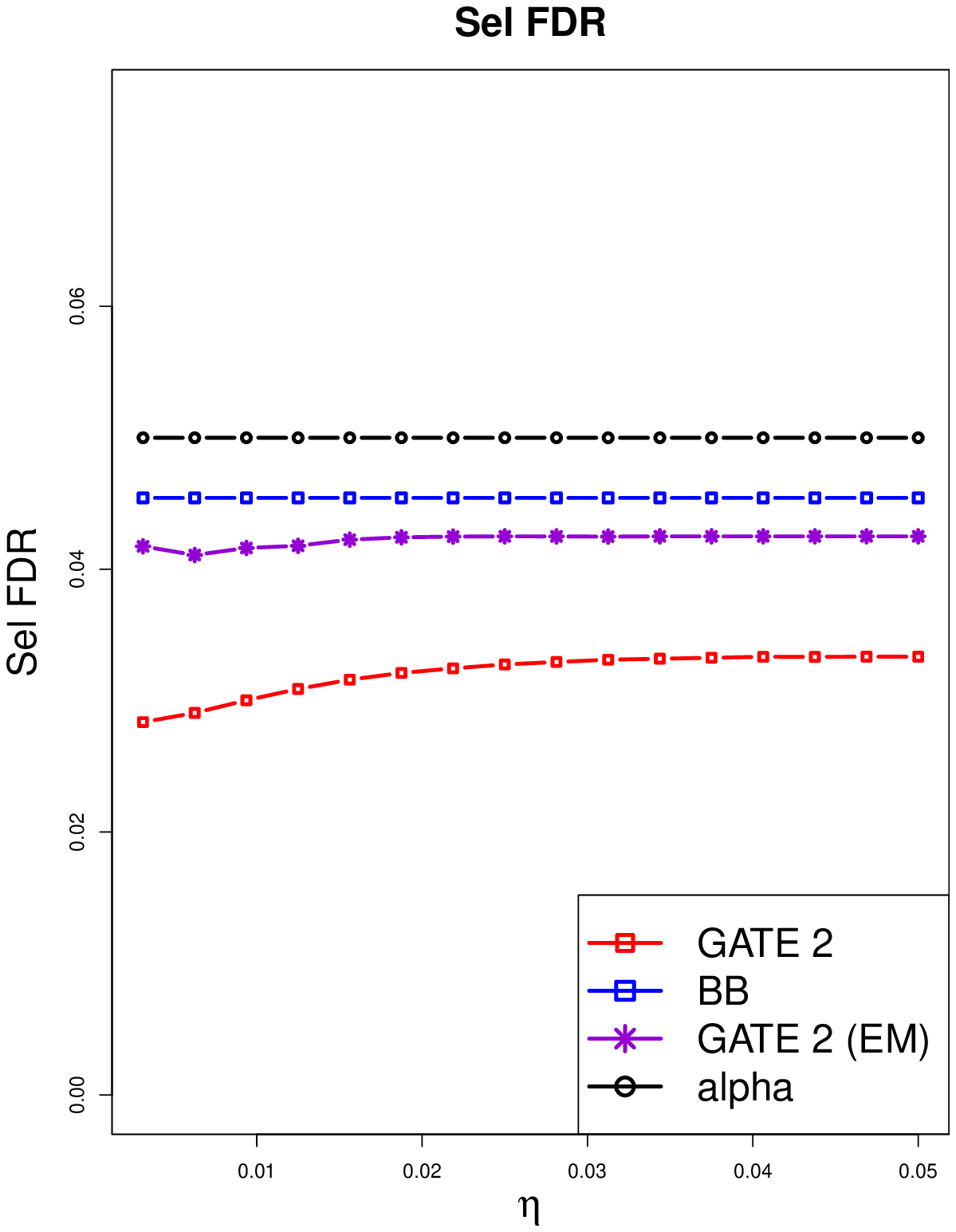}
   \includegraphics[height=40mm,width=40mm]{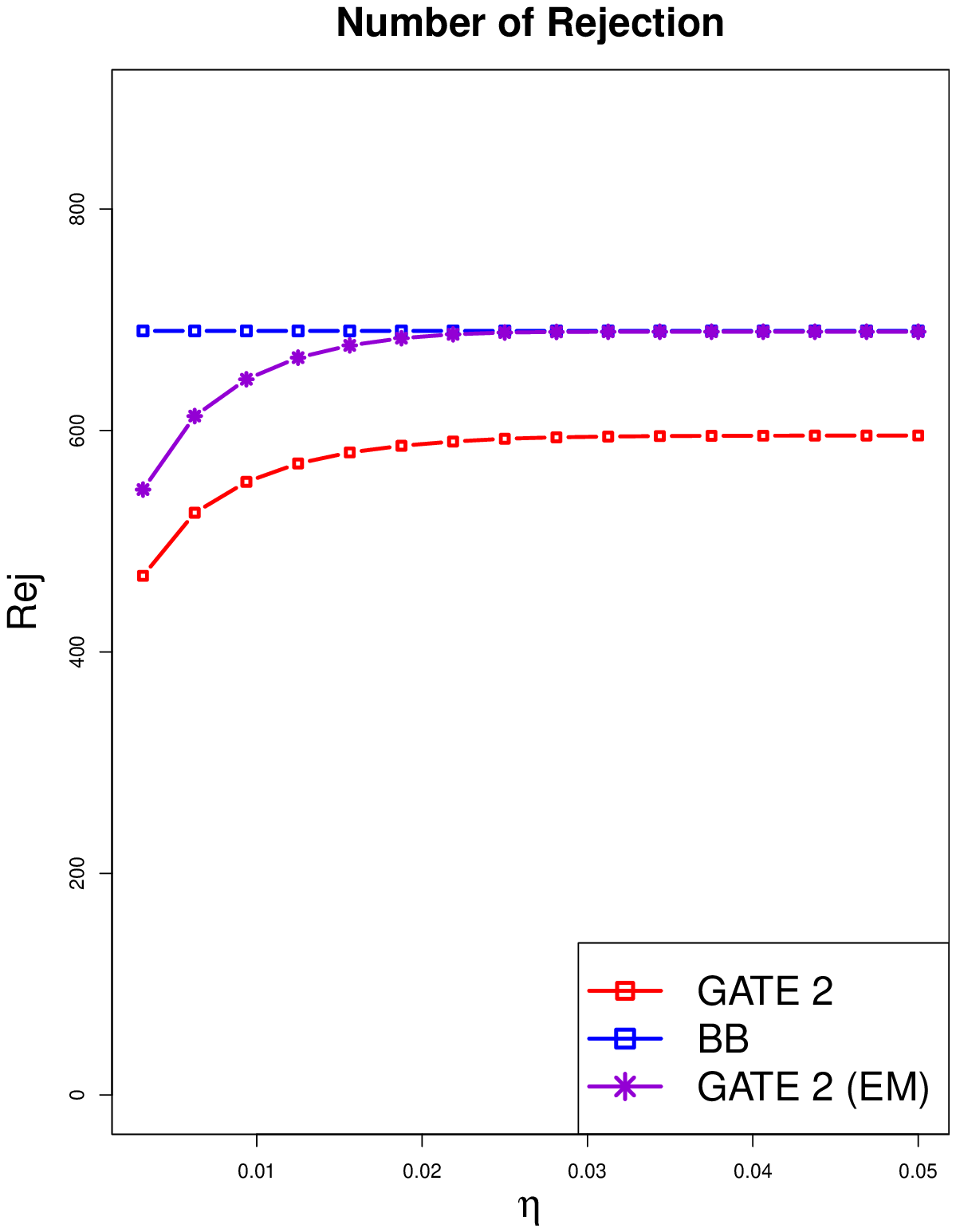}
   \caption{Performance of One-Way GATE 2 when $m=1000, n=5, \pi_{1}=0.308$. The top panels correspond to cases when $K=1$ and the bottoms ones correspond to $K=2$.%The four panels correspond to four choices of $\pi_1$, which are 0.05, 0.3, 0.6, and 0.95 respectively.
   }\label{fig:gate4:s4}
\end{figure}

\begin{figure}[H]
   \centering
   \includegraphics[height=40mm,width=40mm]{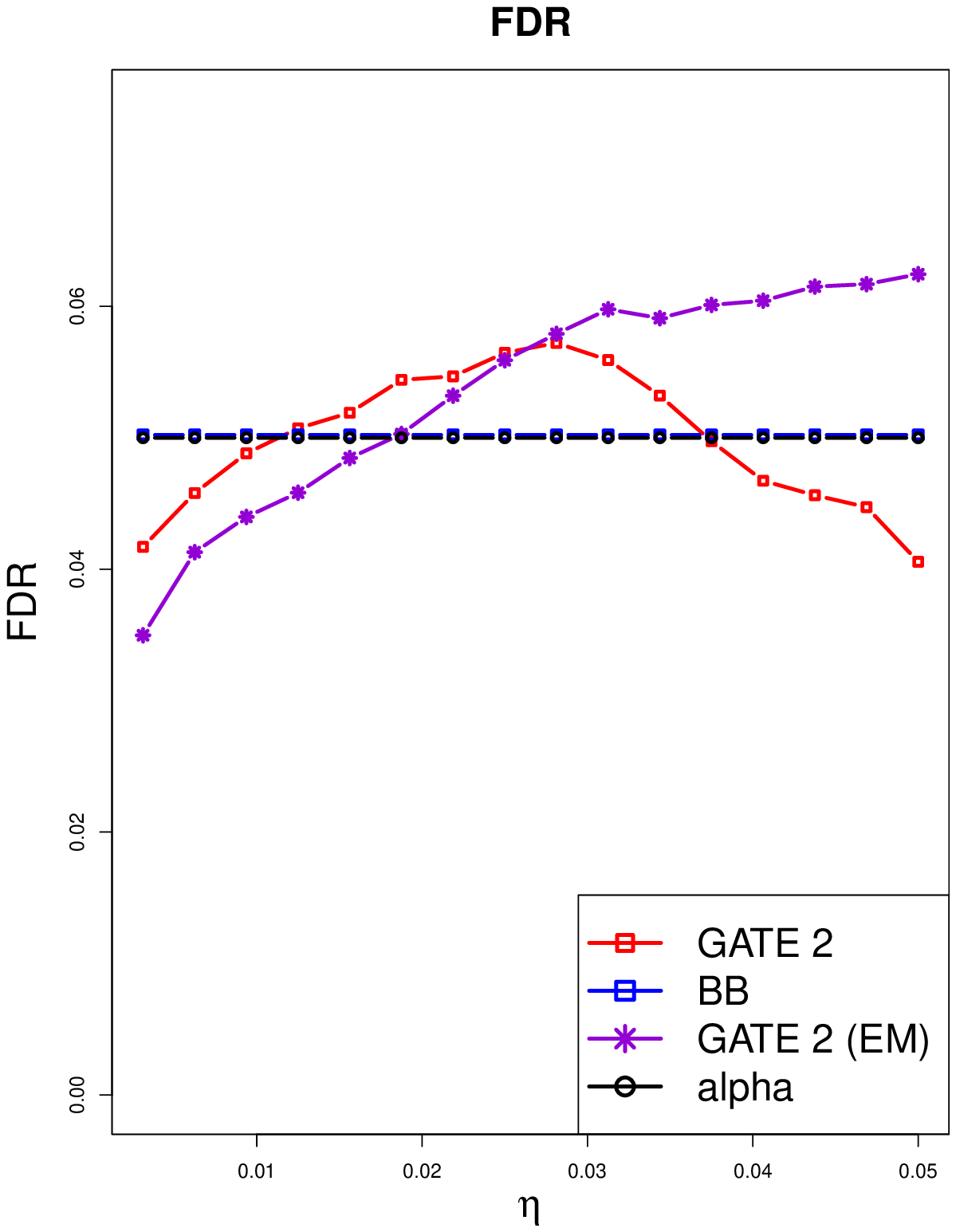}
   \includegraphics[height=40mm,width=40mm]{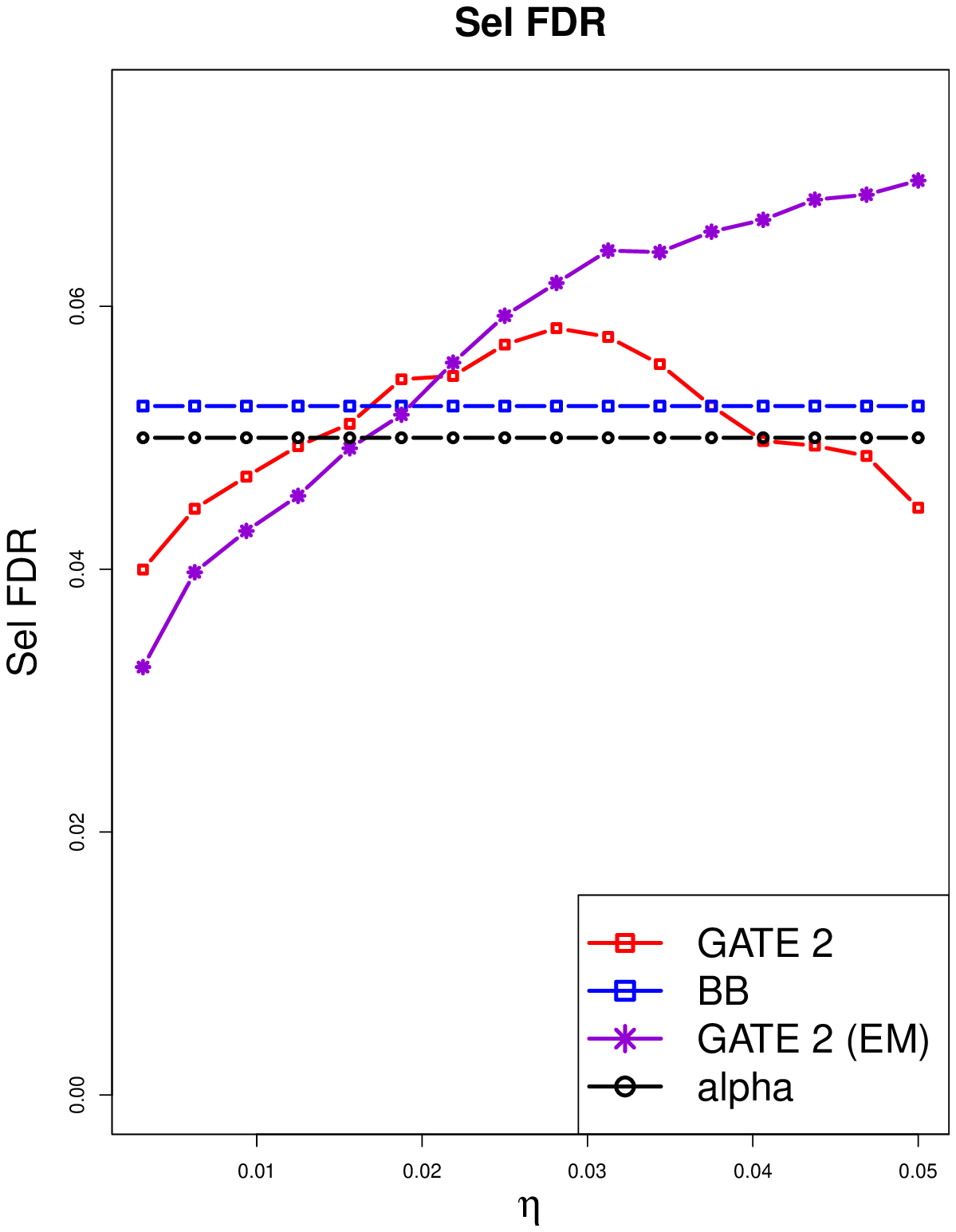}
   \includegraphics[height=40mm,width=40mm]{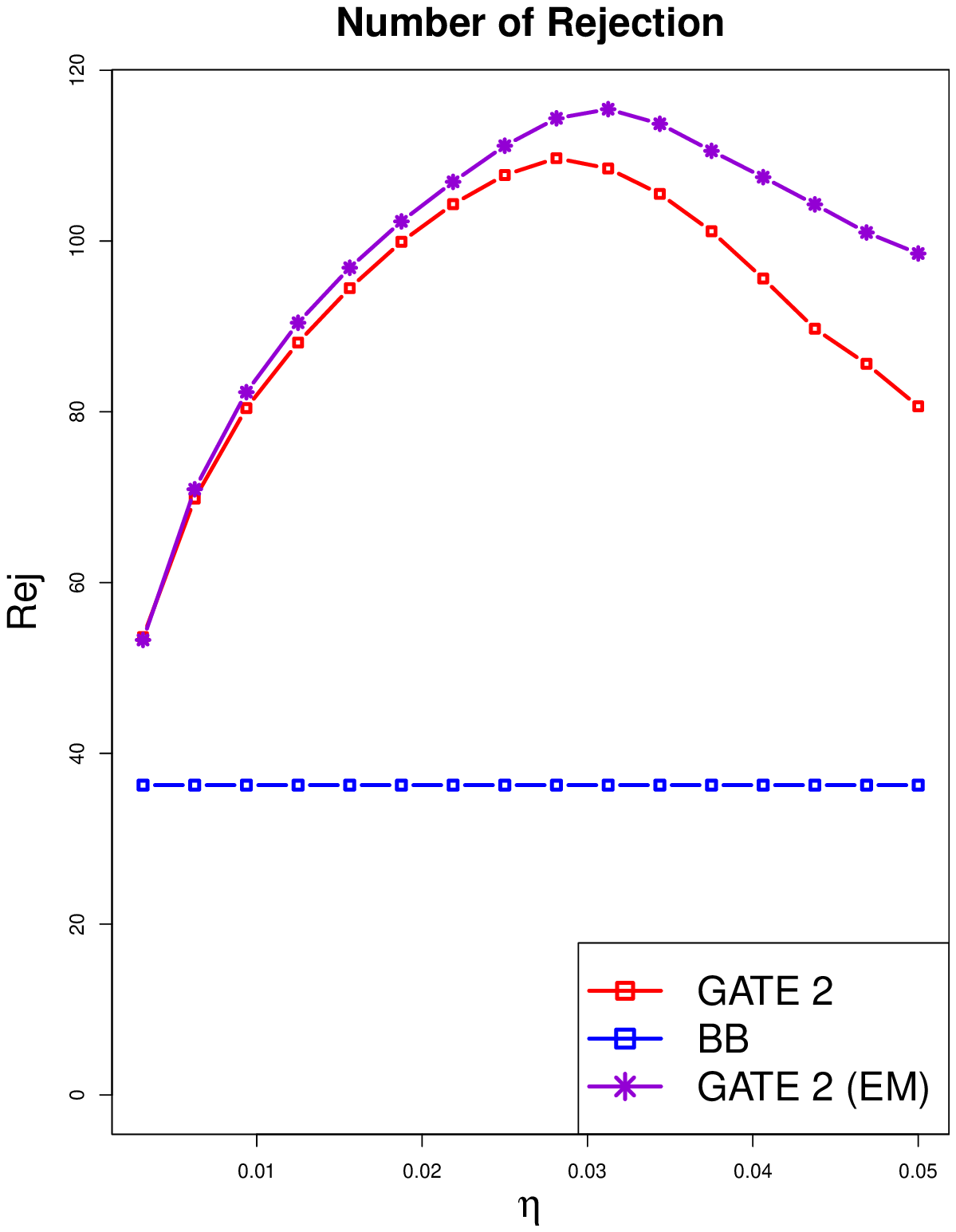}\\
   \includegraphics[height=40mm,width=40mm]{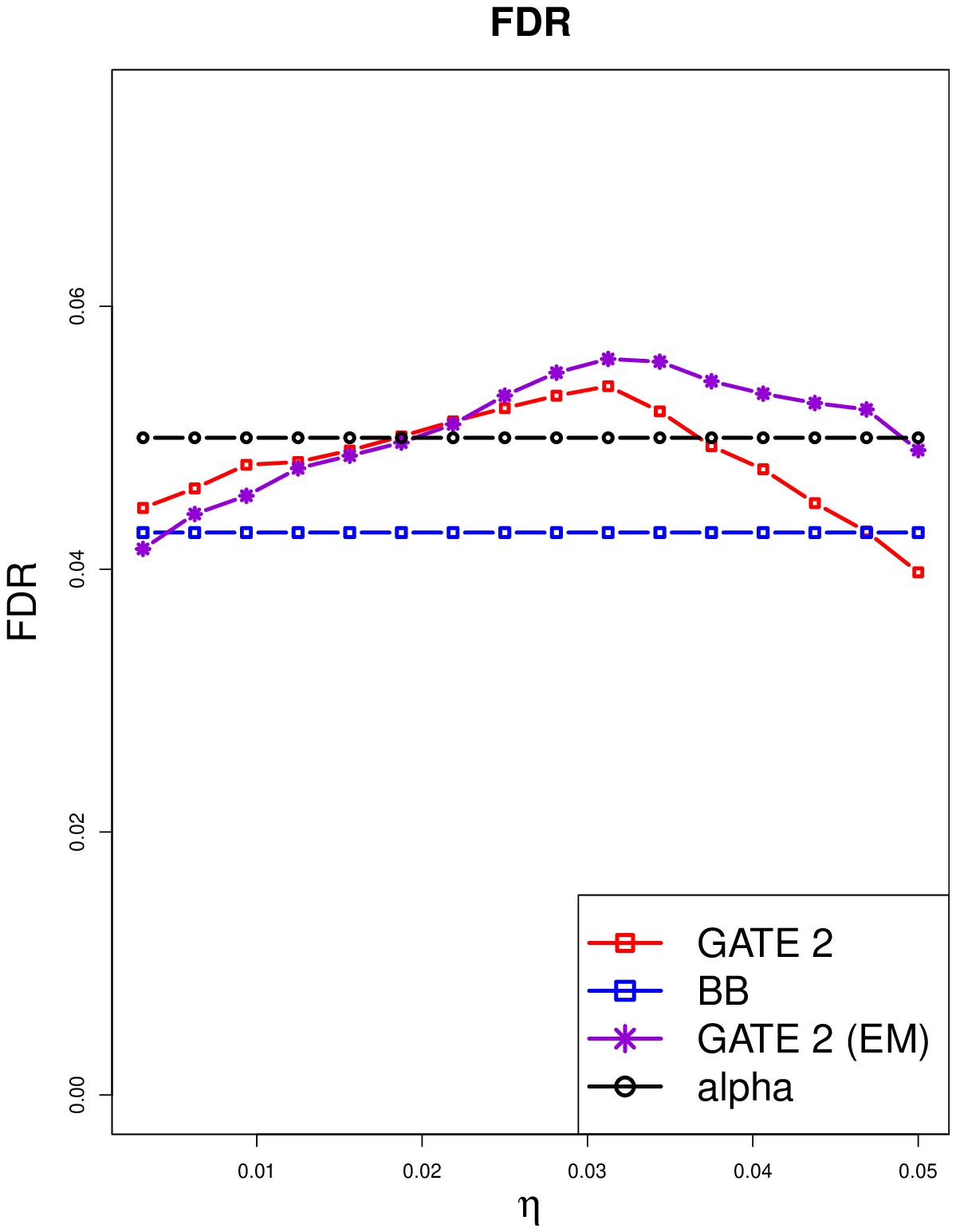}
   \includegraphics[height=40mm,width=40mm]{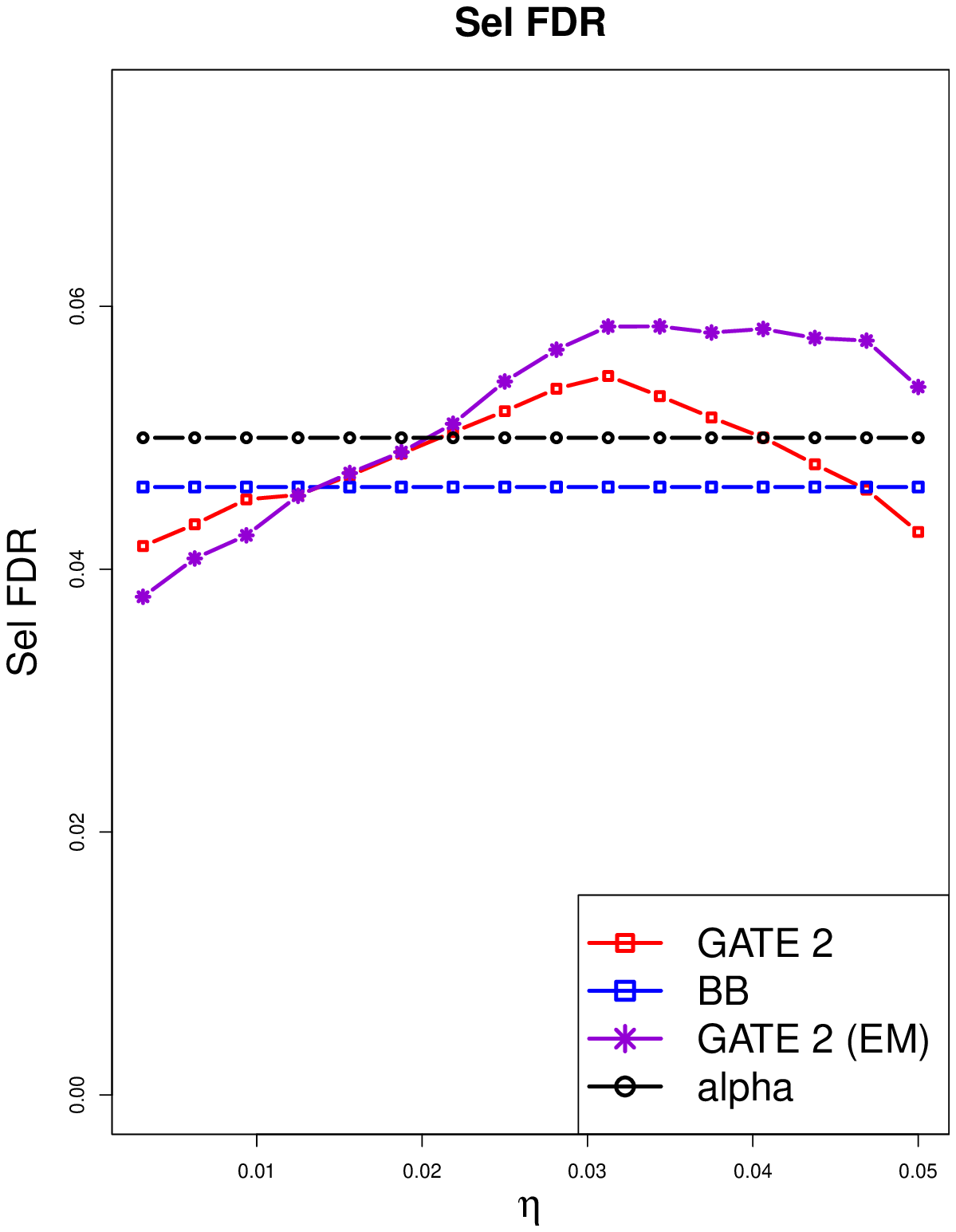}
   \includegraphics[height=40mm,width=40mm]{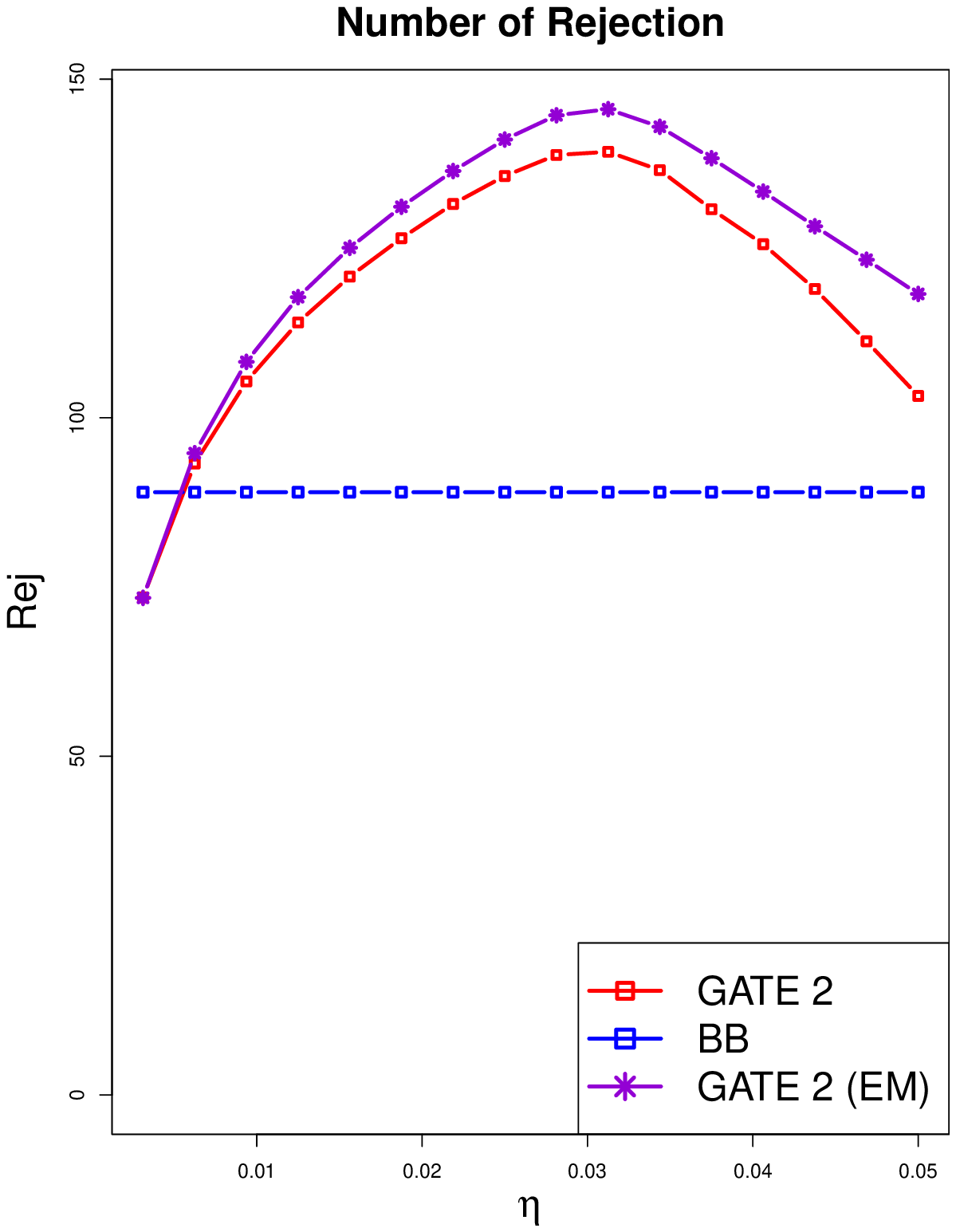}
   \caption{Performance of One-Way GATE 2 when $m=1000, n=5, \pi_{1}=0.47$. The top panels correspond to cases when $K=1$ and the bottoms ones correspond to $K=2$.%The four panels correspond to four choices of $\pi_1$, which are 0.05, 0.3, 0.6, and 0.95 respectively.
   }\label{fig:gate4:s5}
\end{figure}

\subsection{Additional results on the comparison of estimated local fdrs.}

\begin{figure}[H]
  \centering
  \includegraphics[height=50mm,width=50mm]{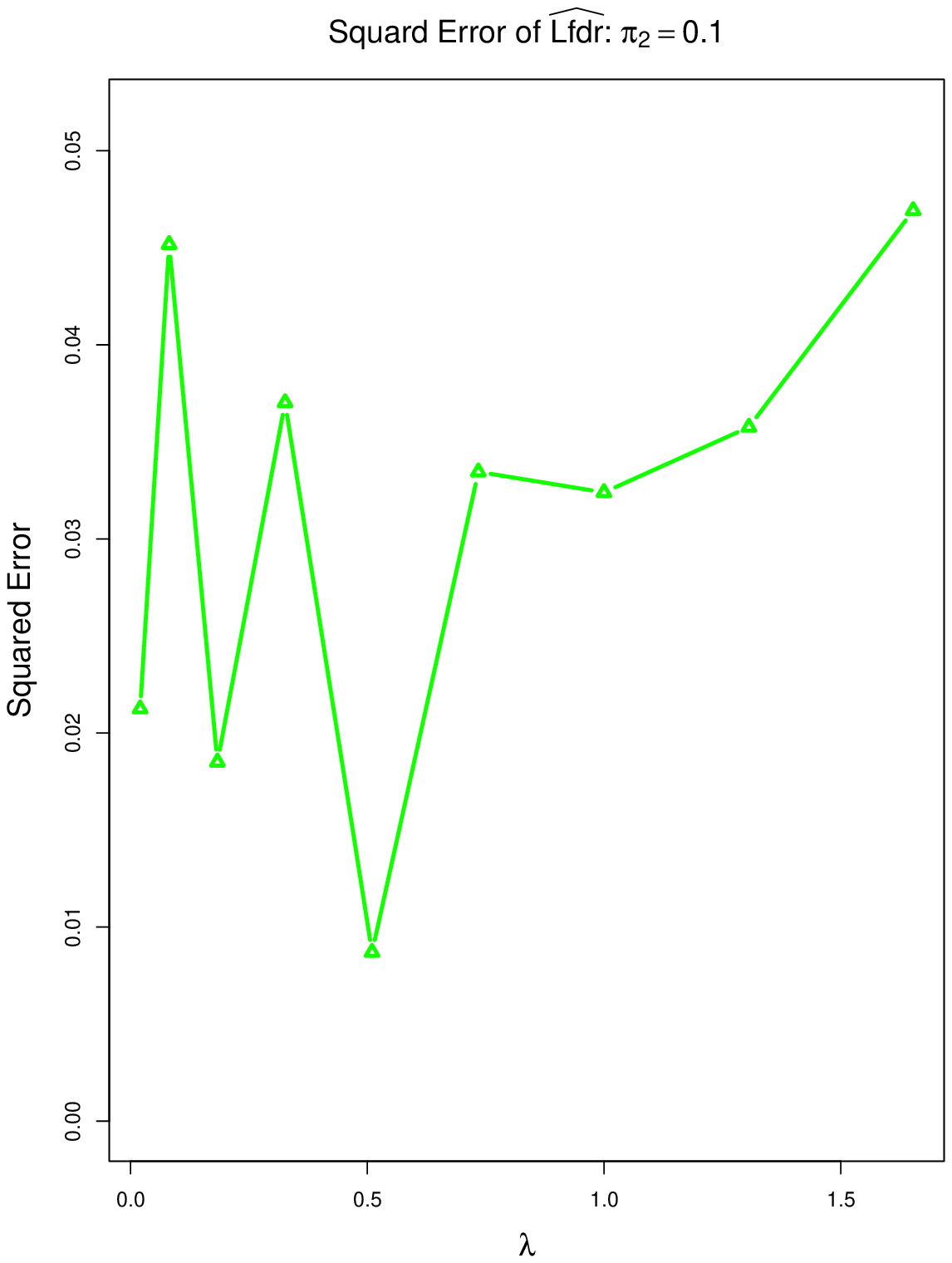}
    \includegraphics[height=50mm,width=50mm]{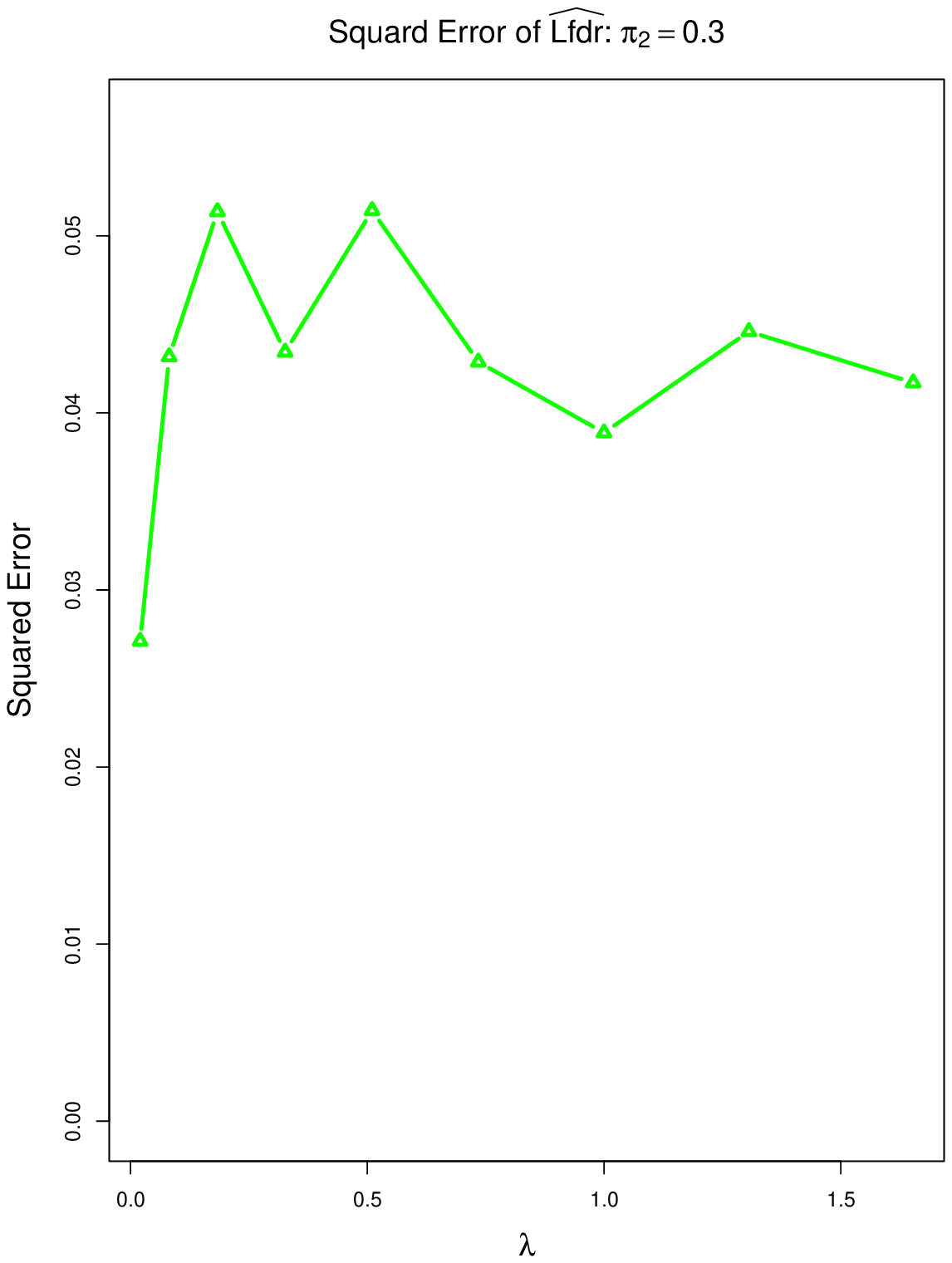}\\
  \includegraphics[height=50mm,width=50mm]{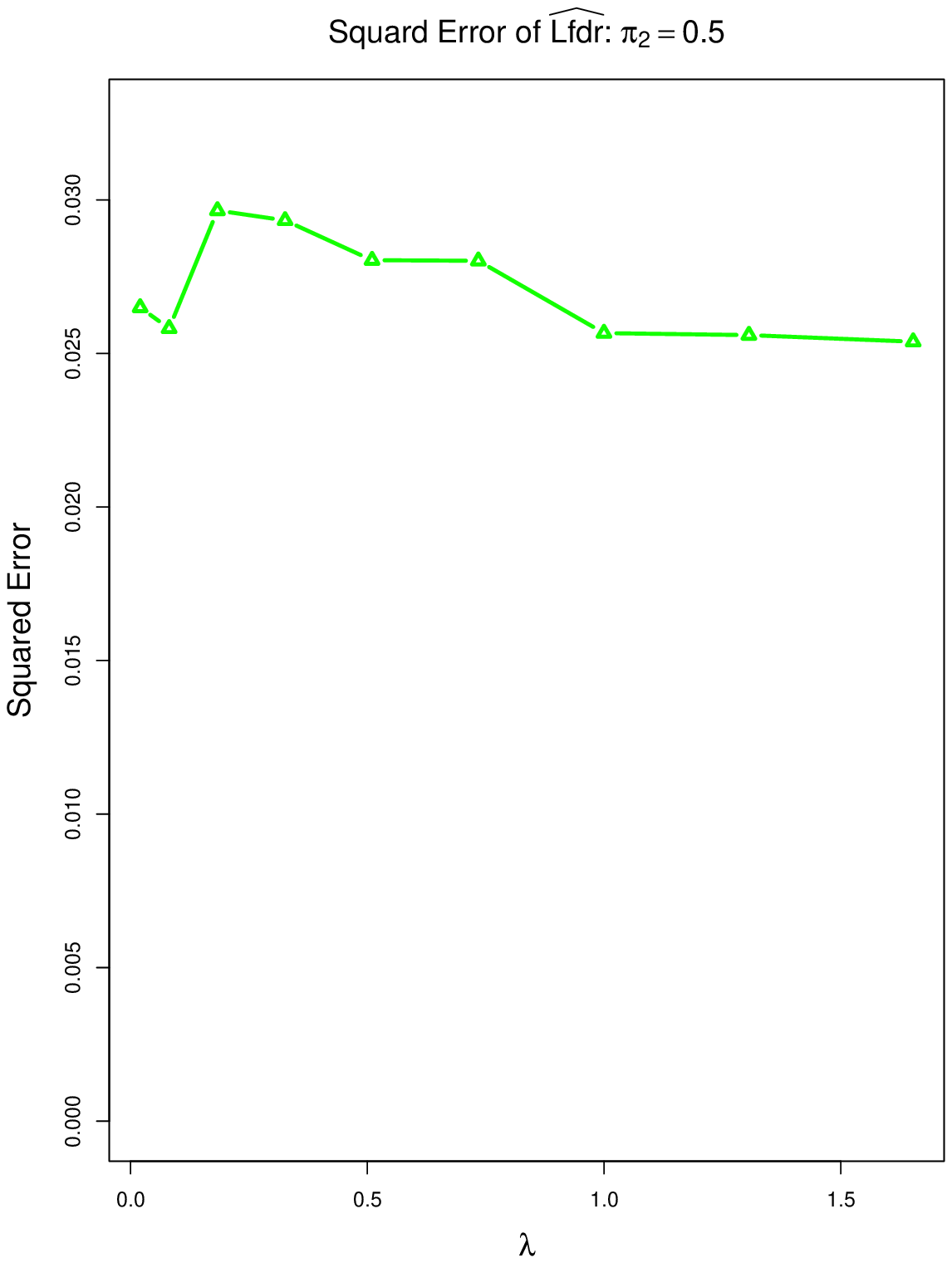}
  \includegraphics[height=50mm,width=50mm]{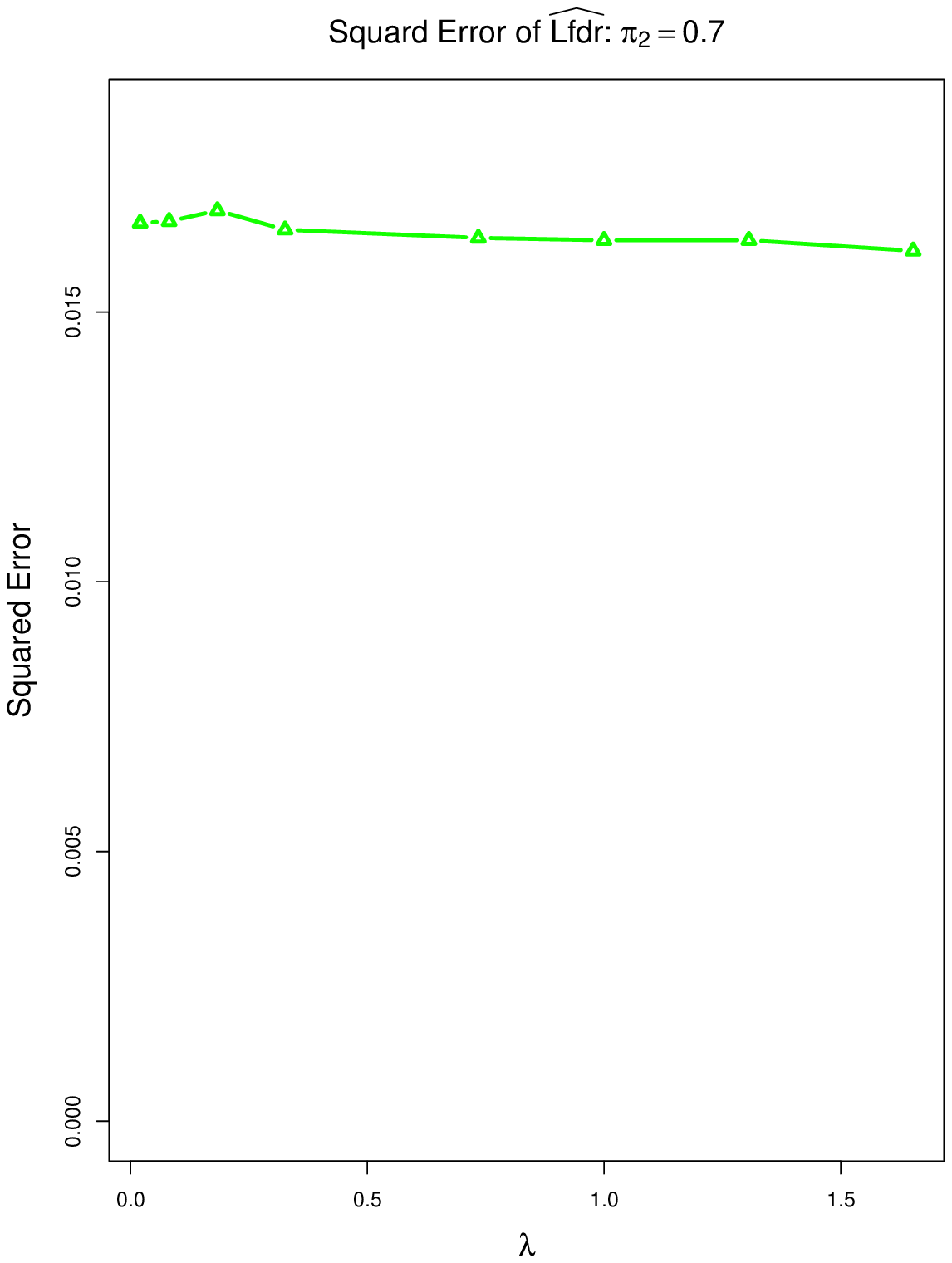}
  \caption{ The mean squared error for the estimation of the local fdr scores using the proposed method: $m=1000, n=5, L=2, \pi_2=0.1, 0.3, 0.5$ and $0.7$ respectively.
  }\label{fig:gate1:lfdr:s2}
\end{figure}

\begin{figure}[H]
  \centering
  \includegraphics[height=50mm,width=50mm]{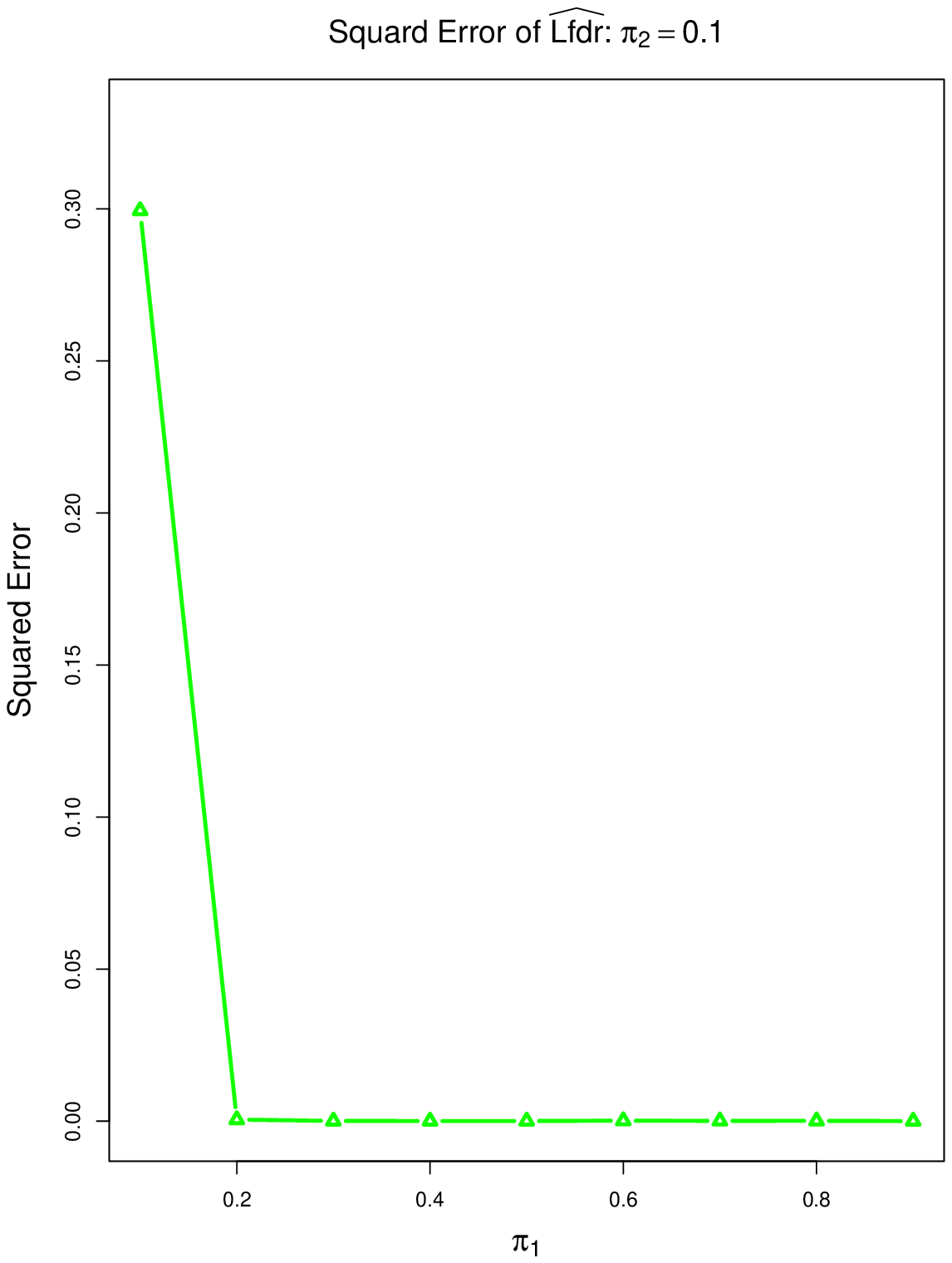}
  \includegraphics[height=50mm,width=50mm]{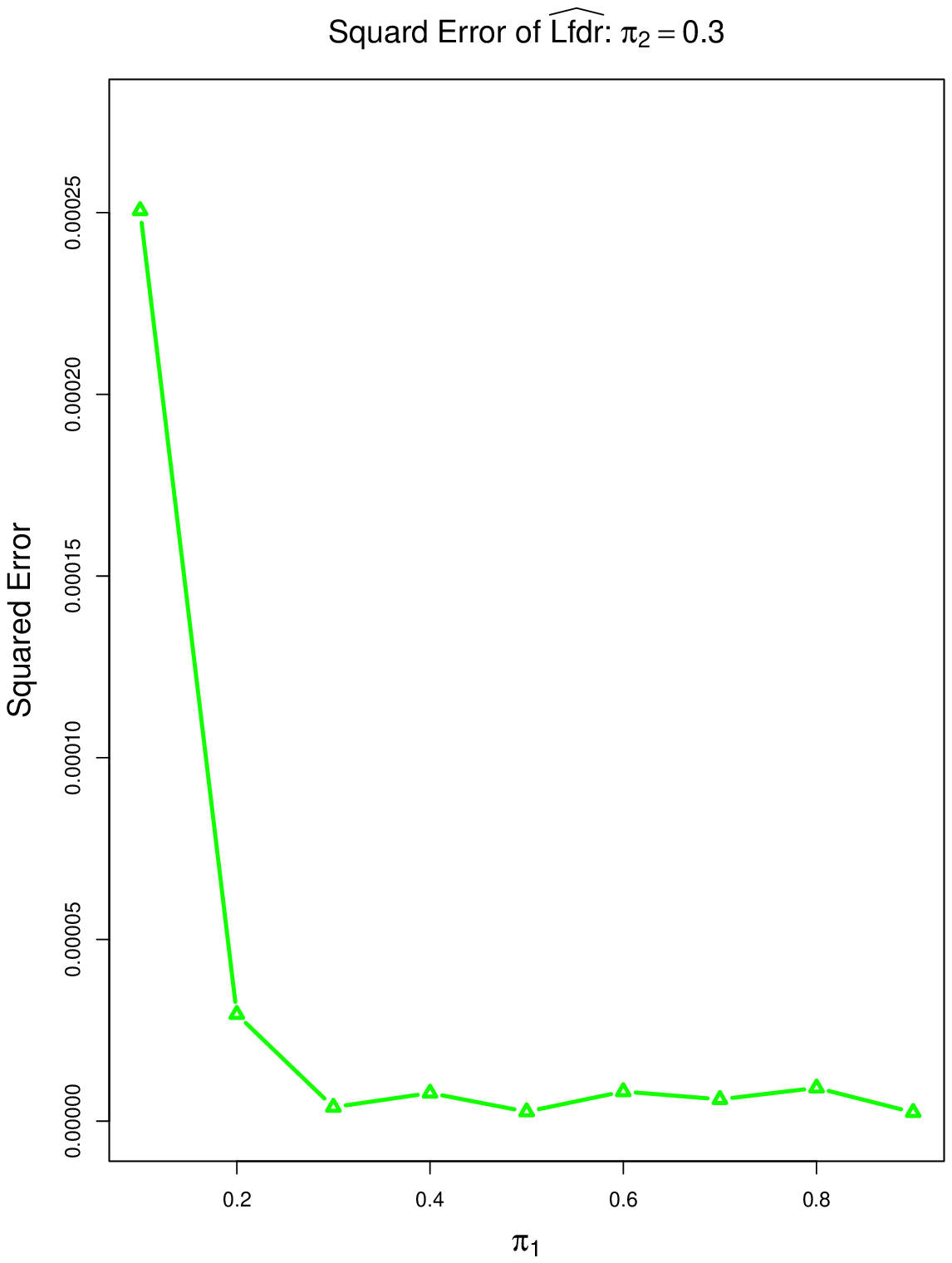}\\
  \includegraphics[height=50mm,width=50mm]{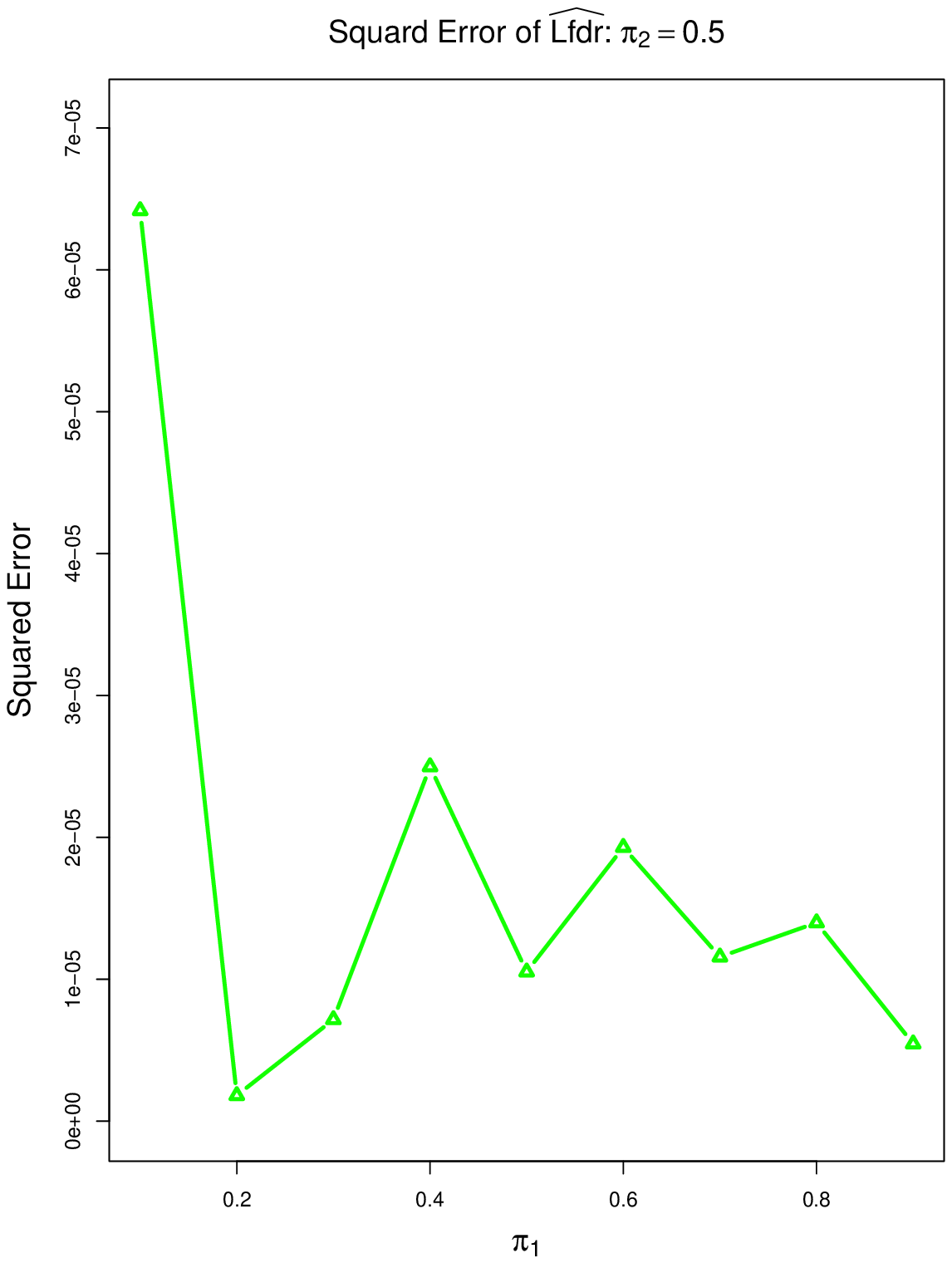}
  \includegraphics[height=50mm,width=50mm]{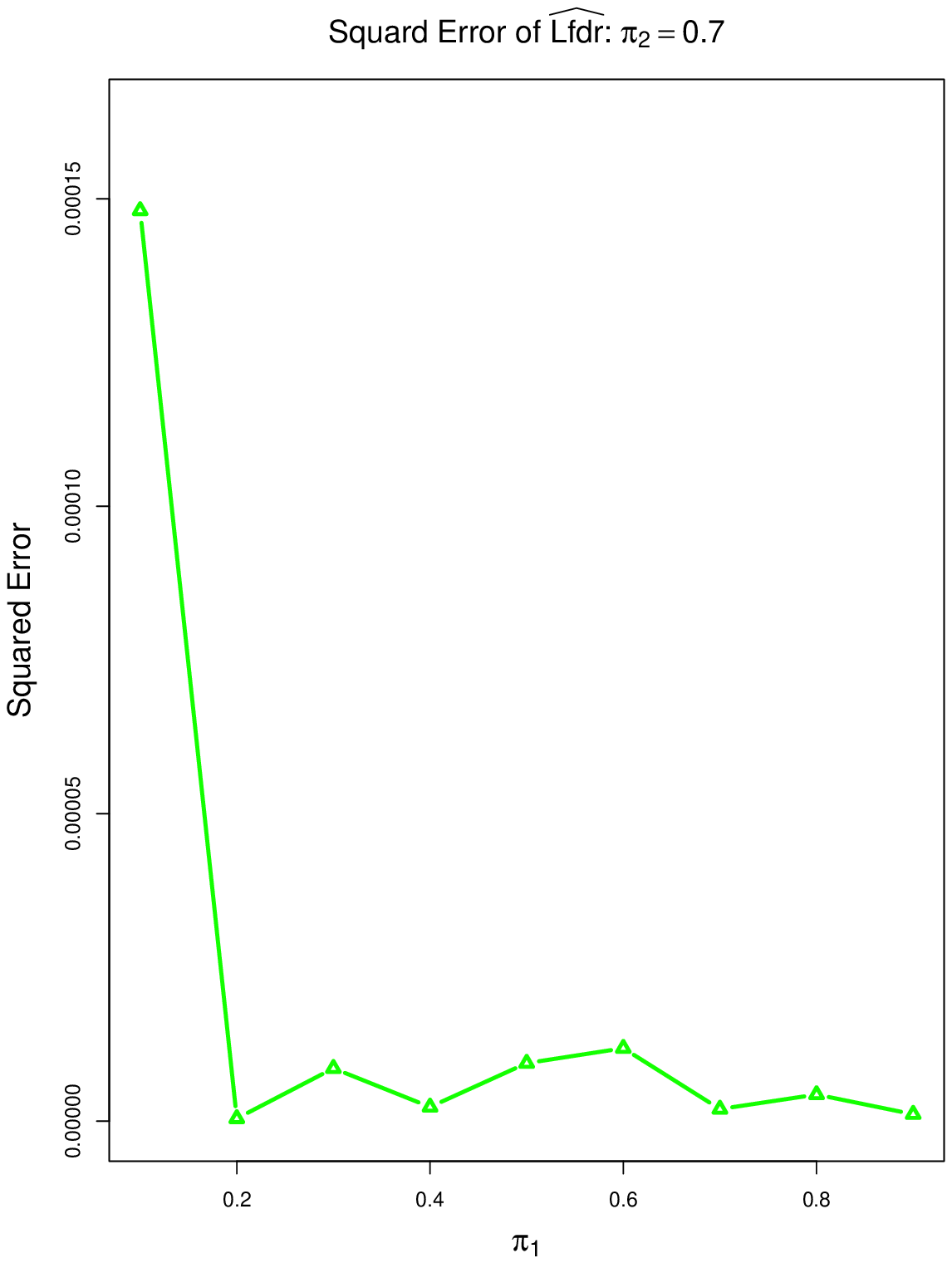}
  \caption{ The mean squared error for the estimation of the local fdr scores using the proposed method: $m=100, n=50, L=1, \pi_2=0.1, 0.3, 0.5$ and $0.7$ respectively.
  }\label{fig:gate1:lfdr:s3}
\end{figure}

\begin{figure}[H]
  \centering
  \includegraphics[height=50mm,width=50mm]{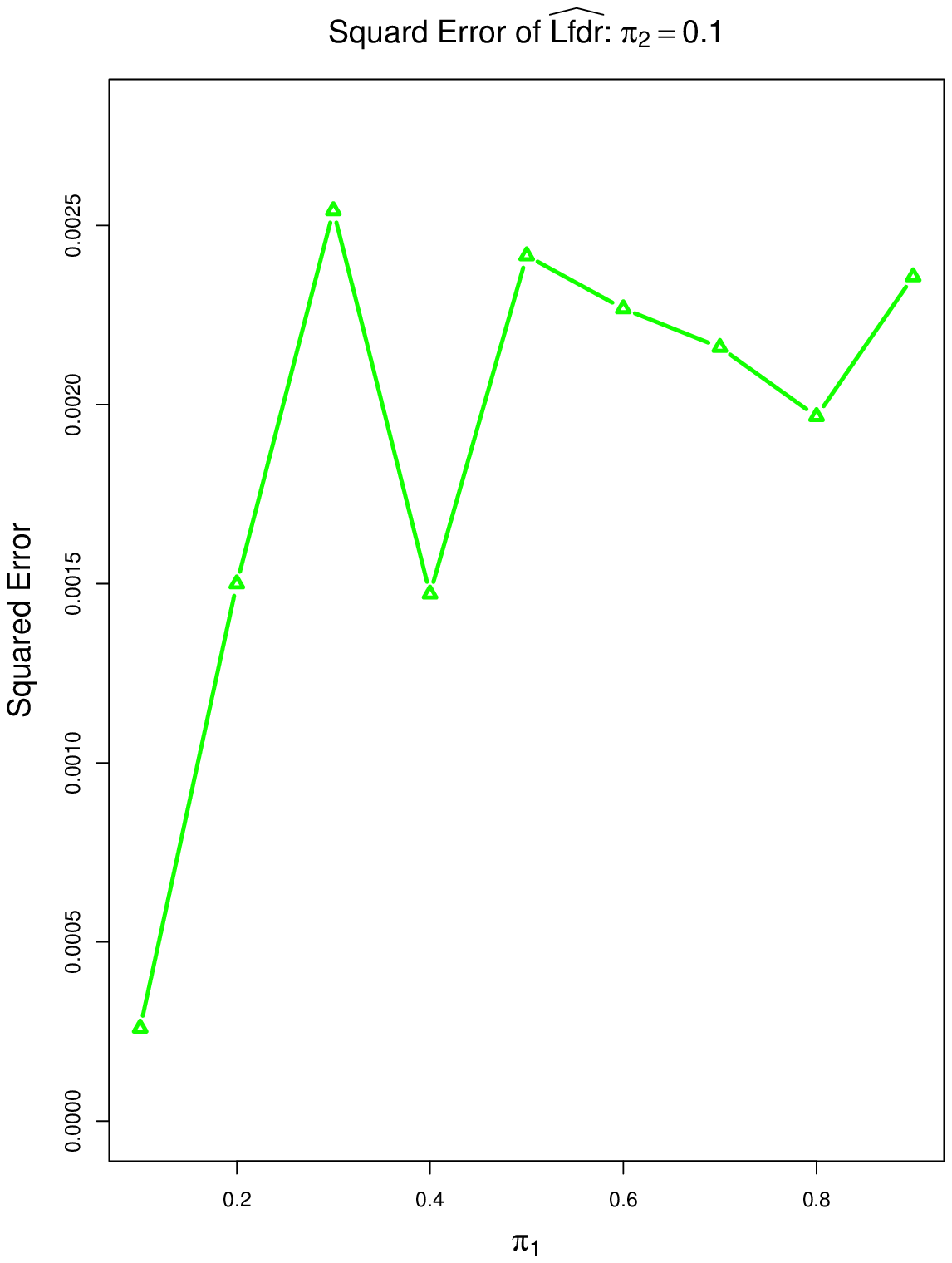}
  \includegraphics[height=50mm,width=50mm]{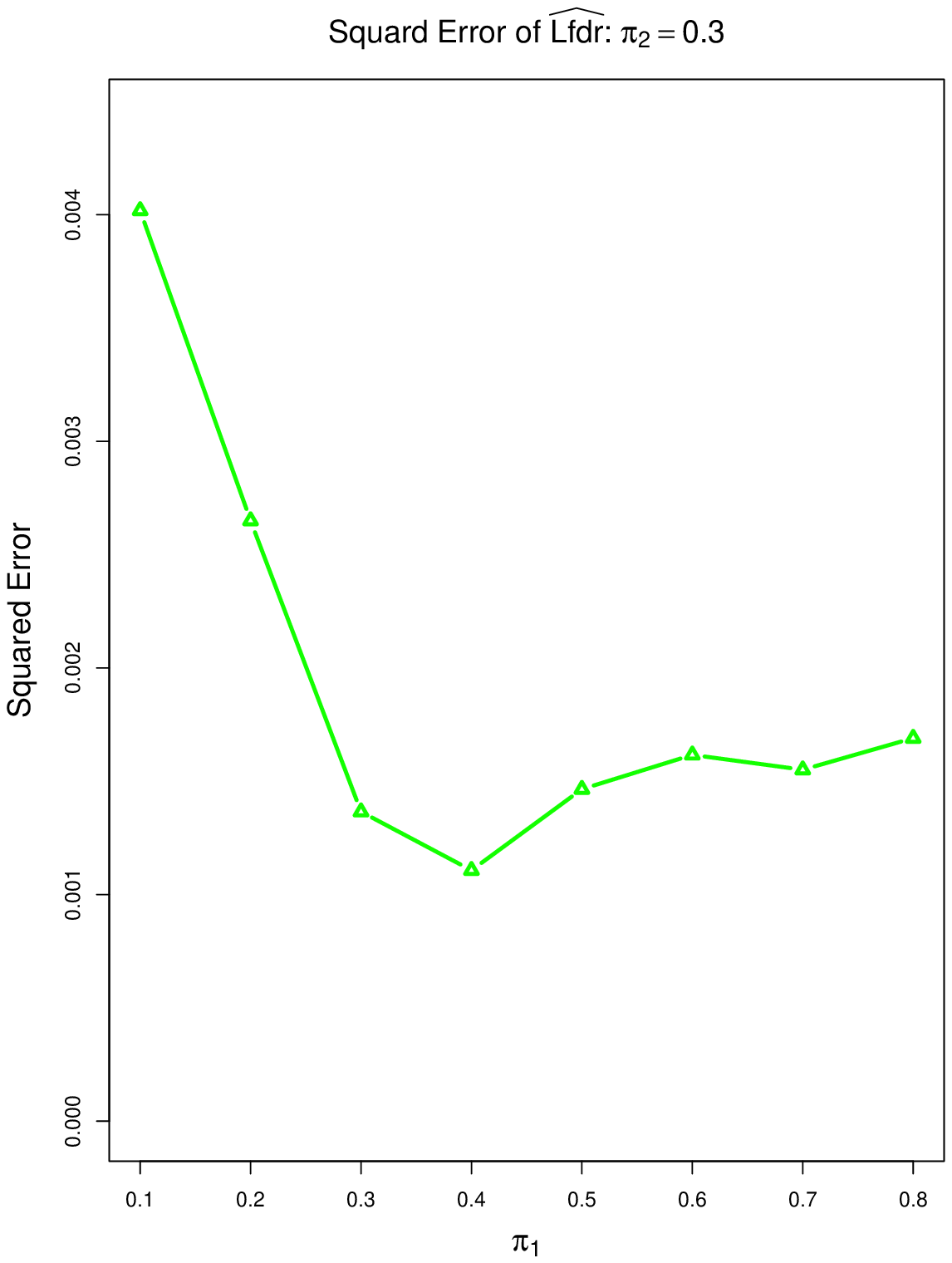}\\
  \includegraphics[height=50mm,width=50mm]{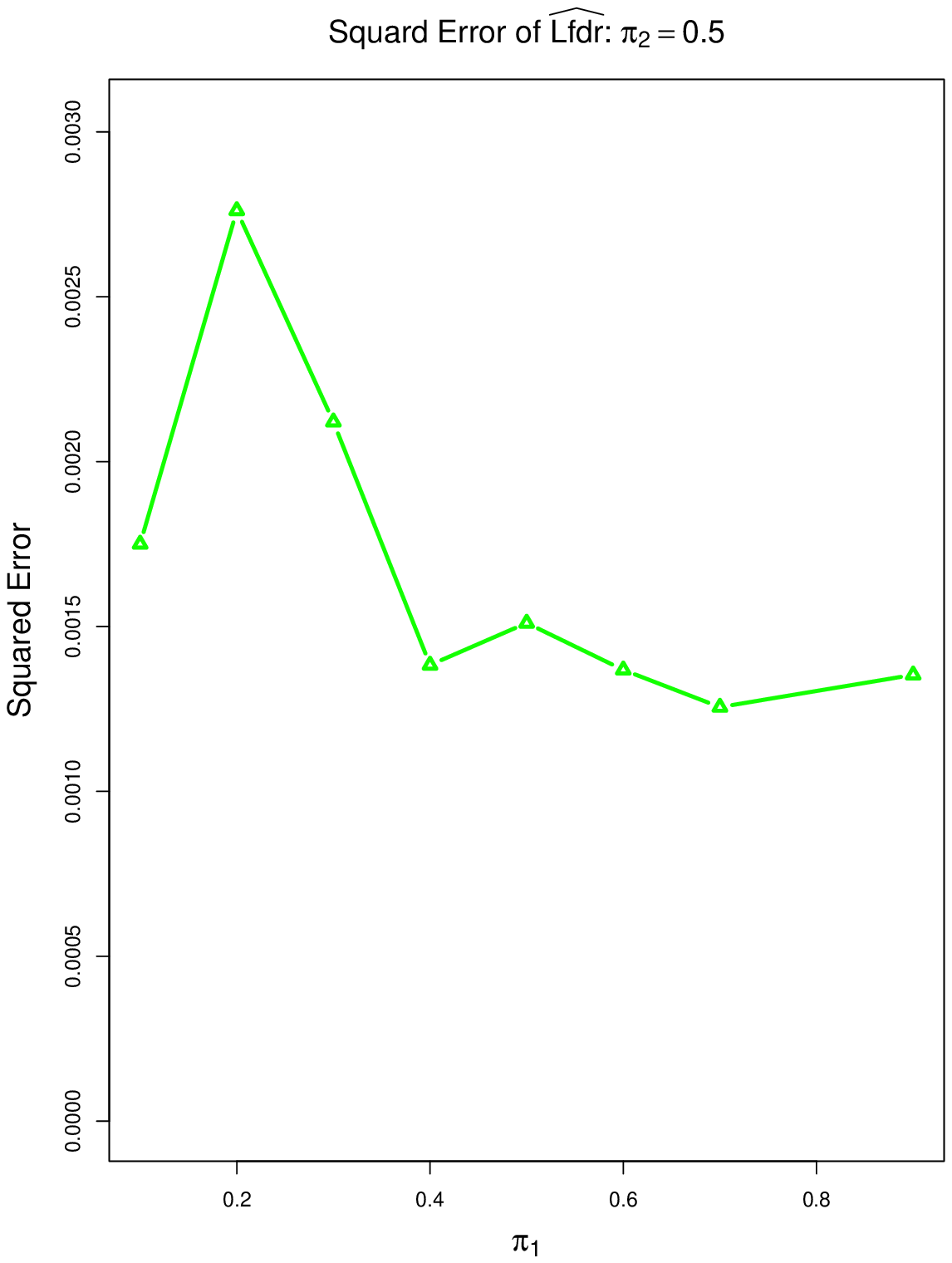}
  \includegraphics[height=50mm,width=50mm]{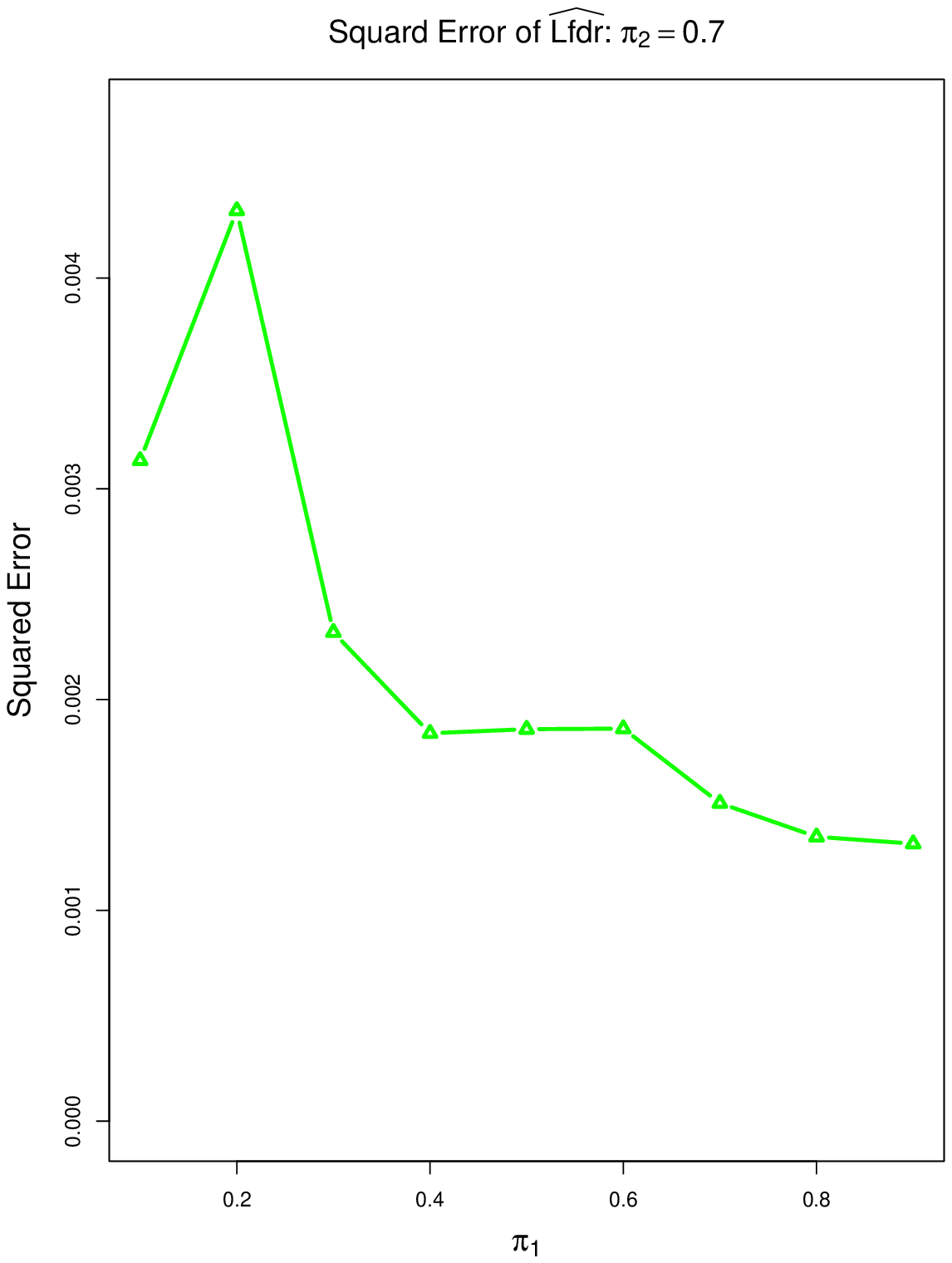}
  \caption{ The mean squared error for the estimation of the local fdr scores using the proposed method: $m=100, n=50, L=2, \pi_2=0.1, 0.3, 0.5$ and $0.7$ respectively.
  }\label{fig:gate1:lfdr:s4}
\end{figure}

\end {document}

%% file: header.tex
\usepackage{amsgen,amsmath,amstext,amsbsy,amsopn,amssymb,amscd,amsthm,lipsum}
\usepackage{color,graphicx,verbatim,multirow}%,mathtools}
\usepackage{url,caption, subcaption}

\newtheorem{theorem}{Theorem}[section]

\newtheorem{lemma}{Lemma}[section]

\newtheorem{definition}{Definition}[section]
\newtheorem{remark}{Remark}[section]

%% file: mathsymbol.tex
\newcommand{\ind}{\stackrel{\mathrm{ind}}{\sim}}

\newcommand{\vx}{\boldsymbol{x}}
\newcommand{\vX}{\boldsymbol{X}}

\newcommand{\vZ}{\boldsymbol{Z}}

\newcommand{\bbeta}{\boldsymbol{\beta}}

\newcommand{\btheta}{\boldsymbol{\theta}}